\documentclass{article}

\usepackage{algorithm}
\usepackage[noend]{algpseudocode}
\usepackage{amsmath,amssymb,amsthm}
\usepackage{bm}
\usepackage{bbm}
\usepackage{bbold}
\usepackage{booktabs}
\usepackage{caption}
\usepackage{enumerate}
\usepackage{float}
\usepackage[colorlinks,citecolor=blue,linkcolor=BrickRed]{hyperref}
\usepackage{cleveref}
\usepackage[utf8]{inputenc}
\usepackage{parskip}
\usepackage{subcaption}
\usepackage{times}
\usepackage[normalem]{ulem}
\usepackage[usenames,dvipsnames]{xcolor}
\usepackage{xspace}
\usepackage[numbers,sort&compress]{natbib}

\newtheorem{theorem}{Theorem}

\newtheorem{lemma}{Lemma}
\newtheorem{definition}{Definition}

\newcounter{note}[section]

\newcommand{\poly}{\text{poly}}

\newcommand{\mcA}{\mathcal{A}}

\newcommand{\mcC}{\mathcal{C}}

\newcommand{\mcG}{\mathcal{G}}
\newcommand{\mcH}{\mathcal{H}}

\newcommand{\mcM}{\mathcal{M}}

\newcommand{\mcP}{\mathcal{P}}

\newcommand{\mcT}{\mathcal{T}}

\allowdisplaybreaks

\usepackage{fullpage}
\usepackage{natbib}
\usepackage{thm-restate}
\usepackage{MnSymbol}
\usepackage{tcolorbox}
\usepackage{wrapfig}

\newcommand{\TBFS}{\ensuremath{T_\text{BFS}}}

\newcommand{\highV}{\text{high}}
\newcommand{\lowV}{\text{low}}
\newcommand{\intV}{\text{internal}}
\newcommand{\lca}{\text{lca}}
\newcommand{\mcPHJ}{\ensuremath{H_{\text{HJ}}}}

\title{$O(1)$ Steiner Point Removal in Series-Parallel Graphs}
\author{
    \begin{tabular}[t]{c@{\extracolsep{2em}}c} 
        D Ellis Hershkowitz & \qquad Jason Li \\
        Carnegie Mellon University & \qquad Carnegie Mellon University \\
        \small \texttt{dhershko@cs.cmu.edu} & \qquad \small \texttt{jmli@andrew.cmu.edu}
    \end{tabular}
}
\date{}

\begin{document}

\maketitle\thispagestyle{empty}

\begin{abstract}
    We study how to vertex-sparsify a graph while preserving both the graph's metric and structure. Specifically, we study the Steiner point removal (SPR) problem where we are given a weighted graph $G=(V,E,w)$ and terminal set $V' \subseteq V$ and must compute a weighted minor $G'=(V',E', w')$ of $G$ which approximates $G$'s metric on $V'$. A major open question in the area of metric embeddings is the existence of $O(1)$ multiplicative distortion SPR solutions for every (non-trivial) minor-closed family of graphs. To this end prior work has studied SPR on trees, cactus and outerplanar graphs and showed that in these graphs such a minor exists with $O(1)$ distortion.
    
    We give $O(1)$ distortion SPR solutions for series-parallel graphs, extending the frontier of this line of work. The main engine of our approach is a new metric decomposition for series-parallel graphs which we call a hammock decomposition. Roughly, a hammock decomposition is a forest-like structure that preserves certain critical parts of the metric induced by a series-parallel graph.
\end{abstract}

\newpage

{\hypersetup{linkcolor=Blue}
    \tableofcontents
}\thispagestyle{empty}\newpage

\newpage\setcounter{page}{1}
\section{Introduction}
    
    Graph sparsification and metric embeddings aim to produce compact representations of graphs that approximately preserve desirable properties of the input graph. For instance, a great deal of work has focused on how, given some input graph $G$, we can produce a simpler graph $G'$ whose metric is a good proxy for $G$'s metric; see, for example, work on tree embeddings \cite{bartal1996probabilistic,fakcharoenphol2004tight}, distance oracles \cite{thorup2005approximate,roditty2005deterministic} and graph spanners \cite{althofer1993sparse,ahmed2020graph} among many other lines of work. Simple representations of graph metrics enable faster and more space efficient algorithms, especially when the input graph is very large. For this reason these techniques are the foundation of many modern algorithms for massive graphs.
    
    Compact representations of graphs are particularly interesting when we assume that $G$, while prohibitively large and so in need of compact representations, is a member of a minor-closed graph family such as tree, cactus, series-parallel or planar graphs.\footnote{A graph $G'$ is a minor of a graph $G$ if $G'$ can be attained (up to isomorphism) from $G$ by edge contractions as well as vertex and edge deletions. A graph is $F$-minor-free if it does not have $F$ as a minor. A family of graphs $\mcG$ is said to be minor-closed if for any $G \in \mcG$ if $G'$ is a minor of $G$ then $G' \in \mcG$. A seminal work of Robertson and Seymour \cite{robertson2004graph} demonstrated that every minor-closed family of graphs is fully characterized by a finite collection of ``forbidden'' minors. In particular, if $\mcG$ is a minor-closed family then there exists a finite collection of graphs $\mcM$ where $G \in \mcG$ iff $G$ does not have any graph in $\mcM$ as a minor. Here and throughout this work we will use ``minor-closed'' to refer to all non-trivial minor-closed families of graphs; in particular, we exclude the family of all graphs which is minor-closed but trivially so.} As many algorithmic problems are significantly easier on such families---see e.g.\ \cite{hadlock1975finding,arora1998polynomial,wimer1989linear}---it is desirable that $G'$ is not only a simple approximation of $G$'s metric but that it also belongs to the same family as $G$.
    
    Steiner point removal (SPR) formalizes the problem of producing a simple $G'$ in the same graph family as $G$ that preserves $G$'s metric. In SPR we are given a weighted graph $G = (V, E, w)$ and a terminal set $V' \subseteq V$. We must return a weighted graph $G' = (V', E', w')$ where:
    \begin{enumerate}
        \item $G'$ is a minor of $G$;
        \item $ d_{G}(u,v) \leq d_{G'}(u,v) \leq \alpha \cdot d_{G}(u,v)$ for every $u, v \in V'$;
    \end{enumerate}
    and our aim is to minimize the multiplicative distortion $\alpha$ where we refer to a $G'$ with distortion $\alpha$ as an $\alpha$-SPR solution. In the above $d_G$ and $d_{G'}$ give the distances in $G$ and $G'$ respectively. 
    
    If we only required that $G'$ satisfy the second condition then we could always achieve $\alpha = 1$ by letting $G'$ be the complete graph on $V'$ where $w'(\{u, v\}) = d_G(u,v)$ for every $u,v \in V'$. However, such a $G'$ forfeits any nice structure that $G$ may have exhibited. Thus, the first condition ensures that if $G$ belongs to a minor-closed family then so does $G'$. The second condition ensures that $G'$'s metric is a good proxy for $G$'s metric. Thus, $G'$ is simpler than $G$ since it is a graph only on $V'$ while $G'$ is a proxy for $G$'s metric by approximately preserving distances on $V'$.
    
    As \citet{gupta2001steiner} observed, even for the simple case of trees we must have $\alpha > 1$. For example, consider the star graph with unit weight edges where $V'$ consists of the leaves of the star. Any tree $G' =(V',E', w')$ has at least two vertices $u$ and $v$ whose connecting path consists of at least two edges. On the other hand, the length of any edge in $G'$ is at least $2$ and so $d_{G'}(u,v) \geq 4$. Since $d_G(u,v)=2$ it follows that $\alpha \geq 2$. While this simple example rules out the possibility of a $G'$ which perfectly preserves $G$'s metric on trees, it leaves open the possibility of small distortion solutions for minor-closed families.
    
    In this vein several works have posed the existence of $O(1)$-SPR solutions for minor-closed families as an open question: see, for example, \cite{basu2008steiner,filtser2020scattering,chan2006tight,krauthgamer2014preserving,cheung2016graph} among other works. A line of work (summarized in \Cref{fig:priorWork}) has been steadily making progress on this question for the past two decades. \citet{gupta2001steiner} showed that trees (i.e. $K_3$-minor-free graphs) admit $8$-SPR solutions.\footnote{Strictly speaking trees are not the class of all $K_3$-minor-free graphs nor are they really minor-closed since they are not closed under vertex and edge deletion. Really, the class of $K_3$-minor-free graphs are all \emph{forests}. The stated results regarding trees also hold for forests but the literature seems to generally gloss over this detail and so we do so here and throughout the paper.} \citet{filtser2018relaxed} recently gave a simpler proof of this result. \citet{chan2006tight} proved this was tight by showing that $\alpha \geq 8$ for trees which remains the best known lower bound for $K_h$-minor-free graphs. In an exciting recent work, \citet{filtser2020scattering} reduced $O(1)$-SPR in $K_h$-minor-free graphs to computing ``$O(1)$ scattering partitions'' and showed how to compute these partitions for several graph classes, including cactus graphs (i.e.\ all $F$-minor-free graphs where $F$ is $K_4$ missing one edge). Lastly, a work of  \citet{basu2008steiner} generalizes these results by showing that outerplanar graphs (i.e.\ graphs which are both $K_4$ and $K_{2,3}$-minor-free) have $\alpha = O(1)$ solutions.
    
\begin{figure}
    \centering
    \includegraphics[width=\textwidth,trim=0mm 0mm 0mm 130mm, clip]{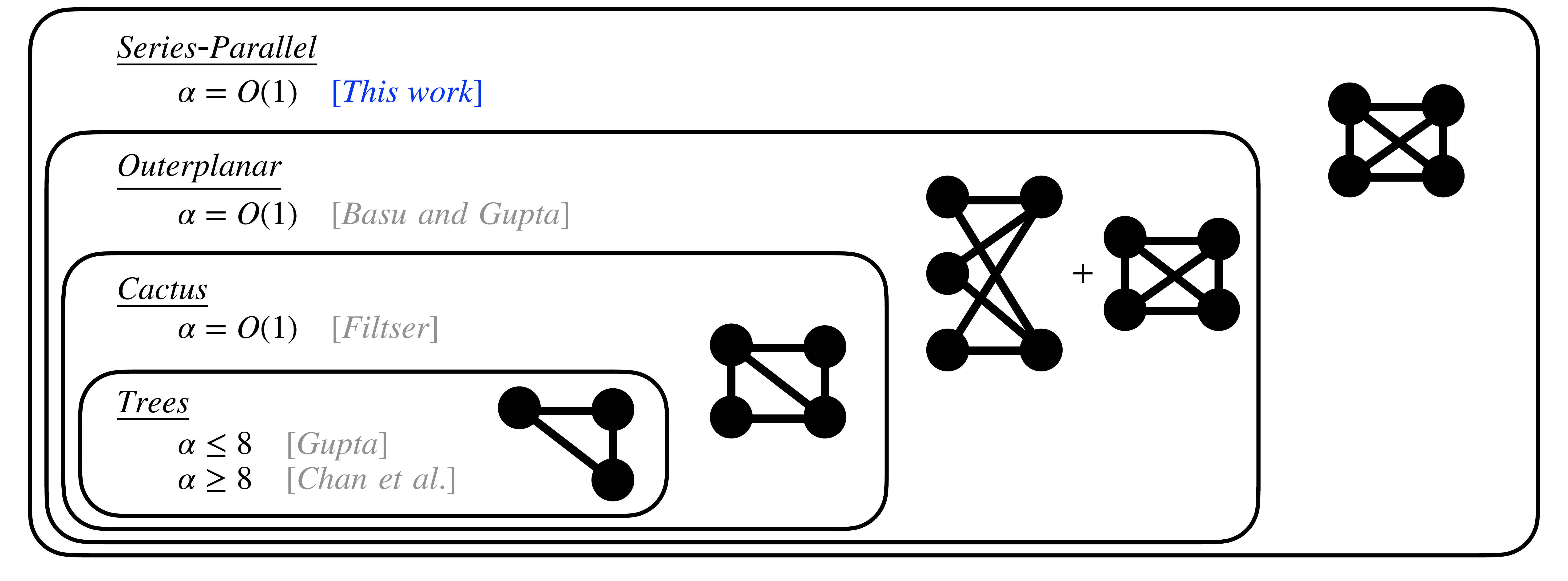}
    \caption{A summary of the SPR distortion for $K_h$-minor-free graphs achieved in prior work and our own. Graph classes illustrated according to containment. We also give the forbidden minors for each graph family.}\label{fig:priorWork}
\end{figure}    
    
\subsection{Our Contributions}

In this work, we advance the state-of-the-art for Steiner point removal in minor-closed graph families. We show that series-parallel graphs (i.e. graphs which are $K_4$-minor-free) have $O(1)$-SPR solutions. Series-parallel graphs are a strict superset of all of the aforementioned graph classes for which $O(1)$-SPR solutions were previously known; again, see \Cref{fig:priorWork}. The following theorem summarizes the main result of our paper.

\begin{restatable}{theorem}{majorThm}\label{thm:maj}
    Every series-parallel graph $G = (V,E, w)$ with terminal set $V' \subseteq V$ has a weighted minor $G' = (V', E', w')$ such that for any $u,v \in V'$ we have
    \begin{align*}
        d_G(u,v) \leq d_{G'}(u,v) \leq O(1) \cdot d_G(u,v).
    \end{align*}
    Moreover, $G'$ is poly-time computable by a deterministic algorithm.
\end{restatable}

In addition to making progress on the existence of $O(1)$-SPR solutions for every minor-closed family, our work also explicitly settles several open questions. The existence of $O(1)$-SPR solutions for series-parallel graphs was stated as an open question by both \citet{basu2008steiner} and \citet{chan2006tight}; our result answers this question in the affirmative. Furthermore, \citet{filtser2018relaxed} posed the existence of $O(1)$ scattering partitions for outerplanar and series-parallel graphs as an open question; we prove our main result by showing that series-parallel graphs admit $O(1)$ scattering partitions, settling both of these questions.

At least two aspects of our techniques may be of independent interest. First, much of our approach generalizes to any $K_h$-minor-free graph so our approach seems like a promising avenue for future work on $O(1)$-SPR in minor-closed families. Specifically, we prove our result by beginning with the ``chops'' used by \citet{klein1993excluded} to build low diameter decompositions for $K_h$-minor-free graphs. We then slightly perturb these chops to respect the shortest path structure of the graph, resulting in  what we call $O(1)$-scattering chops. The result of repeated such scattering chops is a scattering partition which by the results of \citet{filtser2020scattering} can be used to construct an $O(1)$-SPR solution. The key to this strategy is arguing that series-parallel graphs admit a certain structure---which we call a hammock decompositions---that enables one to perform these perturbations in a principled way. If one could demonstrate a similar structure for $K_h$-minor-free graphs or otherwise demonstrate the existence of $O(1)$-scattering chops for such graphs, then the techniques laid out in our work would immediately give $O(1)$-SPR solutions for all $K_h$-minor-free graphs.

Second, our hammock decompositions are a new metric decomposition for series-parallel graphs which may be interesting in their own right. We give significantly more detail in \Cref{sec:hamDec} but briefly summarize our decomposition for now. We show that for any fixed BFS tree $\TBFS$ there is a forest-like subgraph---which we call a forest of hammocks---which \emph{exactly} preserves the shortest paths between all cross edges of $\TBFS$. A hammock graph consists of two subtrees of a BFS tree and the cross edges between them. A forest of hammocks is a subgraph for which all cycles are fully contained inside one of the constituent hammocks. See \Cref{sfig:hamDecomp} for a visual preview of our decomposition. Our hammock decompositions stand in contrast to the fact that the usual way in which one embeds a graph into a tree---probabilistic tree embeddings---are known to incur distortion $\Omega(\log n)$ in series-parallel graphs \cite{gupta2004cuts}. Furthermore, our decomposition can be seen as a metric-strengthening of the classic nested ear decompositions for series-parallel graphs of \citet{khuller1989ear} and \citet{eppstein1992parallel}. In general, a nested ear decomposition need not reflect the input metric. However, not only can one almost immediately recover a nested ear decomposition from a hammock decomposition, but the output nested ear decomposition interacts with the graph's metric in a highly structured way (see \Cref{sec:metricNestedEar}). We defer a more thorough overview of our techniques to \Cref{sec:overview}

\section{Related Work}

We briefly review additional related work.

\subsection{SPR and Related Problems}

Since the introduction of SPR by \citet{gupta2001steiner}, a variety of works have studied the bounds achievable for well-behaved families of graphs for several very similar problems. \citet{krauthgamer2014preserving} studied a problem like SPR but where distances in $G$ must be exactly preserved by $G'$ and the number of Steiner vertices---that is, vertices not in $V'$---must be made as small as possible; this work showed that while $O(k^4)$ Steiner vertices suffice (where $k = |V'|$) for general graphs, better bounds are possible for well-behaved families of graphs. More generally, \citet{cheung2016graph} studied how to trade off between the number of terminals and distortion of $G'$, notably showing $(1+\epsilon)$ distortion is possible in planar graphs with $\tilde{O}(k^2/\epsilon^2)$ Steiner vertices. \citet{englert2014vertex} showed that in minor-closed graphs distances can be preserved up to $O(1)$ multiplicative distortion in expectation by a distribution over minors as opposed to preserving distances deterministically with a single minor as in SPR.

A variety of recent works have also studied how to find minors which preserve properties of $G$ other than $G$'s metric. \citet{englert2014vertex} studied a flow/cut version of SPR where the goal is for $G'$ to be a minor of $G$ just on the specified terminals while preserving the congestion of multicommodity flows between terminals: this work showed that a convex combination of planar graphs can preserve congestion on $V'$ up to a constant while for general graphs a convex combination of trees preserves congestion up to an $O(\log k)$. Similarly, \citet{krauthgamer2020refined} studied how to find minimum-size planar graphs which preserve terminal cuts. \citet{goranci2020improved} studied how to find a minor of a directed graph with as few Steiner vertices and which preserves the reachability relationships between all $k$ terminals, showing that $O(k^3)$ vertices suffices for general graphs but $O(\log k \cdot k^2)$ vertices suffices for planar graphs.

There has been considerable effort in the past few years on developing good SPR solutions for general graphs. \citet{kamma2015cutting} gave $O(\log ^ 5 k)$-SPR solutions for general graphs. This was improved by \citet{cheung2018steiner} who gave $O(\log ^ 2 k)$-SPR solutions which was, in turn, improved by \citet{filtser2018steinerNV} and \citet{filtser2018steinerRV} who gave $O(\log k)$-SPR solutions for general graphs. We also note that \citet{filtser2020scattering} also achieved similar results by way of scattering partitions, albeit with a worse poly-log factor. 

\subsection{Series-Parallel Graphs}

Series-parallel graphs are significantly more general than trees but still sufficiently structured so as to make many NP-hard problems such as maximum matchings, maximum independent sets, minimum dominating sets solvable in polynomial-time \cite{bern1987linear,kikuno1983linear,takamizawa1982linear,chebolu2012exact}. Series-parallel graphs have received a good deal of attention from the parallel algorithms community as the structure of series-parallel graphs lends itself well to parallelism; see e.g.\ \cite{bodlaender1996parallel,caspi1995edge,de1997parallel}. Particular attention has been paid to the problem of efficiently recognizing whether or not a graph is series-parallel \cite{eppstein1992parallel, he1987parallel}. Lastly, we note that seminal works of \citet{khuller1989ear} and \citet{eppstein1992parallel} also observed that series-parallel graphs admit forest-like decompositions; so-called nested ear decompositions. As mentioned above our hammock decompositions can be seen as a strengthening of the ear decomposition so as to to reflect the structure of the input metric (see \Cref{sec:metricNestedEar} for more details).

\section{Preliminaries}
Before giving an overview of our approach we summarize the characterization of series-parallel graphs we use throughout this work as well as the scattering partition framework of \citet{filtser2020scattering} on which we build.

\begin{figure}
    \centering
    \begin{subfigure}[b]{0.49\textwidth}
        \centering
        \includegraphics[width=\textwidth,trim=0mm 0mm 0mm 0mm, clip]{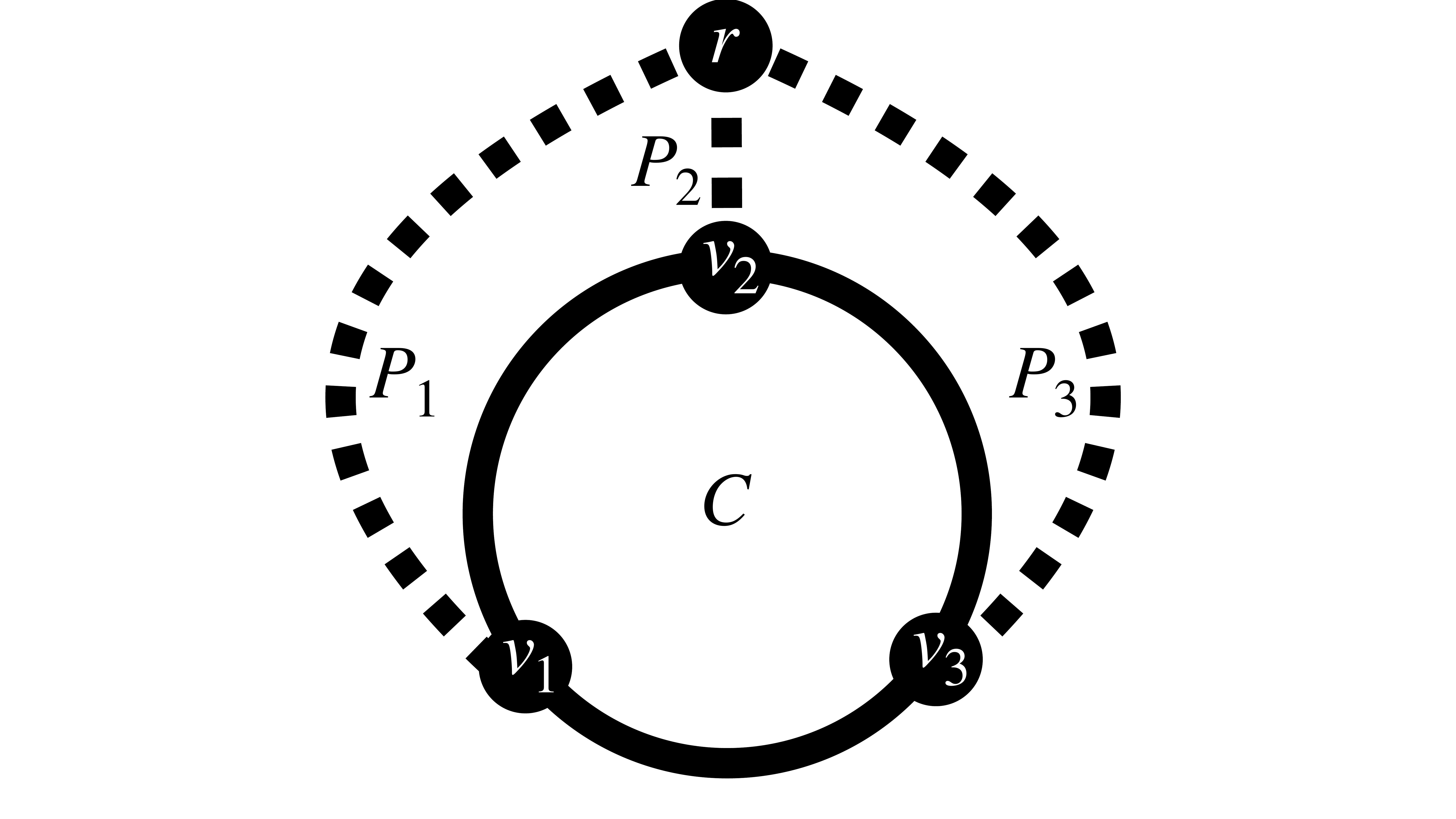}
        \caption{A clawed cycle.}\label{sfig:clawedCyc}
    \end{subfigure}
    \hfill
    \begin{subfigure}[b]{0.49\textwidth}
        \centering
        \includegraphics[width=\textwidth,trim=0mm 0mm 0mm 0mm, clip]{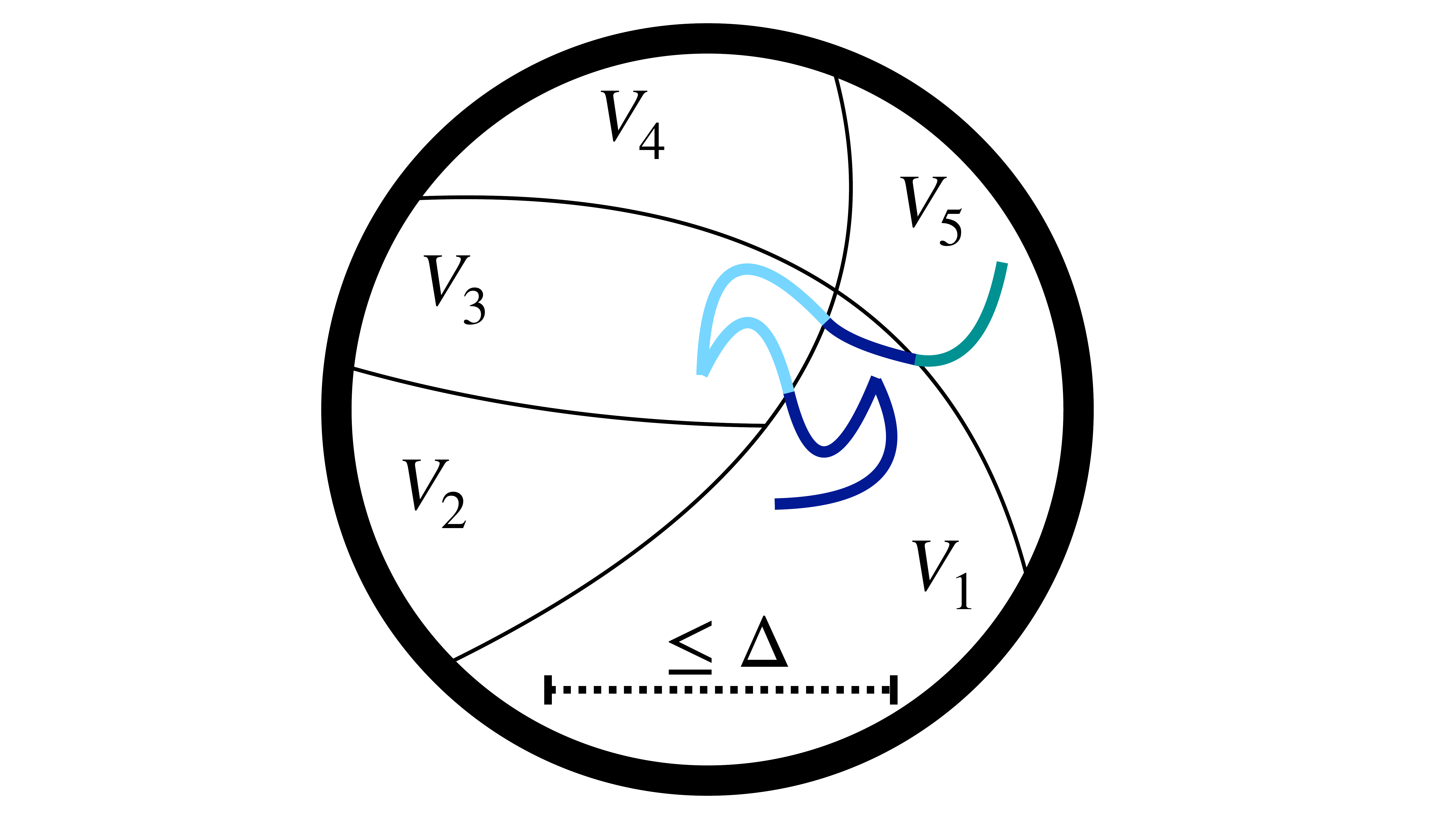}
        \caption{A scattering partition.}\label{sfig:SP}
    \end{subfigure}
    \caption{In (a) we illustrate a clawed cycle where the cycle $C$ is given in solid black and each path is given in dotted black. In (b) we illustrate a scattering partition with $\tau = 3$ and how one path $P$ of length at most $\Delta$ is incident to at most three parts where we color the subpaths of $P$ according to the incident part.}\label{fig:preLims}
\end{figure}

\subsection{Characterizations of Series-Parallel Graphs}\label{sec:characterizations}

There are some minor inconsistencies in the literature regarding what is considered a series-parallel graph and so we clarify which notion of series-parallel we use throughout this paper. Some works---e.g.\ \citet{eppstein1992parallel}---take series parallel graphs to be those which can be computed by iterating parallel and series compositions of graphs. Call these series-parallel A graphs.\footnote{The following is a definition of series-parallel A graphs due to \citet{eppstein1992parallel}. A graph is two-terminal if it has a distinct source $s$ and sink $t$. Let $G$ and $H$ be two two-terminal graphs with sources $s$ and $s'$ and sinks $t$ and $t'$. Then the series composition of $G$ and $H$ is the graph resulting from identifying $t$ and $s'$ as the same vertex. The parallel composition of $G$ and $H$ is the graph resulting from identifying $s$ and $s'$ as the same vertex and $t$ and $t'$ as the same vertex. A two-terminal series-parallel graph is a two-terminal graph which is either a single edge or the graph resulting from the series or parallel composition of two two-terminal series-parallel graphs. A graph is series-parallel A if it has some pair of vertices with respect to which it is two-terminal series-parallel.} Strictly speaking, series-parallel A graphs are not even minor-closed as they are not closed under edge or vertex deletion. Other works---e.g.\ \citet{filtser2020scattering}---take series-parallel graphs to be graphs whose biconnected\footnote{A connected component $C$ is biconnected if $C$ remains connected even after the deletion of any one vertex in $C$.} components are each series-parallel A graphs; call these series-parallel B graphs. Series-parallel B graphs clearly contain series-parallel A graphs and, moreover, are minor-closed.  For the rest of this work we will use the more expansive series-parallel B notion; henceforth we use ``series-parallel'' to mean series-parallel B.

It is well-known that a graph is $K_4$-minor-free iff it is series-parallel \cite{bodlaender1998partial}. Similarly a graph has treewidth at most $2$ iff it is series-parallel \cite{bodlaender1998partial}. In this work we will use an alternate definition in terms of ``clawed cycles'' which we illustrate in \Cref{sfig:clawedCyc}.

\begin{definition}[Clawed Cycle]\label{dfn:clawedCyc}
    A clawed cycle is a graph consisting of a root $r$, a cycle $C$ and three paths $P_1$ $P_2$ and $P_3$ from $r$ to vertices $v_1, v_2, v_3 \in C$ where $v_1 \neq v_2 \neq v_3$
\end{definition}

The fact that series-parallel graphs are exactly those that do not have any clawed cycles as a subgraph seems to be folklore; we give a proof for completeness.
\begin{lemma}
    A graph $G$ is series-parallel iff it does not contain a clawed cycle as a subgraph.
\end{lemma}
\begin{proof}
    $K_4$ is itself a clawed cycle and so a graph with no clawed cycle subgraphs is $K_4$-minor-free and therefore series-parallel. If a graph contains a clawed cycle then we can construct a $K_4$ minor by arbitrarily contracting the graph into $v_1, v_2, v_3$ and $r$, as defined in \Cref{dfn:clawedCyc}.
\end{proof}

\subsection{Scattering Partitions}

Our result will be based on a new graph partition introduced by \citet{filtser2020scattering}, the scattering partition. Roughly speaking, a scattering partition of a graph is a low-diameter partition which respects the shortest path structure of the graph; see \Cref{sfig:SP}.\footnote{We drop one of the parameters of the definition of \citet{filtser2020scattering} as it will not be necessary for our purposes.}

\begin{definition}[Scattering Partition]\label{dfn:SP}
    Given weighted graph $G = (V, E, w)$, a partition $\mcP = \{V_i\}_i$ of $V$ is a ($\tau$, $\Delta$) scattering partition if:
    \begin{enumerate}
        \item \textbf{Connected:} Each $V_i \in \mcP$ is connected;
        \item \textbf{Low Weak Diameter:} For each $V_i \in \mcP$ and $u,v \in V_i$ we have $d_G(u,v) \leq \Delta$;
        \item \textbf{Scattering:} Every shortest path $P$ in $G$ of length at most $\Delta$ satisfies $|\{V_i : V_i \cap P \neq \emptyset\}| \leq \tau$.
    \end{enumerate}
\end{definition}

\citet{filtser2020scattering} extended these partitions to the notion of a scatterable graph.
\begin{definition}[Scatterable Graph]\label{dfn:scatterableGraph}
    A weighted graph $G = (V, E, w)$ is $\tau$-scatterable if it has a $(\tau, \Delta)$-scattering partition for every $\Delta \geq 0$. 
\end{definition}
We will say that $G$ is deterministic poly-time $\tau$-scatterable if for every $\Delta \geq 0$ a $(\tau, \Delta)$-scattering partition is computable in deterministic poly-time.

Lastly, the main result of \citet{filtser2020scattering} is that solving SPR reduces to showing that every induced subgraph is scatterable. In the following $G[A]$ is the subgraph of $G$ induced by the vertex set $A$.

\begin{theorem}[\citet{filtser2020scattering}]\label{thm:SPRReduction}
    A weighted graph $G = (V, E, w)$ with terminal set $V' \subseteq V$ has an $O(\tau^3)$-SPR solution if $G[A]$ is $\tau$-scatterable for every $A \subseteq V$. Furthermore, if $G[A]$ is deterministic poly-time scatterable for every $A \subseteq V$ then the $O(\tau^3)$-SPR solution is computable in deterministic poly-time.
\end{theorem}

\section{Intuition and Overview of Techniques}\label{sec:overview}

We now give intuition and a high-level overview of our techniques. As discussed in the previous section, solving SPR with $O(1)$ distortion for any fixed graph reduces to showing that the subgraph induced by every subset of vertices is $O(1)$-scatterable. Moreover, since every subgraph of a $K_h$-minor-free graph is itself a $K_h$-minor-free graph, it follows that in order to solve SPR on any fixed $K_h$-minor-free graph, it suffices to argue that every $K_h$-minor-free graph is $O(1)$-scatterable. 

Thus, the fact that we dedicate the rest of this document to showing is as follows.
\begin{theorem}\label{thm:main}
    Every series-parallel graph $G$ is deterministic, poly-time $O(1)$-scatterable.
\end{theorem}
Combining this with \Cref{thm:SPRReduction} immediately implies \Cref{thm:maj}.

\subsection{General Approach}

Given a series-parallel graph $G$ and some $\Delta \geq 1$, our goal is to compute an $(O(1), \Delta)$-scattering partition for $G$. Such a partition has two non-trivial properties to satisfy: (1) each constituent part must have weak diameter at most $\Delta$ and (2) each path of length at most $\Delta$ must be in at most $O(1)$ parts (a property we will call ``scattering'').

A well-known technique of \citet{klein1993excluded}---henceforth ``KPR''---has proven useful in finding so-called low diameter decompositions for $K_h$-minor-free graphs and so one might reasonably expect these techniques to prove useful for finding scattering partitions. Specifically, KPR shows that computing low diameter decompositions in a $K_h$-minor-free graph can be accomplished by $O(h)$ levels of recursive ``$\Delta$-chops''. In particular, fix a root $r$ and a BFS tree $\TBFS$ rooted at $r$. Then, a $\Delta$-chop consists of the deletion of every edge with one vertex at depth $i \cdot \Delta$ and another vertex at depth $i \cdot \Delta + 1$ for every $i \in \mathbb{Z}_{\geq 1}$; that is, it consists of cutting edges between each pair of adjacent $\Delta$-width annuli. KPR proved that if one performs a $\Delta$-chop and then recurses on each of the resulting connected component then after $O(h)$ levels of recursive depth in a $K_h$-minor free graph the resulting components all have diameter at most $O(\Delta)$. We illustrate KPR on the grid graph in \Cref{fig:KPRGrid}.

\begin{figure}
    \centering
    \begin{subfigure}[b]{0.32\textwidth}
        \centering
        \includegraphics[width=\textwidth,trim=30mm 0mm 30mm 180mm, clip]{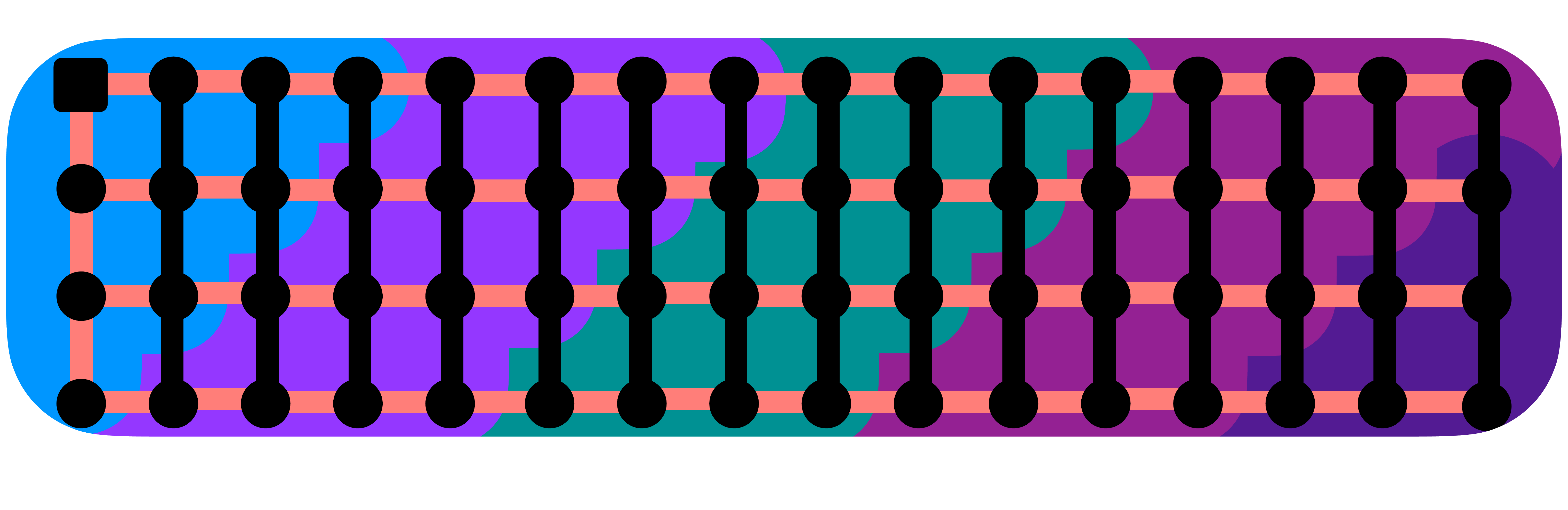}
        \caption{First $\Delta$-chop.}
    \end{subfigure}    \hfill
    \begin{subfigure}[b]{0.32\textwidth}
        \centering
        \includegraphics[width=\textwidth,trim=30mm 0mm 30mm 180mm, clip]{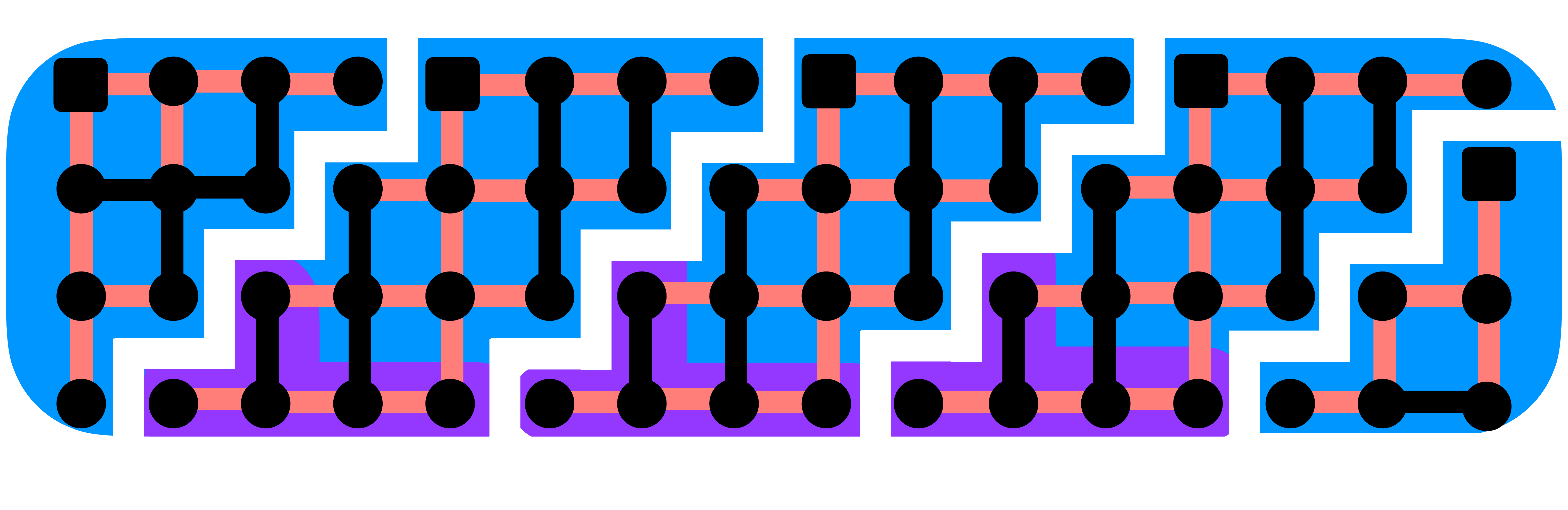}
        \caption{Second $\Delta$-chop.}
    \end{subfigure}    \hfill
    \begin{subfigure}[b]{0.32\textwidth}
        \centering
        \includegraphics[width=\textwidth,trim=30mm 0mm 30mm 180mm, clip]{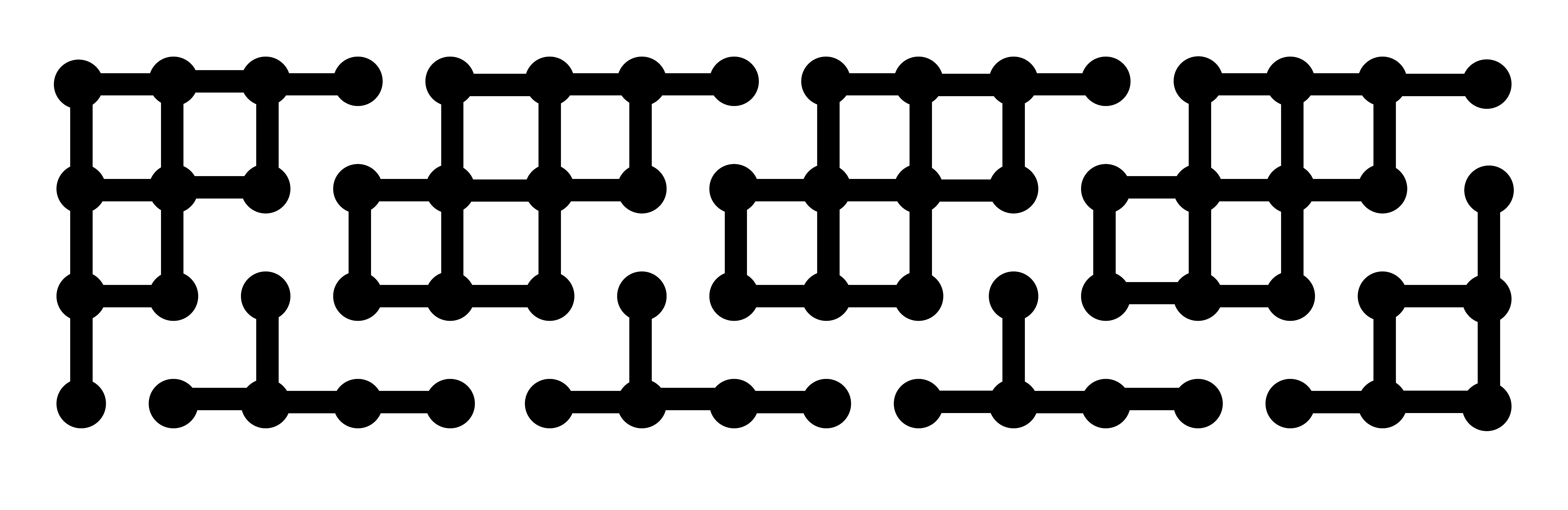}
        \caption{Resulting connected components.}
    \end{subfigure}
    \caption{Two levels of $\Delta$-chops on the grid graph for $\Delta=3$. We give the edges of the BFS trees we use in pink; roots of these trees are given as squares. Background colors give the annuli of nodes.}\label{fig:KPRGrid}
\end{figure}

Thus, we could simply apply $\Delta$-chops $O(h)$ times to satisfy our diameter constraints (up to constants) and hope that the resulting partition is also scattering. Unfortunately, it is quite easy to see that (even after just $1$ $\Delta$-chop!) a path of length at most $\Delta$ can end up in arbitrarily many parts of the resulting partition. For example, consider the example in \ref{sfig:badEG1} and \ref{sfig:badEG2} where a shortest path moves repeatedly moves back and forth between two annuli. Nonetheless, this example is suggestive of the basic approach of our work. In particular, if we merely perturbed our first $\Delta$-chop to cut ``around'' said path as in Figures \ref{sfig:badEG3} and \ref{sfig:badEG4} then we could ensure that this path ends up in a small number of partitions.

\begin{figure}
    \centering
    \begin{subfigure}[b]{0.24\textwidth}
        \centering
        \includegraphics[width=\textwidth,trim=20mm 0mm 20mm 10mm, clip]{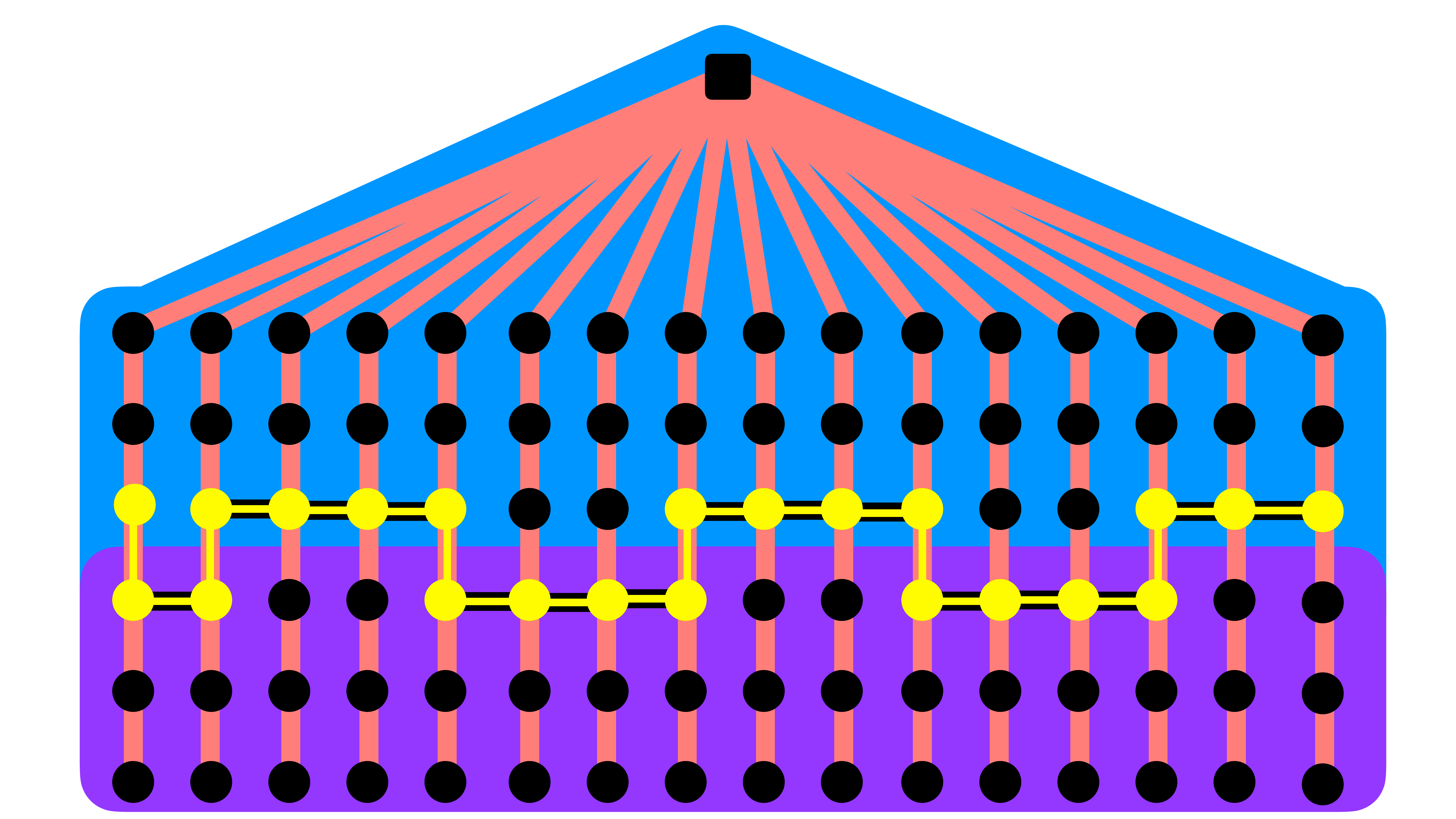}
        \caption{A $\Delta$-chop}\label{sfig:badEG1}
    \end{subfigure}    \hfill
    \begin{subfigure}[b]{0.24\textwidth}
        \centering
        \includegraphics[width=\textwidth,trim=20mm 0mm 20mm 10mm, clip]{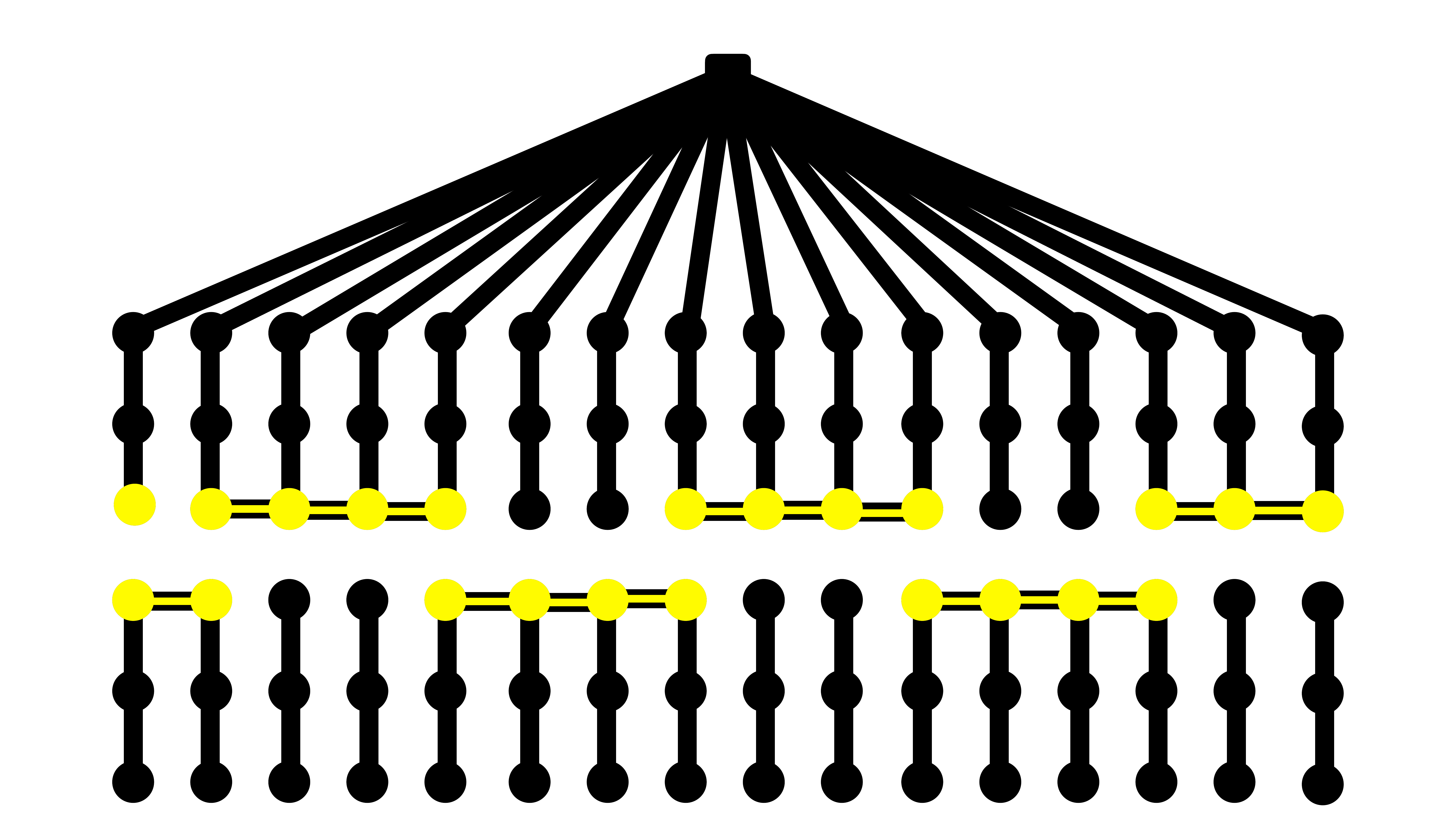}
        \caption{Components not scattering.}\label{sfig:badEG2}
    \end{subfigure}    \hfill
    \begin{subfigure}[b]{0.24\textwidth}
        \centering
        \includegraphics[width=\textwidth,trim=20mm 0mm 20mm 10mm, clip]{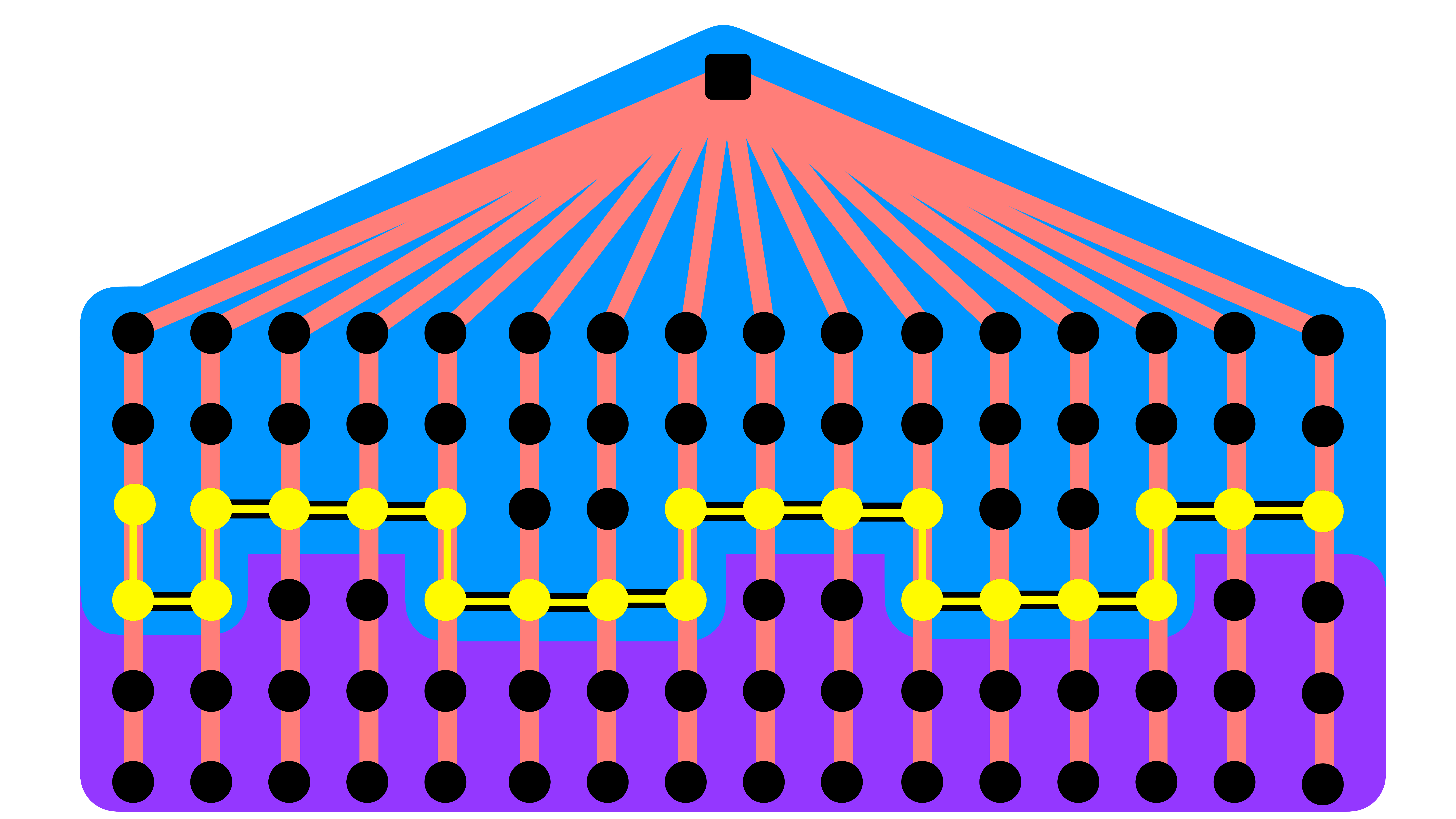}
        \caption{A perturbed $\Delta$-chop.}\label{sfig:badEG3}
    \end{subfigure}    \hfill
    \begin{subfigure}[b]{0.24\textwidth}
        \centering
        \includegraphics[width=\textwidth,trim=20mm 0mm 20mm 10mm, clip]{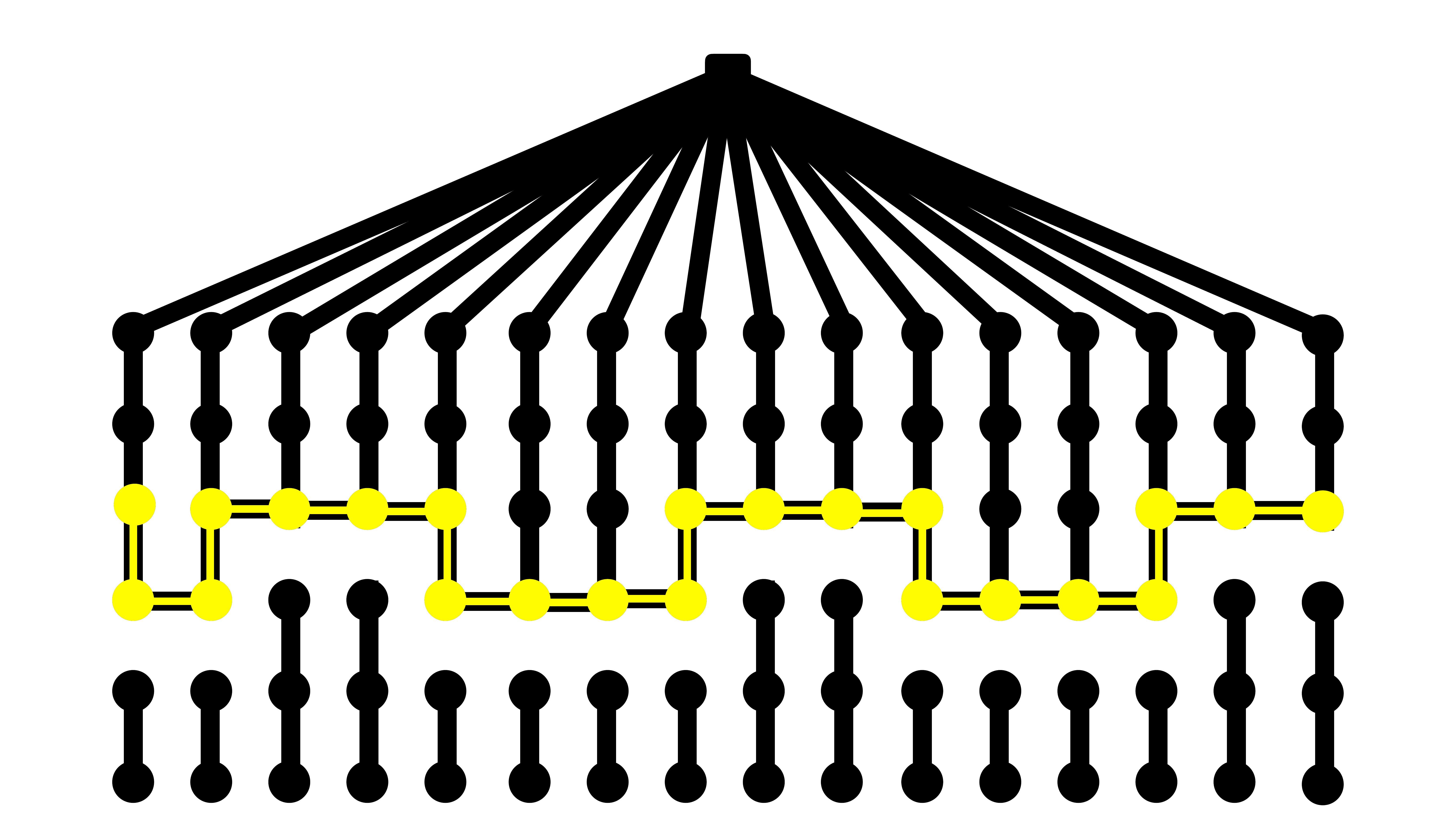}
        \caption{Components scattering.}\label{sfig:badEG4}
    \end{subfigure}
    \caption{An example where a $\Delta$-chop does not produce a scattering partition but how perturbing said chop does. Here, we imagine that the root is at the top of the graph and each edge incident to the root has length $\Delta-3$. We highlight the path $P$ that either ends up in many or one connected component depending on whether we perturb our $\Delta$-chop in yellow.}\label{fig:KPRGrid}
\end{figure}

More generally, the approach we take in this work is to start with the KPR chops but then slightly perturb these chops so that they do not cut \emph{any} shortest path of length at most $\Delta$ more than $O(1)$ times. That is, all but $O(1)$ edges of any such path will have both vertices in the same (perturbed) annuli. We then repeat this recursively on each of the resulting connected components to a constant recursion depth. Since each subpath of a shortest path of length at most $\Delta$ is itself a shortest path with length at most $\Delta$, we know that each such shortest path is broken into a constantly-many-more shortest paths at each level of recursion. Moreover, since we recurse a constant number of times, each path ends up in a constant number of components.

Implementing this strategy requires meeting two challenges. First, it is not clear that the components resulting from KPR still have low diameter if we allow ourselves to perturb our chops. Second, it is not clear how to perturb a chop so that it works \emph{simultaneously} for every shortest path. Solving the first challenge will be somewhat straightforward while solving the second will be significantly more involved. In particular, what makes the second challenge difficult is that we cannot, in general, perturb a chop on the basis of one violated shortest path as in the previous example; doing so might cause other paths to be cut too many times which will then require additional, possibly conflicting, perturbations and so on. Rather, we must somehow perturb our chops in a way that takes every shortest path into account all at once.

\subsection{Scattering Chops}

The easier issue to solve will be how to ensure that our components have low diameter even if we perturb our chops. Here, by closely tracking various constants through a known analysis of KPR we show that the components resulting from KPR with (boundedly) perturbed cuts are still low diameter.

We summarize this fact and the above discussion with the idea of a scattering chop. A $(\tau,\Delta)$-scattering chop consists of cutting all edges at \emph{about} every $\Delta$ levels in the BFS tree in such a way that no shortest path of length at most $\Delta$ is cut more than $\tau$ times. Our analysis shows that if all $K_h$-minor-free graphs admit $(O(1), \Delta)$-scattering chops for every $\Delta$ then they are also $O(1)$-scatterable and therefore also admit $O(1)$-SPR solutions; this holds even for $h > 4$.

\subsection{Hammock Decompositions and How to Use Them}

The more challenging issue we must overcome is how to perturb our chops so that every shortest path of length at most $\Delta$ is only cut $O(1)$ times. Moreover, we must do so in a way that does not perturb our boundaries by too much so as to meet the requirements of a scattering chop. We solve this issue with our new metric decomposition for series-parallel graphs, the hammock decomposition.

Consider a shortest path $P$ of length at most $\Delta$. Such a path can be partitioned into a (possibly empty) prefix consisting of only edges in $\TBFS$, a middle portion whose first and last edges are cross edges of $\TBFS$ and a (possibly empty) suffix which also only has edges in $\TBFS$. Thus, if we want to compute a scattering chop, it suffices to guarantee than any shortest path of length at most $\Delta$ which is either fully contained in $\TBFS$ or which is between two cross edges of $\TBFS$ is only cut $O(1)$ times by our chop; call the former a BFS path and the latter a cross edge path.

Next, notice that all BFS paths are only cut $O(1)$ times by our initial KPR chops. Specifically, each BFS path can be partitioned into a subpath which goes ``up'' in the BFS tree and a subpath which goes ``down'' in the BFS tree. As our initial KPR chops are $\Delta$ apart and each such subpath is of length at most $\Delta$, each such subpath is cut at most $O(1)$ times. Thus provided our perturbations do not interfere \emph{too much} with the initial structure of our KPR chops we should expect that our BFS paths will only be cut $O(1)$ times.





Thus, our goal will be to perturb our KPR chops to not cut any cross edge path more than $O(1)$ times while mostly preserving the initial structure of our KPR chops. Our hammock decompositions will allow us to do exactly this. They will have two key components. 


The first part is a ``forest of hammocks.'' Suppose for a moment that our input graph had a forest subgraph $F$ that contained all cross edge paths of our graph which were also shortest paths. Then, it is not too hard to see how to use $F$ to perturb our chops to be scattering for all cross edge paths. In particular, for each edge $\{u,v\} \in F$ where $u$ is in some annuli $A$ but $v$ is in some other annuli $A'$ (before any perturbation) then we propagate $A$ an additional $\Theta(\Delta)$ deeper into $F$; that is if we imagine that $v$ is the child of $u$ in $F$ then we move all descendants of $u$ in $F$ within $\Theta(\Delta)$ of $u$ into $A$. A simple amortized analysis shows that after performing these perturbations every cross edge path is cut $O(1)$ times. Unfortunately, it is relatively easy to see that such an $F$ may not exist in a series-parallel graph. The forest of hammocks component of our decompositions is a subgraph which will be ``close enough'' to such an $F$, thereby allowing us to perturb our chops similarly to the above-mentioned way. As mentioned in the introduction, a hammock graph consists of two subtrees of a BFS tree and the cross edges between them. A forest of hammocks is a graph partitioned into hammocks where every cycle is fully contained in one of the constituent hammocks. While the above perturbation will guarantee that our cross edge paths are not cut too often, it is not clear that such a perturbation does not change the structure of our initial chops in a way that causes our BFS paths to be cut too many times.



The second part of our hammock decompositions is what we use to guarantee that our BFS paths are not cut too many times by preserving the structure of our initial KPR chops. Specifically, the forest structure of our hammocks will reflect the structure of $\TBFS$. In particular, we can naturally associate each hammock $H_i$ with a single vertex, namely the lca of any $u$ and $v$ where $u$ is in one tree of $H_i$ and $v$ is in the other. Then, our forest of hammocks will satisfy the property that if hammock $H_i$ is a ``parent'' of hammock $H_j$ in our forest of hammocks then the lca corresponding to $H_i$ is an ancestor of the lca corresponding to $H_j$ in $\TBFS$; even stronger, the lca of $H_j$ will be contained in $H_i$. Roughly, the fact that our forest of hammocks mimics the structure of $\TBFS$ in this way will allow us to argue that the above perturbation does not alter the initial structure of our KPR chops too much, thereby ensuring that BFS paths are not cut too many times.

 


The computation of our hammock decompositions constitutes the bulk of our technical work but is somewhat involved. The basic idea is as follows. We will partition all cross edges into equivalence classes where each cross edge in an equivalence class share an lca in $\TBFS$ (though there may be multiple, distinct equivalence classes with the same lca). Each such equivalence class will eventually correspond to one hammock in our forest of hammocks. To compute our forest of hammocks we first connect up all cross edges in the same equivalence class. Next we connect our equivalence classes to one another by cross edge paths which run between them. We then extend our hammocks along paths towards their lcas to ensure the above-mentioned lca properties. Finally, we add so far unassigned subtrees of $\TBFS$ to our hammocks. We will argue that when this process fails it shows the existence of a $K_4$-minor and, in particular, a clawed cycle.




\section{Notation and Conventions}\label{sec:notationAssumptions}

Before proceeding ot our formal results we specify the notation we use throughout this work as well as some of the assumptions we make on our input series-parallel graph without loss of generality (WLOG).

\textbf{Graphs:} For a weighted graph $G = (V, E, w)$, we let $V(G) = V$ and $E(G) = E$ give the vertex set and edge sets of $G$ respectively. We will sometimes abuse notation and use $G$ to stand for $V$ or $E$ when it is clear from context if we mean $G$'s vertex or edge set. $w : E \to \mathbb{Z}_{\geq 1}$ will give the weights. Given graphs $G$ and $H$, we will use the notation $H \subseteq G$ to indicate that $H$ is a subgraph of $G$. The weak diameter of a subgraph $H$ is $\max_{u,v \in V(H)} d_G(u,v)$.

\textbf{Assumption of Unique Shortest Paths and Unit Weights:} We will assume throughout this work that in our input series-parallel graph for any vertices $u$ and $v$ the shortest path between $u$ and $v$ is unique. This is without loss of generality: one can randomly perturb the initial weights of the input graph so as to guarantee this condition at an arbitrarily small cost in the quality of the resulting SPR solution. Similarly, in this work we will assume without loss of generality that $w(e) = \poly(n)$ for every $e$ for our input series-parallel graph where $n = |V|$. This assumption allows us to simplify presentation by assuming that our edges all have unit weights. In particular, it allows us to assume that $w(e) = 1$ for every $e$ while only increasing the number of vertices by a polynomial factor since we can expand an edge of weight $w(e)$ into a path of $w(e)$ edges while preserving series-parallelness and the metric. For this reason, we will assume that our input series-parallel graph has unique shortest paths and unit weight edges henceforth, but, again, none of our results rely on these assumptions.

\textbf{Induced Graphs and Edges:} Given an edge set $E$ and disjoint vertex sets $V_1$ and $V_2$, we let $E(V_1, V_2) := \{e = \{v_1, v_2\} \in E : v_1 \in V_1, v_2 \in V_2\}$ be all edges between $V_1$ and $V_2$. Given graph $G = (V,E)$ and a vertex set $U \subseteq V$, we let $G[U] = (U, E_U)$ be the ``induced subgraph'' of $G$ where $\{u', v'\} \in E_U$ iff $\{u', v'\} \in E$. Given a collection of subgraphs $\mcH = \{H_i\}_i$ of a graph we call $G[\mcH] := (\bigcup_i V(H_i), \bigcup_i E(H_i))$ the induced subgraph of $\mcH$. Similarly, we will let $E(\mcH) := \bigcup_i E(H_i)$ give the edges of $\mcH$.

\textbf{Paths:} Given a path $P = (v_0, v_1, \ldots, v_k, v_{k+1})$ we will use $\intV(P) := \{v_1, \ldots, v_k\}$ to refer to the internal vertices of $P$. We will say that a path $P$ is between two vertex sets $U$ and $W$ if its first and last vertices are in $U$ and $W$ respectively and $\intV(P) \cap U = \emptyset$ and $\intV(P) \cap W = \emptyset$. We will sometimes abuse notation and use $P$ and $E(P)$ interchangeably. We will also sometimes say such a path is ``from'' $U$ to $W$ interchangeably with a path is ``between'' $U$ and $W$. We will use $P \oplus P'$ to refer to the concatenation of two paths which share an endpoint throughout this paper. For a tree $T$, we will let $T(u, v)$ stand for the unique path between $u$ and $v$ in $T$ for $u,v \in V(T)$. We will sometimes assume that a path from a vertex set to another vertex set is directed in the natural way.

\textbf{BFS Tree Notation:} For much of this work we will fix a a series-parallel graph $G= (V,E)$ along with a fixed but arbitrary root $r\in V$ and a fixed but arbitrary BFS tree $\TBFS$ with respect to $r$. When we do so we will let $E_c := E \setminus E(\TBFS)$ be all cross edges of $\TBFS$. For $u,v \in V$, if $u \in \TBFS(r,v) \setminus \{v\}$ then we say that $u$ is an ancestor of $v$. In this case, we also say that $v$ is a descendant of $u$. If $u$ is an ancestor of $v$ or $v$ is an ancestor of $u$ then we say that $u$ and $v$ are related; otherwise, we say that $u$ and $v$ are unrelated. For two vertices $u, v \in V$ we will use the notation $u \prec v$ to indicate that $v$ is an ancestor of $u$ and we will use the notation $u \preceq v$ to indicate that $v$ is an ancestor of or equal to $u$. It is easy to verify that $\preceq$ induces a partial order. We let $\TBFS(v) := \TBFS[\{v\} \cup \{u \in V : \text{$u$ is a descendant of $v$}\}]$ be the subtree of $\TBFS$ rooted at $v$. Given a connected subgraph $T \subseteq \TBFS$, we will let $\highV(T)$ be the vertex in $V(T)$ which is an ancestor of all vertices in $V(T)$. Given a path $P \subseteq \TBFS$ we will say that $P$ is monotone if $\highV(P)$ is an ancestor of all vertices in $P$ and there is some vertex $\lowV(P)$ which is a descendant of all vertices in $P$. We let $h(v)$ give the height of a vertex in $\TBFS$ (where we imagine that the nodes furthest from $r$ are at height $0$). We let $\lca(e)$ be the least common ancestor of $u$ and $v$ in $\TBFS$ for each $e = \{u,v\} \in E$.

\textbf{Miscellaneous:} We will use $\sqcup$ for disjoint set union throughout this paper. That is $A \sqcup B$ is equal to $A \cup B$ but indicates that $A \cap B = \emptyset$.


\section{Perturbing KPR and Scattering Chops}
In this section we show that KPR still gives low diameter components even if its boundaries are perturbed and  therefore somewhat ``fuzzy.'' We then observe that this fact shows that ``$O(1)$-scattering chops'' imply the existence of $O(1)$-scattering partitions for $K_h$-minor-free graphs and therefore $O(1)$-SPR solutions.

\subsection{Perturbing KPR}
We will repeatedly take the connected components of annuli with ``fuzzy'' boundaries. We formalize this with the idea of a $c$-fuzzy $\Delta$-chop; see \Cref{sfig:sChop1} for an illustration.


\begin{definition}[$c$-Fuzzy $\Delta$-Chop]
   Let $G=(V, E, w)$ be a weighted graph with root $r$ and fix $0 \leq c < 1$ and $\Delta \geq 1$. Then a $c$-fuzzy $\Delta$-chop is a partition of $V$ into ``fuzzy annuli'' $\mcA = \{A_1, A_2, \ldots\}$ where for every $i$ and $v \in A_i$ we have
    \begin{align*}
    (i-1)\Delta - \frac{c \Delta}{2} \leq d(r, v) < i \cdot \Delta + \frac{c \Delta}{2} .
    \end{align*}
\end{definition}

As each fuzzy annuli in a fuzzy chop may contain many connected components we must be careful when specifying how recursive application of these chops break a graph into connected components; hence the following definitions. Given fuzzy annuli $A_i$, we let $\mcC_i$ be the connected components of $A_i$.
\begin{definition}[Components Resulting from a $c$-Fuzzy $\Delta$-Chop]
    Let $G=(V, E, w)$ be a weighted graph and let $\mcC$ be a partition of $V$ into connected components. Then we say that $\mcC$ results from one level of $c$-fuzzy $\Delta$-chops if there is a $c$-fuzzy $\Delta$-chop $\mcA$ with respect to some root $r \in V$ satisfying $\mcC = \bigcup_{i : A_i \in \mcA} \mcC_i$. Similarly, for $h \geq 2$ we say that $\mcC$ results from $h$-levels of $c$-fuzzy $\Delta$-chops if there is some $\mcC'$ which results from one level of $c$-fuzzy $\Delta$-chops and $\mcC$ is the union of the result of $h-1$ levels of $c$-fuzzy $\Delta$-chops on each $C' \in \mcC'$.
\end{definition}



We will now claim that taking $h-1$ levels of fuzzy chops in a $K_h$-minor-free graph will result in a connected, low weak diameter partition. In particular, we show the following lemma, the main result of this section.

\begin{restatable}{lemma}{fuzzyChop}\label{lem:fuzChop}
    Let $\Delta$ and $h$ satisfy $2 \leq h$, $\Delta \geq 1$ and fix constant $0 \leq c <1$. Suppose $\mcC$ is the result of $h-1$ levels of $c$-fuzzy $\Delta$-chops in a $K_h$-minor-free weighted graph $G$. Then, the weak diameter of every $C \in \mcC$ is at most $O(h \cdot \Delta)$.
\end{restatable}

For the rest of this section we identify the nodes of a minor of graph $G$ with ``supernodes.'' In particular, we will think of each of the vertices of the minor as corresponding to a disjoint, connected subset of vertices in $G$ (a supernode) where the minor can be formed from $G$ (up to isomorphism) by contracting the constituent nodes of each such supernodes.

Our proof will closely track a known analysis of KPR \cite{leeKPR}. The sketch of this strategy is as follows. We will argue that if we fail to produce parts with low diameter then we have found $K_h$ as a minor. Our proof will be by induction on the number of levels of fuzzy chops. Suppose $\mcC$ is produced by $h-1$ levels of fuzzy chops; in particular, suppose $\mcC$ is produced by taking some fuzzy chop to get $\mcC'$ and then taking $h-2$ levels of fuzzy chops on each $C' \in \mcC'$. Also assume that there is some $C \in \mcC$ which has \emph{large} diameter. Then, $C$ must result from taking $h-2$ levels of fuzzy chops on some $C' \in \mcC$ where $C'$ lies in some fuzzy annuli $A_i$ of $G$. By our inductive hypothesis it follows that $C$ contains $K_{h-1}$ as a minor. Our goal is to add one more supernode to this minor to get a $K_h$ minor. We will do so by finding disjoint paths of length $O(\Delta)$ in the annuli above $A_i$ from each of the $K_{h-1}$ supernodes all of which converge on a single connected component. By adding these paths to their respective supernodes in the $K_{h-1}$ minor and adding the connected component on which these paths converge to our collection of supernodes, we will end up with a $K_h$ minor.

The main challenge in this strategy is to show how to find paths as above which are disjoint. We will do so by choosing these paths from a ``representative'' from each supernode where initially the representatives are $\Omega(h\Delta)$ far-apart and grow at most $O(\Delta)$ closer at each level of chops; since we do at most $O(h)$ levels of chops, the paths we choose will never intersect.

To formalize this strategy we must state a few definitions which will aid in arguing that these representatives are far apart.

\begin{definition}[$\Delta$-Dense]
    Given sets $S,U \subseteq V$ we say $S$ is $\Delta$-dense in $U$ if $d(u, S) \leq \Delta$ for every $u \in U$.
\end{definition}

\begin{definition}[$R$-Represented]
    A $K_h$ minor is $R$-represented by set $S \subseteq V$ if each supernode of the minor in $G$ contains a representative in $S$ and these representatives are pairwise at least $R$ apart in $G$.
\end{definition}

Since $V$ is clearly $(1+c)\Delta$-dense in $V$, we can set $S$ to $V$ and $j$ to $h-1$ in the following lemma to get \Cref{lem:fuzChop}.
\begin{lemma}
    Fix $0 \leq c < 1$ and $h > j \geq 1$. Let $S$ be any set which is $(1+c)\Delta$-dense in $V$ and suppose $\mcC$ is the result of $j$ levels of $c$-fuzzy $\Delta$-chops and some $C \in \mcC$ has weak diameter more than $22h \Delta$. Then there exists a $K_{j+1}$ minor which is $8(h-j)\Delta$-represented by $S$.
\end{lemma}
\begin{proof}

    Our proof is by induction on $j$. The base case of $j=0$ is trivial as $K_1$ is a minor of any graph with a supernode $+\infty$-represented by any single vertex in $V$.
    
    Now consider the inductive step on graph $G = (V, E)$. Fix some set $S$ which is $(1+c)\Delta$-dense in $V$ and let $\mcC$ be the result of $j$ levels of $c$-fuzzy chops using root $r$ with some $C' \in \mcC$ of diameter more than $22h\Delta$. Suppose $C'$ is in fuzzy annuli $A_k$ and suppose that $C'$ is the result of applying $j-1$ levels of $c$-fuzzy chops to some $C$ which resulted from 1 level of $c$-fuzzy chops in $G$; note that $C$ is a connected component of $A_k$ and that $C'$ is contained in $C$.
    
     As an inductive hypothesis we suppose that any $j-1$ levels of $c$-fuzzy $\Delta$-chops on any graph $H$ which results in a cluster of weak diameter more than $22 h \Delta $ demonstrates the existence of a $K_j$ minor in $H$ which is $8(h-j - 1)\Delta$-represented by any set $S'$ which is $(1+c)\Delta$-dense in $V(H)$. Here weak diameter is with respect to the distances induced by the original input graph.
    
     Thus, by our inductive hypothesis we therefore know that $C$ contains a $K_{j}$ minor which is $8(h-j - 1)\Delta$-represented by any $S' \subseteq V(C)$ which is $(1+c)\Delta$-dense in $V(C)$. In particular, we may let $S'$ be the ``upper boundary'' of $C$; that is, we let $S'$ be all vertices $v$ in $C$ such that the shortest path from $v$ to $r$ does not contain any vertices in $C$.
     Clearly the shortest path from any vertex in $C$ to $r$ intersects a node in $S'$; moreover, when restricted to $C$ this shortest path has length at most $\Delta + c\Delta$ (since $C$ is contained in $A_k$) which is to say that $S'$ is $(1+c)\Delta$-dense in $C$. Thus, by our induction hypothesis there is a $K_j$ minor in $C$ which is $8(h-j-1)\Delta$-represented by $S'$. Let $V_1, \ldots, V_j$ be the nodes in the supernodes of the $K_j$ minor. 
    
    
    We now describe how to extend the $K_j$ minor to a $K_{j+1}$ minor which is $8(h-j)\Delta$-represented by $S$. We may assume that $k \geq 9h + 1$; otherwise the distance from every node in $A_k$ to $r$ would be at most $(9h+1) \Delta + \frac{c\Delta}{2} \leq (9h+\frac{3}{2})\Delta$ and so the weak diameter of $C'$ would be at most $(18h+3)\Delta \leq 21h\Delta$, contradicting our assumption on $C'$'s diameter. It follows that for every $v \in A_k$ we have
    \begin{align}\label{eq:distLB}
    d(v, r) \geq (k-1)\Delta - \frac{c\Delta}{2} \geq 9h\Delta - \frac{c\Delta}{2} \geq 8h\Delta.
    \end{align}
    
    We first describe how we grow each supernode $V_i$ from the $K_j$ minor to a new supernode $V_i'$. Let $v_i$ be the representative in $S'$ for $V_i$. Consider the path which consists of following the shortest path from $v_i$ to $r$ for distance $2\Delta$ and then continuing on to the nearest node in $S$; let $v_i'$ be this nearest node; this path from $v_i$ to $v_i'$ has length at most $(3+c)\Delta$ since $S$ is $(1+c)\Delta$-dense in $V(G)$. Let $V_i'$ be the union of $V_i$ with the vertices in this path. Since each of these paths is of length at most $(3+c)\Delta \leq 4 \Delta$, it follows that each of these paths for each $i$ must be disjoint since each $v_i$ is at least $8(h-j+1)\Delta > 8 \Delta$ apart. Further, every $v_i'$ must also, therefore, be at least $8(h-j)\Delta$ apart. Therefore, we let these $v_i'$ form the representatives in $S$ for each of the $V_i'$.
    
    We now describe how we construct the additional supernode, $V_0$, which we add to our minor to get a $K_{j+1}$ minor. $V_0$ will ``grow'' from the root to $S$ and each of the $V_i'$. In particular, let $u_i \in V_i'$ be the node in $V_i'$ which is closest to $r$ and let $P_i$ be the shortest path from $r$ to $u_i$, excluding $u_i$. Similarly, let $v_0'$ be the node in $S$ closest to $r$ and let $P_0$ be the shortest path from $r$ to $v_0'$, including $v_0'$. Then, we let $V_0$ be $P_0 \cup P_1 \cup \ldots \cup P_j$ and we let $v_0'$ be the representative for $V_0$ in $S$. We claim that for every $i \geq 1$ we have 
    \begin{align}\label{eq:PNot}
    d(P_0, V_i') \geq 8(h-j)\Delta.
    \end{align}
    In particular, notice that since $S$ is $(1+c)\Delta$-dense in $V(G)$ we know that $d(r, v_0') \leq (1+c)\Delta \leq 2 \Delta$ and since $d(v, r) \geq 8h\Delta$ for every $v \in A_k$ by \Cref{eq:distLB} and $d(v_i', A_k) \leq (3+c) \Delta \leq 4 \Delta$, it follows that $d(P_0, V_i')  \geq (8h - 6) \Delta \geq 8(h-j)\Delta$. Consequently, $d(v_0', v_i') \geq 8(h-j)\Delta$ for every $i \geq 1$. Thus, our representatives of each supernode are appropriately far apart.
    
    It remains to show that our supernodes indeed form a $K_{j+1}$ minor; clearly by construction they are all pair-wise adjacent and so it remains only to show that they are all disjoint from one another. We already argued above that for $i, j \geq 1$ any $V_i'$ and $V_j'$ are disjoint so we need only argue that $V_0'$ is disjoint from each $V_i'$ for $i \geq 1$. $P_0$ must be disjoint from each $V_i'$ for $i \geq 1$ by \Cref{eq:PNot} and so we need only verify that $P_i$ is disjoint from $V_j'$ for $i,j \geq 1$; 
    
    By construction if $i = j$ we know that $P_i$ is disjoint from $V_j'$ so we assume $i \neq j$ and that $P_i$ intersects $V_j'$ for the sake of contradiction. Notice that each $P_i$ has length at most $k\Delta + \frac{c\Delta}{2} - 2 \Delta = k\Delta + c\Delta - 2 \Delta - \frac{c \Delta}{2} < (k-1)\Delta - \frac{c \Delta}{2}$ by how we construct $V_i'$. Thus, $P_i$ must be disjoint from $A_k$. It follows that if $P_i$ intersects $V_j'$ then it must intersect $V_j' \setminus V_j$. However, since $d(v_i, v_j) \geq 8(h-j + 1)\Delta \geq 16 \Delta$ and the length of paths $V_i' \setminus V_i$ and $V_j' \setminus V_j$ are at most $4\Delta$ we know that $d(V_i' \setminus V_i, V_j' \setminus V_j) \geq 8(h-j)\Delta \geq 8 \Delta$. Thus, after intersecting $V_j' \setminus V_j$ and then continuing on to a vertex adjacent to $V_i' \setminus V_i$, we know $P_i$ must travel at least $8 \Delta$; since the vertices of $P_i$ are monotonically further and further from $r$, and the vertex in $V_j'\setminus V_j$ that $P_i$ intersects must be distance at least $(k-1)\Delta - \frac{c \Delta}{2} - 4 \Delta \geq (k-5)\Delta$ from $r$, then the last vertex of $P_i$ must be distance at least $(k+3)\Delta$ from $r$, meaning $P_i$ must intersect annuli $A_k$, a contradiction.
\end{proof}

\subsection{Scattering Chops}

Using \Cref{lem:fuzChop} we can reduce computing a good scattering partition and therefore computing a good SPR solution to finding what we call a scattering chop. The following definitions are somewhat analogous to \Cref{dfn:SP} and \Cref{dfn:scatterableGraph}. However, notice that the second definition is for a family of graphs (as opposed to a single graph as in \Cref{dfn:scatterableGraph}). We illustrate a $\tau$-scattering chop in \Cref{fig:scatChops}.

\begin{figure}
    \centering
    \begin{subfigure}[b]{0.49\textwidth}
        \centering
        \includegraphics[width=\textwidth,trim=0mm 0mm 0mm 90mm, clip]{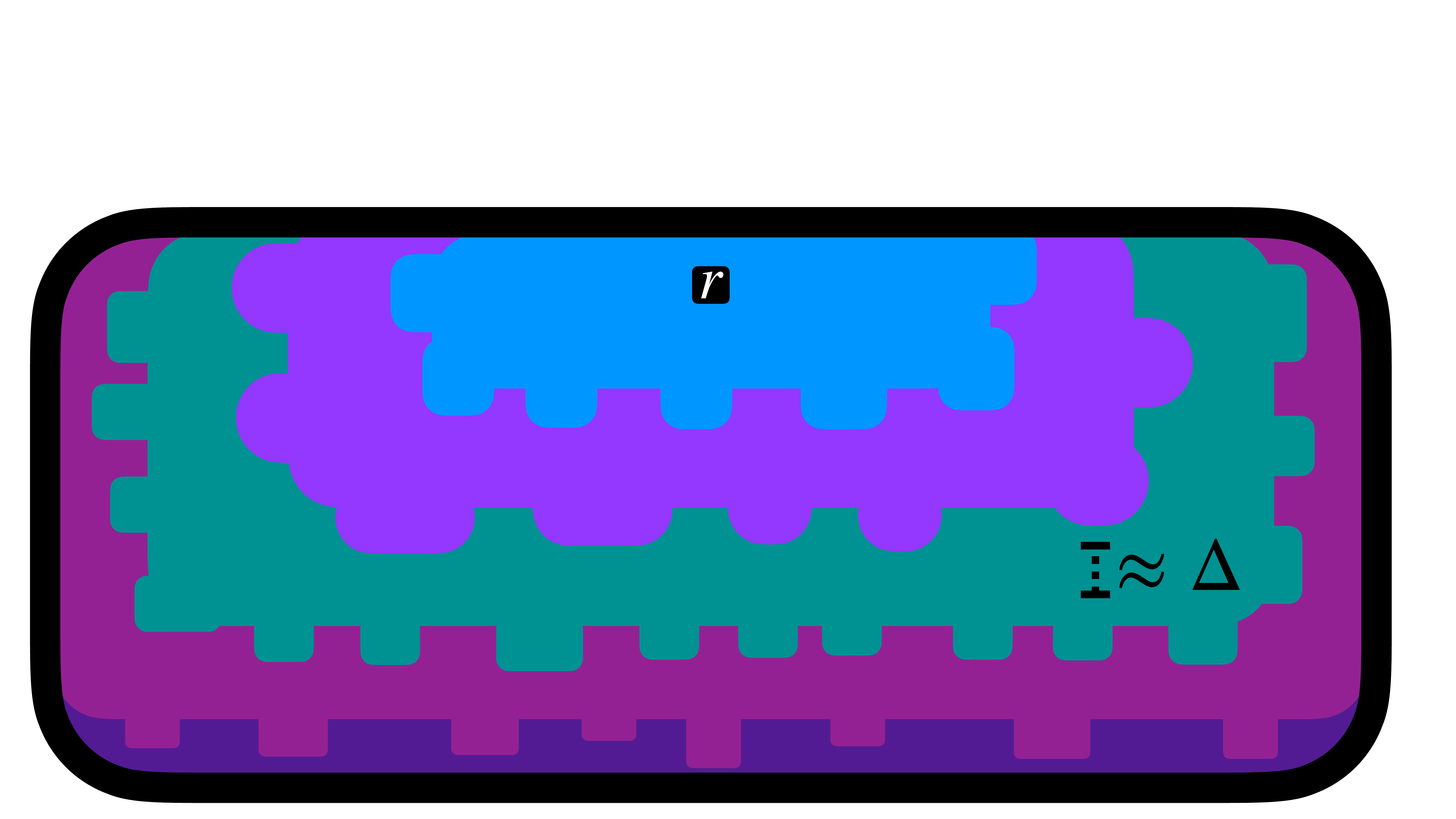}
        \caption{A $c$-fuzzy $\Delta$-chop.}\label{sfig:sChop1}
    \end{subfigure}
    \hfill
    \begin{subfigure}[b]{0.49\textwidth}
        \centering
        \includegraphics[width=\textwidth,trim=0mm 0mm 0mm 90mm, clip]{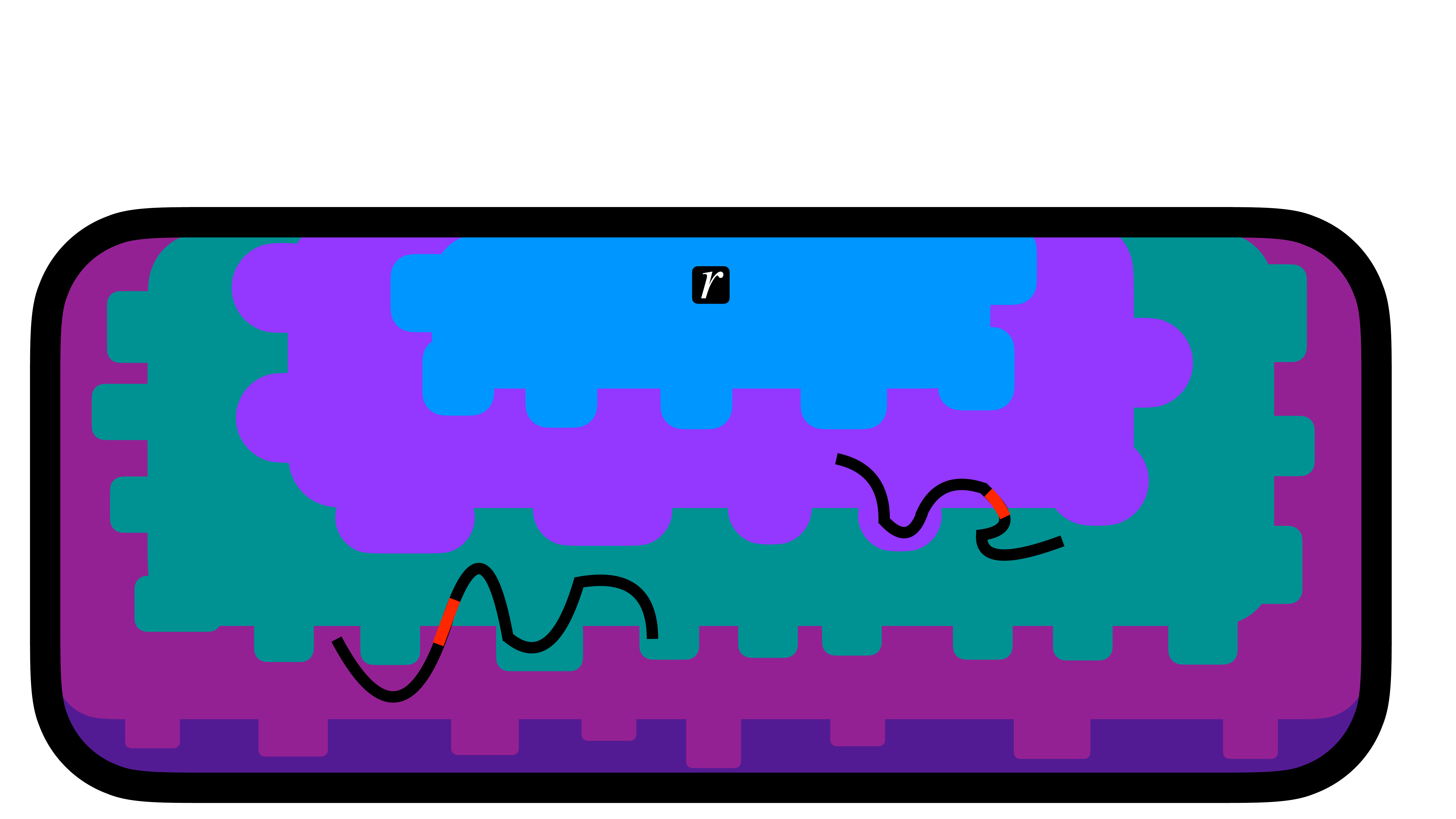}
        \caption{Visualizing some paths.}\label{sfig:sChop2}
    \end{subfigure}
    \caption{A $c$-fuzzy $\Delta$-chop that is $1$-scattering. We draw each fuzzy annuli in a distinct color. In (b) we visualize some shortest paths of length at most $\Delta$ and highlight cut edges in red.}\label{fig:scatChops}
\end{figure}

\begin{definition}[$\tau$-Scattering Chop]
    Given weighted graph $G=(V,E,w)$, let $\mcA$ be a $c$-fuzzy $\Delta$-chop with respect to some root $r \in V$. $\mcA$ is a $\tau$-scattering chop if each shortest path of length at most $\Delta$ has at most $\tau$ edges cut by $\mcA$ where we say that an edge is cut by $\mcA$ if it has endpoints in different fuzzy annuli of $\mcA$.
\end{definition}


\begin{definition}[$\tau$-Scatter-Choppable Graphs]\label{dfn:scatChop}
    A family of graphs $\mcG$ is $\tau$-scatter-choppable if there exists a constant $0 \leq c < 1$ such that for any $G \in \mcG$ and $\Delta \geq 1$ there is some $\tau$-scattering and $c$-fuzzy $\Delta$-chop $\mcA$ with respect to some root.
\end{definition}
We will say that $\mcG$ is deterministic poly-time $\tau$-scatter-choppable if the above chop $\mcA$ for each $G \in \mcG$ can be computed in deterministic poly-time.

Lastly, we conclude that to give an $O(1)$-scattering partition---and therefore to give an $O(1)$-SPR solution---for a $K_h$-minor-free graph family it suffices to show that such a family is $O(1)$-scatter choppable.
\begin{lemma}\label{lem:scatterChop}
    Fix a constant $h \geq 2$ and let $\mcG_h$ be all $K_h$-minor-free graphs. Then, if $\mcG_h$ is $\tau$-scatter-choppable then every $G \in \mcG_h$ is $O(\tau^{h-1})$-scatterable.
\end{lemma}
\begin{proof}
    The claim is almost immediate from \Cref{lem:fuzChop} and the fact that all subpaths of a shortest path are themselves shortest paths. 
    
    In particular, first fix a sufficiently small constant $c'$ to be chosen later. Then, consider a $G \in \mcG_h$ and fix a $\Delta$. By assumption we know that $G$ is $\tau$-scatter-choppable and since each subgraph of $G$ is in $\mcG_h$ so too is each subgraph of $G$. Thus, we may let $\mcC$ be the connected components resulting from $h-1$ levels of $c$-fuzzy and $(c'\Delta)$-chops which are $\tau$-scattering. 
    
    We claim that for sufficiently small $c'$ we have that $\mcC$ is a $\left(\frac{\tau^{h-1}}{c'}, \Delta\right)$-scattering partition. By \Cref{lem:fuzChop} the diameter of each part in $\mcC$ is at most $O(c' \cdot h \cdot \Delta) \leq \Delta$ for sufficiently small $c'$. Next, consider a shortest path $P$ of length at most $\Delta$. We can partition the edges of $P$ into at most $\frac{1}{c'}$ shortest paths $P_1, P_2, \ldots$, each of length at most $c' \cdot \Delta$. Thus, it suffices to show that each $P_i$ satisfies $|\{C \in \mcC: P_i \cap C \neq \emptyset \}| \leq \tau^{h-1}$. 
    
    We argue by induction on the number of levels of chops that after $h' < h$ chops we have $|\{C \in \mcC: P_i \cap C \neq \emptyset \}| \leq \tau^{h'}$. Suppose we perform just one chop; i.e.\ $h'=1$. Then, since our chops are $\tau$-scattering we know that $P$ will be cut at most $\tau$ times and so be incident to at most $\tau$ components of $\mcC$ as required. Next, suppose we perform $h' > 1$ levels of chops. Then our top-level chop will partition the vertices of $P_i$ into at most $\tau$ components. By induction and the fact that each subpath of $P_i$ is itself a shortest path of length at most $c' \Delta$, we know that the vertices of $P_i$ in each such component are broken into at most $\tau^{h'-1}$ components and so $P_i$ will be incident to at most $\tau^{h'}$ components as required. As we perform $h-1$ levels of chops, it follows that $\mcC$ is indeed a $\left(\frac{\tau^{h-1}}{c'}, \Delta\right)$-scattering partition.
\end{proof}



\section{Hammock Decompositions}\label{sec:hamDec}

In this section we formally define our hammock decompositions and give some of their properties. For the rest of this section we will assume we are working with a fixed graph $G = (V,E)$ with a fixed but arbitrary root $r\in V$ and a fixed but arbitrary BFS tree $\TBFS$ rooted at $r$. Throughout this section we will extensively use the notational conventions specified in \Cref{sec:notationAssumptions}, especially the BFS tree notation.

\subsection{Trees of Hammocks}

We begin by formalizing and establishing the properties of the key component of our hammock decompositions, the tree of hammocks. Roughly speaking, a tree of hammocks will be a tree in the usual graph-theoretic sense whose nodes are structured subgraphs which we call hammocks. Much of this section will be concerned with showing that a tree of hammocks, despite having subgraphs as its nodes, satisfies nice properties analogous to those of a tree of nodes.

We begin by defining a hammock graph as two subtrees of $\TBFS$ along with all cross edges between them. We illustrate a hammock in \Cref{sfig:ham}. Recall that $E_c := E \setminus E(\TBFS)$ gives the ``cross'' edges of $\TBFS$ and $\highV$ gives the vertex closest to the root of $\TBFS$.

\begin{figure}
    \centering
    \begin{subfigure}[b]{0.49\textwidth}
        \centering
        \includegraphics[width=\textwidth,trim=0mm 0mm 0mm 0mm, clip]{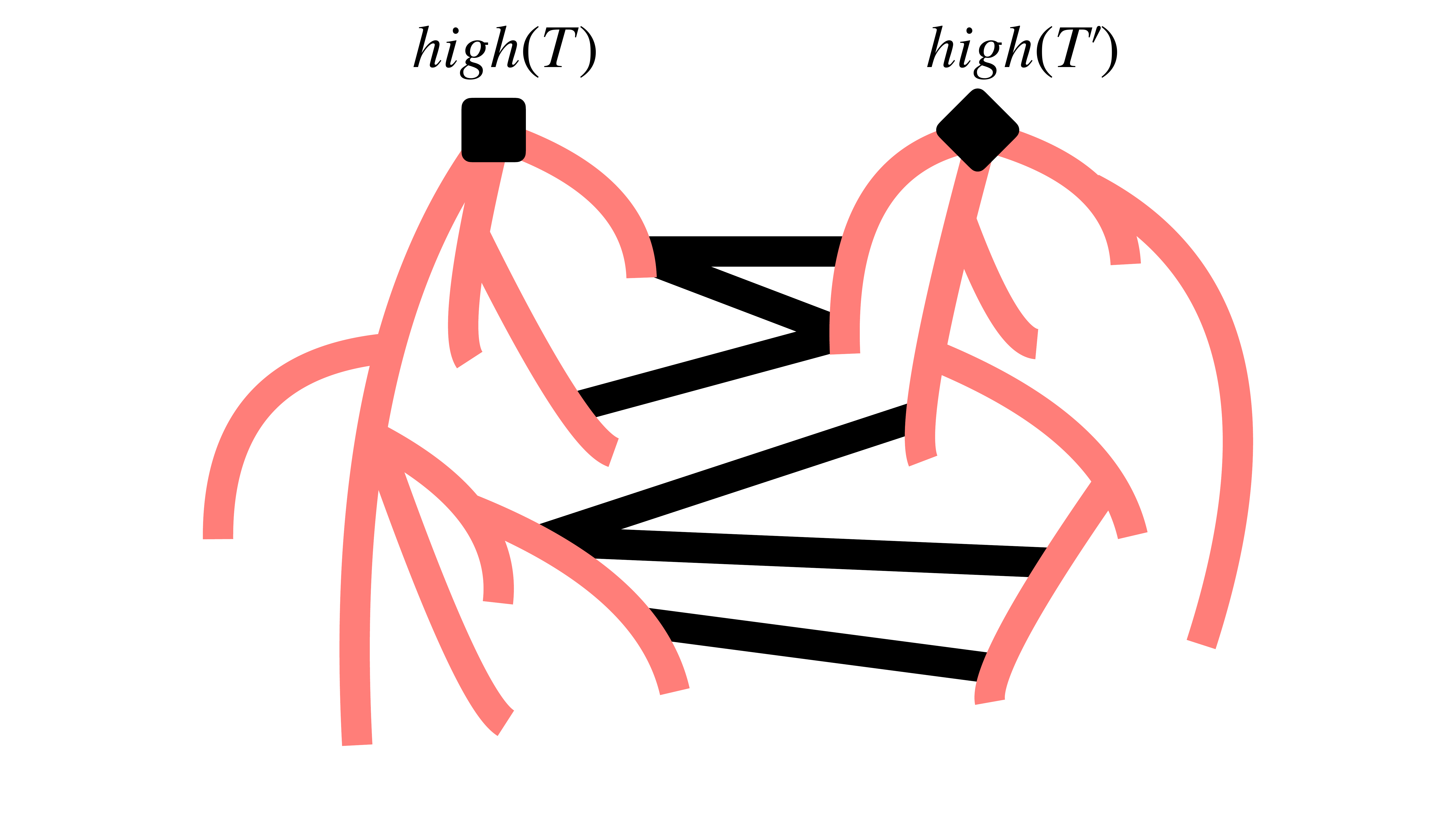}
        \caption{A hammock.}\label{sfig:ham}
    \end{subfigure}
    \hfill
    \begin{subfigure}[b]{0.49\textwidth}
        \centering
        \includegraphics[width=\textwidth,trim=0mm 0mm 0mm 0mm, clip]{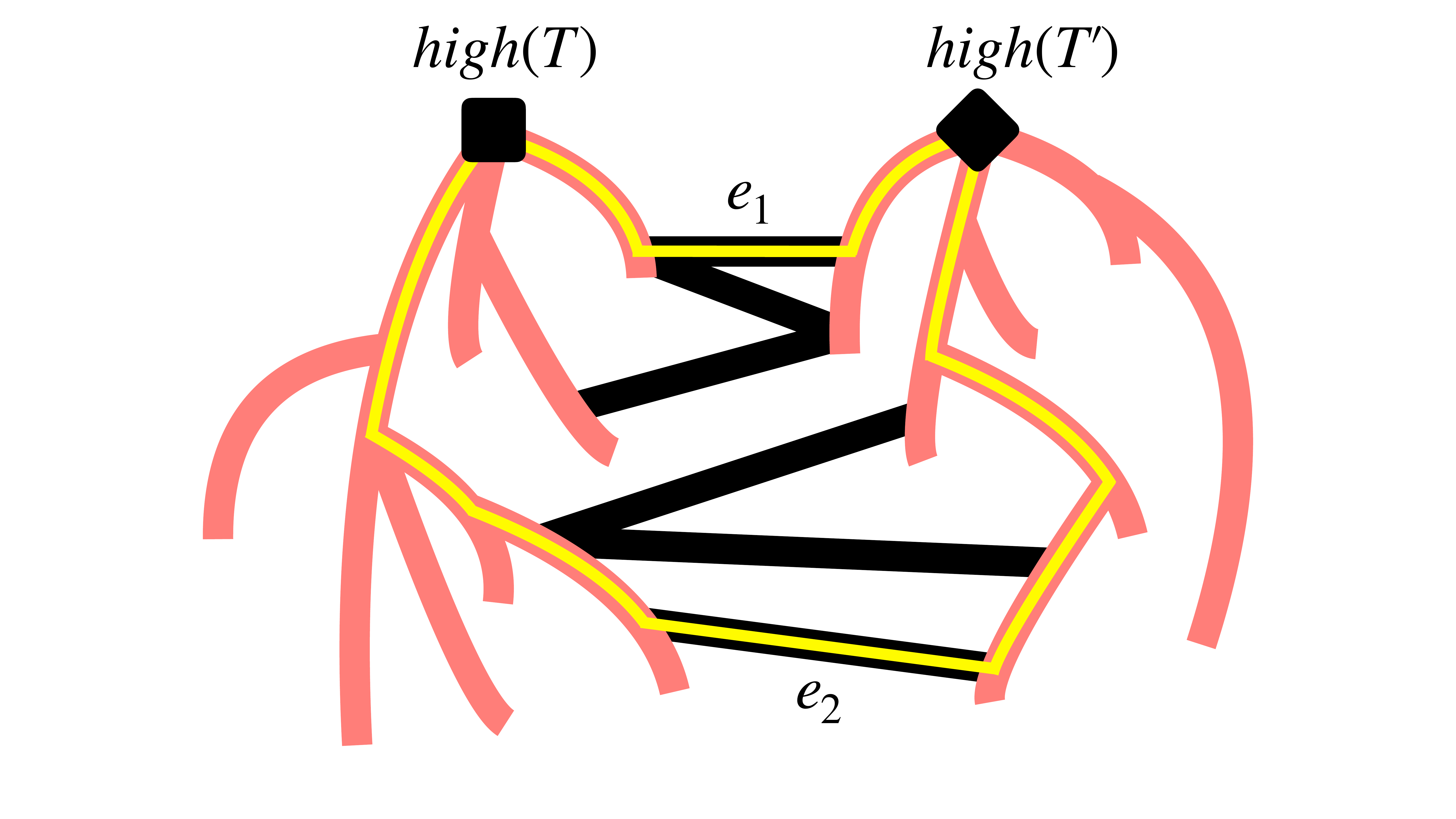}
        \caption{A hammock-fundamental cycle.}\label{sfig:hamCyc}
    \end{subfigure}
    \caption{An illustration of a hammock and a hammock-fundamental cycle for $e_1, e_2 \in E_c$. Edges of $\TBFS$ in pink, cross edges in black, hammock roots are  a black square and diamond and the hammock-fundamental cycle is in yellow.}\label{fig:ham}
\end{figure}

\begin{definition}[Hammock Graph]
    We say that subgraph $H \subseteq G$ is a hammock if $H = G[V(H)]$ and $V(H)$ can be partitioned into sets $V_1$ and $V_2$ where its ``hammock trees'' $T := \TBFS[V_1]$ and $T' := \TBFS[V_2]$ are connected, $E_c(T, T') \neq \emptyset$ and the ``hammock roots'' $\highV(T)$ and $\highV(T')$ are unrelated.
\end{definition}

Notice that given a hammock $H$ with two distinct edges $e_1 = \{v, v'\}, e_2 = \{u, u'\} \in E_c(T, T')$ where $v, u \in T$ and $v', u' \in T'$, we have that there is a unique cycle containing $e_1$ and $e_2$, namely $T(v, u) \oplus e_1 \oplus T'(u', v') \oplus e_2$. We will refer to this cycle as the hammock-fundamental cycle of $e_1$ and $e_2$. We illustrate a hammock-fundamental cycle in \Cref{sfig:hamCyc}.

For the following definition of a tree of hammocks, recall that if $\mcH$ is a collection of subgraphs of $G$ then $G[\mcH]$ gives the graph induced by the union of the subgraphs contained in $\mcH$.
\begin{definition}[Tree of Hammocks]\label{dfn:forsHam}
    We say that a collection of edge-disjoint hammocks $\mcH = \{H_i \subseteq G\}_i$ forms a tree of hammocks if every simple cycle $C$ in $G[\mcH]$ satisfies $|H_i : E(C) \cap H_i \neq \emptyset| = 1$ and $G[\mcH]$ has a single connected component.
\end{definition}

We illustrate a tree of hammocks in \Cref{sfig:treeHam}. Just as it is often useful to root a tree of vertices, so too will it be useful for us to root our trees of hammocks.

\begin{definition}[Rooted Tree of Hammocks]
    Suppose $\mcH = \{H_i\}_i$ forms a tree of hammocks. Then we say that $\mcH$ forms a rooted tree of hammocks if some $H_k \in \mcH$ is designated a ``root hammock.'' 
\end{definition}

While the above definitions are concerned with trees of hammocks, they extend naturally to a notion of forests of hammocks.

\begin{definition}[Rooted Forests of Hammocks]
    We say that a collection of edge-disjoint hammocks  $\mcH = \{H_i\}_i$ forms a (rooted) forest of hammocks if for each connected component $W$ of $G[\mcH]$ we have that $\{H_i : V(H_i) \cap V(W) \neq \emptyset\}$ forms a (rooted) tree of hammocks. We call the trees of hammocks formed by each such connected component the trees of hammocks of $\mcH$.
\end{definition}

All of our notation and definitions will extend from trees of hammocks to forests of hammocks in the natural way.

\subsubsection{Parents and Ancestors in Trees of Hammocks}

While defining a root for a tree of nodes immediately determines the parents and ancestors of each node, it is not so clear that defining a root hammock in a tree of hammocks determines a reasonable notion of parents and ancestors of hammocks. In this section, we provide some simple observations that, in turn, will allow us to coherently define parent and ancestor relationships in a tree of hammock. Along the way we will provide an explicit tree representation of a tree of hammocks $\mcH$ as a tree $T_{\mcH}$.

\begin{figure}
    \centering
    \begin{subfigure}[b]{0.49\textwidth}
        \centering
        \includegraphics[width=\textwidth,trim=50mm 20mm 20mm 15mm, clip]{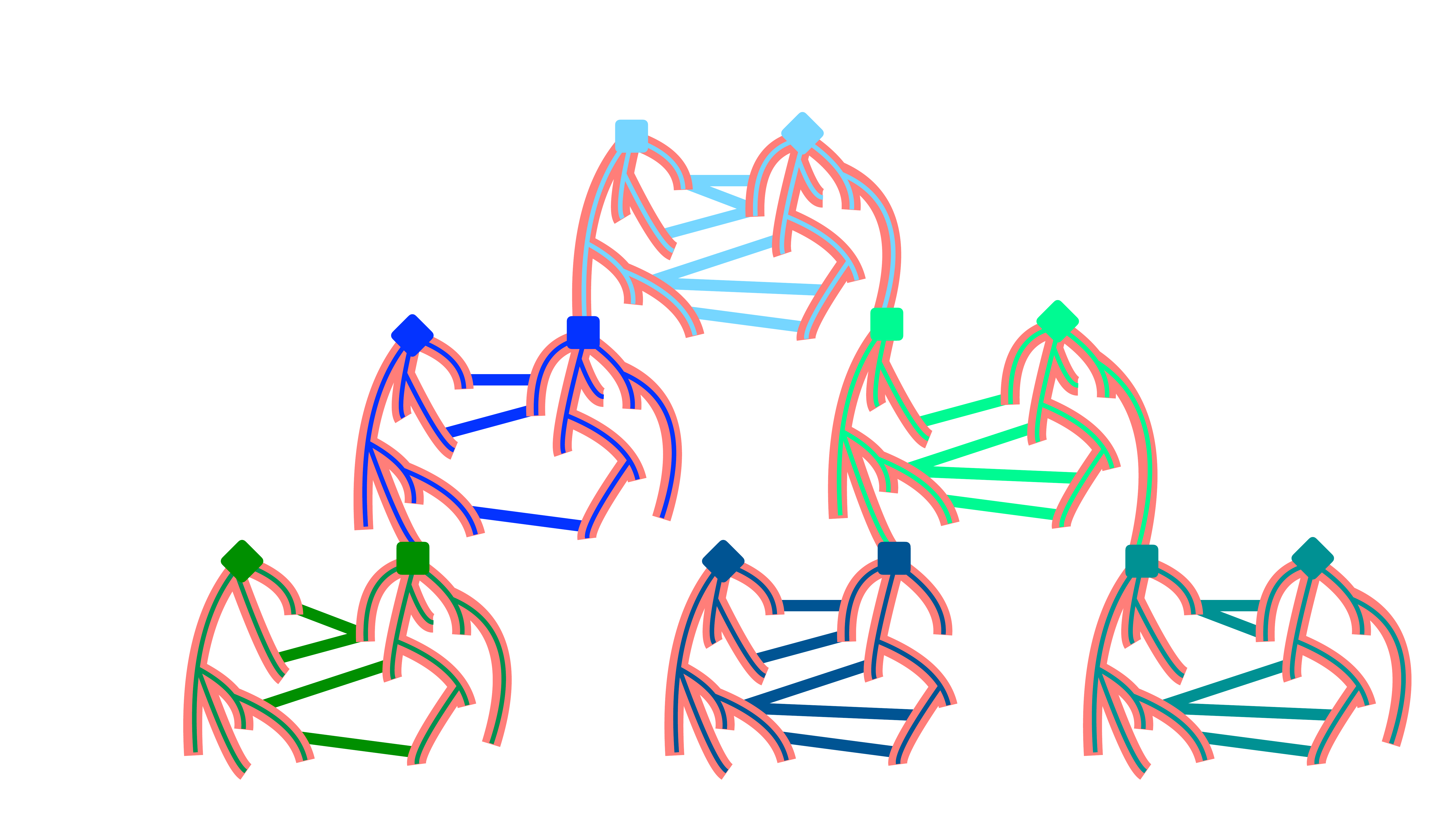}
        \caption{A tree of hammocks.}\label{sfig:treeHam}
    \end{subfigure}
    \hfill
    \begin{subfigure}[b]{0.49\textwidth}
        \centering
        \includegraphics[width=\textwidth,trim=0mm 50mm 0mm 0mm, clip]{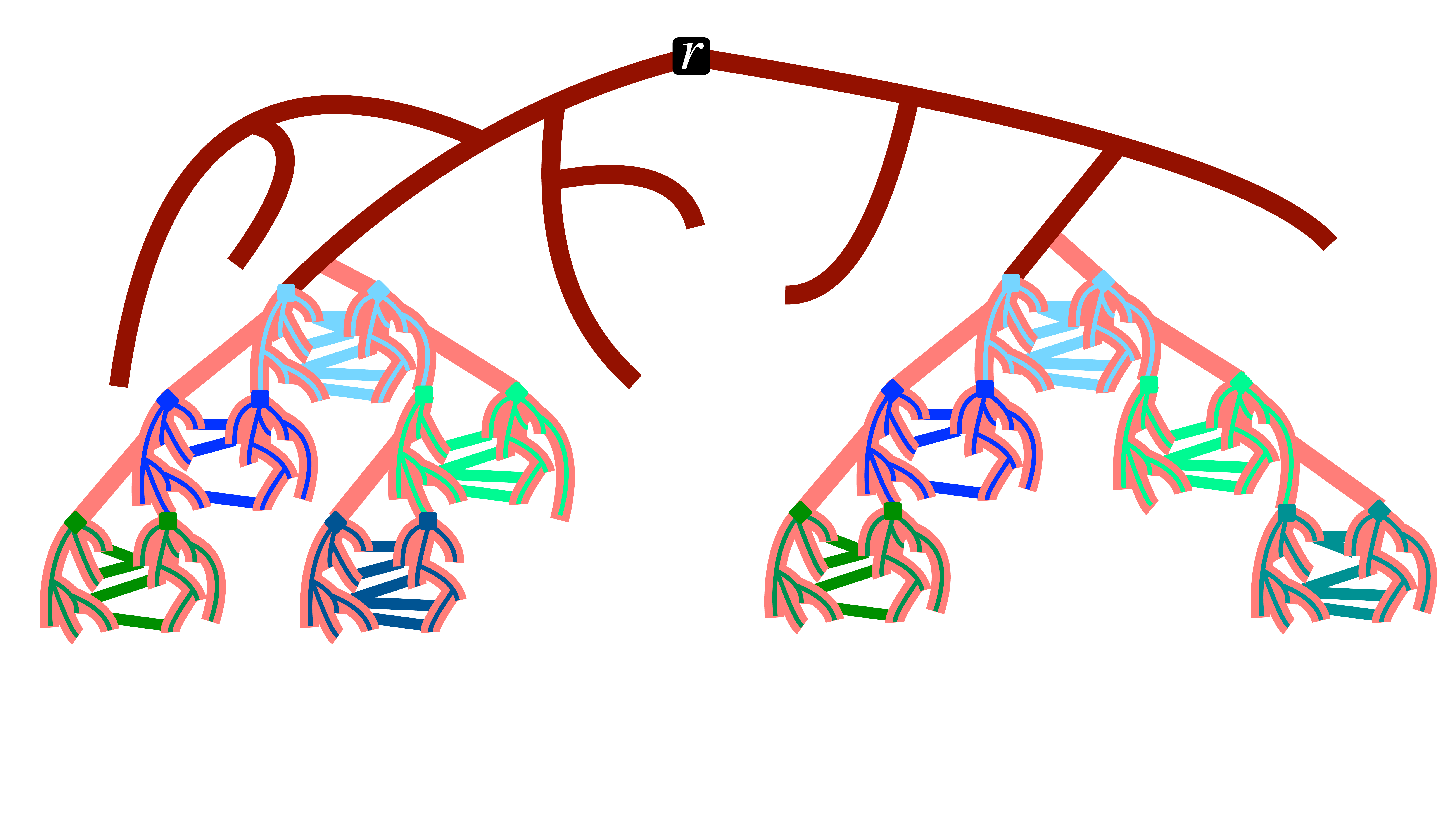}
        \caption{A hammock decomposition.}\label{sfig:hamDecomp}
    \end{subfigure}
    \caption{An illustration of one tree of hammocks in a hammock decomposition $\{H_i\}_i$ and a hammock decomposition consisting of two trees of hammocks. Each $r_i$ and $r_i'$ given as a square and diamond colored according to corresponding $H_i$. Edges in $\TBFS$ in pink and highlighted according to the $H_i$ which contains them. Edges of $E_c$ colored according to the $H_i$ which contains them. $T_0$ in the hammock decomposition given in dark red. 
    }\label{fig:hamTree}
\end{figure}

We begin by showing how, given a fixed root hammock in a tree of hammocks, we can define what it means for one hammock to be the parent of another hammock. To do so we need the following simple observation. Let $\mcH$ be a forest of hammocks and let $P$ be a path in the graph induced by $\mcH$ between hammocks $H_i$ and $H_j$. We say that $P$ \emph{passes through} hammock $H_l \in \mcH$ if it contains at least one edge of $H_l$.
Then, we have the following fact which is analogous to the fact that there is a unique simple path in a tree between any two vertices in a tree.  We remind the reader that a path is from subgraph $H_i$ to subgraph $H_j$ (in what follows these subgraphs are hammocks) iff its first and last vertices are in $H_i$ and $H_j$ and these are the only vertices of the path in $H_i$ and $H_j$.
\begin{lemma}\label{lem:uniquePath}
    Suppose $\mcH$ forms a tree of hammocks. Then for any $H_i, H_j \in \mcH$ if $P$ and $P'$ are both from $H_i$ to $H_j$ then both $P$ and $P'$ pass through the same hammocks of $\mcH$ in the same order.
\end{lemma}
\begin{proof}
    Suppose that $P$ passes through $H_1, H_2, \ldots$ and $P'$ passes through $H_1', H_2', \ldots$. Moreover, let $H_0 = H_0' := H_i$ and let $H_k = H_k' := H_j$. Then, let $x$ be any index such that $H_x = H_x'$ but $H_{x+1} \neq H_{x+1}'$ and let $y$ be the lowest index larger than $x$ such that $H_y = H_y'$. By our choice of $H_0$ and our assumption that both paths pass through different hammocks we know that $x$ is well-defined while our choice of $H_k$ shows that $y$ is well-defined. Let $\bar{P}$ and $\bar{P}'$ be $P$ and $P'$ restricted to the corresponding subpath between $H_x$ and $H_y$ and let $v_x$ and $v_y$ and $v_x'$ and $v_y'$ be the endpoints of $\bar{P}$ and $\bar{P}'$ in $H_x$ and $H_y$ respectively. Then we have that $\bar{P}$ combined with $\bar{P}'$ along with the path between $v_x$ and $v_x'$ in $H_x$ and the path between $v_y$ and $v_y'$ in $H_y$ is a cycle; call this cycle $C$. Moreover, since $\bar{P}$ and $\bar{P}'$ each pass through at least 1 distinct hammock we have that such a cycle satisfies $|H_l : E(C) \cap H_l \neq \emptyset| \geq 2$, contradicting our assumption that $\mcH$ is a forest of hammocks as defined in \Cref{dfn:forsHam}.
\end{proof}
The above lemma now allows us to formally define what it means for a hammock to be the parent of another hammock in a tree of hammocks.

\begin{definition}[Parent/Child Hammocks]
    Suppose $\mcH$ is a rooted tree of hammocks with root $H_k$. Then, we say that $H_k$ is the parent of all $H_j$ for $j\neq k$ which share a vertex with $H_k$. Similarly, for any other $H_j \neq H_k$ we let the parent of $H_j$ be $H_i$ where $H_i$ is the unique hammock any path from $H_j$ to $H_k$ first passes through (as guaranteed to exist by \Cref{lem:uniquePath}). Symmetrically, $H_j$ is a child of hammock $H_i$ if $H_i$ is a parent of $H_j$.
\end{definition}

More generally, we can use this notion of what it means for a hammock to be a parent of another hammock in order to make the ``treeness'' of a tree of hammocks more explicit and to define what it means for a hammock to be the ancestor of another hammock.

\begin{definition}[Tree Representation $T_\mcH$]
    Suppose $\mcH = \{H_i\}_i$ forms a rooted tree of hammocks with root $H_k$. Then $T_\mcH$ is the graph with vertex set $\{H_i\}_i$ and root $H_k$ which has an edge between $H_j$ and $H_i$ iff $H_i$ is the parent of $H_j$ in $\mcH$.
\end{definition}

\begin{lemma}
    If $\mcH$ is a rooted tree of hammocks then $T_\mcH$ is a tree where $H_i$ is a parent of $H_j$ in $T_\mcH$ iff $H_i$ is a parent of $H_j$ in $\mcH$.
\end{lemma}
\begin{proof}
    Let $H_k$ be the designated root hammock in $\mcH$. Now assume for the sake of contradiction that we have a cycle $H_0, H_1, \ldots, H_{x-1}$ in $\mcH$. Since each vertex/hammock in $\mcH$ has at most one parent, it follows that $H_{i+1 \mod x}$ is the parent of $H_{i}$ for every $i \in [x-1] $ (where we have assumed WLOG that $H_{i+1 \mod x}$ is the parent of $H_{i}$ and not that $H_i$ is the parent of $H_{i+1 \mod x}$). By definition of what it means for $H_{i+1 \mod x}$ to be the parent of $H_i$, it follows that every path from $H_i$ to $H_k$ must pass through $H_{i+1 \mod x}$, $H_{i+2 \mod x}$ and so on. However, since we have a cycle it then follows that every path from $H_i$ to $H_k$ must pass through $H_i$ itself, contradicting the fact that there is a path in $G[\mcH]$ between $H_i$ and $H_k$ which does not pass through $H_i$ since $G[\mcH]$ is connected.
    
    Moreover, by how we define $T_{\mcH}$ it follows that $H_i$ is a parent of $H_j$ in $T_\mcH$ iff $H_i$ is a parent of $H_j$ in $\mcH$.
\end{proof}

As we have established that $T_{\mcH}$ is a tree, we henceforth assume it is rooted at $H_k$, the designated root hammock of $\mcH$. With this in mind, we can define ancestors and descendants in $\mcH$.

\begin{definition}[Ancestor/Descendant Hammock]
    Let $\mcH = \{H_i\}_i$ be a rooted tree of hammocks. $H_i$ is an ancestor (resp. descendant) of $H_j$ in $\mcH$ iff $H_i$ is an ancestor (resp. descendant) of $H_j$ in $T_{\mcH}$.
\end{definition}

Similarly to our BFS tree notation, we will use the notation $H_j \preceq_{\mcH} H_i$ to indicate that $H_j$ is a descendant of $H_i$ in tree of hammocks $\mcH$.





\subsubsection{The Structure of Paths in Trees of Hammocks}

We now observe that if $\mcH$ is a tree of hammocks then any path in $G[\mcH]$ coincides with the corresponding path in $T_{\mcH}$.

We begin by observing the following simple technical lemma.

\begin{lemma}\label{lem:artPoint}
    If $\mcH = \{H_i\}_i$ is a tree of hammocks then for every $i,j$ we have $H_i \cap H_j \neq \emptyset$ implies that $H_i \cap H_j$ consists of a single vertex which is an articulation vertex for $G[\mcH]$.
\end{lemma}
\begin{proof}
    Suppose for the sake of contradiction that two hammocks $H_i$ and $H_j$ share a vertex $v$ which is not an articulation point for the graph induced by $\mcH$ and let $v_i$ and $v_j$ be vertices adjacent to $v$ in $H_i$ and $H_j$ respectively. Then, since $v$ is not an articulation point there is a path $P$ from $v_i$ to $v_j$ not containing $v$. It follows that $C := \{v, v_i\} \oplus P \oplus \{v_j, v\}$ is a cycle. However, we then have $\{H_i, H_j \} \subseteq \{H_l : E(C) \cap H_l \neq \emptyset\}$, violating \Cref{dfn:forsHam}.
\end{proof}


Concluding, we have the fact that the paths in $G[\mcH]$ adhere to the structure of $T_{\mcH}$. This fact is obvious if one inspects a picture of a tree of hammocks (see e.g. \Cref{sfig:treeHam}), but formalizing it requires a little bit of careful thought and notation. Again recall that if a path is between two subgraphs $H_i$ and $H_j$ then by definition the only vertices of $H_i$ and $H_j$ in this path are its first and last vertices; for this reason the indices of $P_x$ begin at $1$ and end at $l-1$ in the following definition. 
\begin{lemma}\label{lem:pathStruct}
    Let $\mcH$ be a rooted tree of hammocks, let $P$ be a path in $G[\mcH]$ between hammocks $H_i$ and $H_j$ and let $T_{\mcH}(H_i, H_j) = (H_i = H_0, H_1, H_2, \ldots, H_l = H_j)$ be the path between $H_i$ and $H_j$ in the tree representation $T_{\mcH}$ of $\mcH$. Then $P$ is of the form $P_1 \oplus P_2 \ldots \oplus P_{l-1}$ where for each $x \in [l-1]$ we have:
    \begin{enumerate}
        \item $P_x$ is a subpath of $P$ gotten by restricting $P$ to $V(H_x)$ where $E(P_x) \subseteq E(H_x)$ and;
        \item The first and last vertices of $P_x$ are articulation vertices $b_{(x-1)x}$ and $b_{x(x+1)}$ where $V(H_{x-1}) \cap V(H_{x}) = \{b_{(x-1)x}\}$ and $V(H_{x}) \cap V(H_{x+1}) = \{b_{x(x+1)}\}$.
    \end{enumerate}
\end{lemma}
\begin{proof}
    
    We begin with the following simple observation. Suppose $e = \{u, v\} \in P$. Then if $u$ is in some $H_x$ but $v \not \in H_x$ then the suffix of $P$ after and including $v$ cannot include any vertices of $H_x$. Suppose for the sake of contradiction that it did and let $P'$ be the subpath of $P$ from $u$ back to some $u' \in H_x$. Then combining $P'$ with the path in $H_x$ connecting $u$ and $u'$ gives a cycle with edges in more than one hammock, a contradiction. The symmetric statement holds for prefixes of $P$.
    
    As our hammocks in $\mcH$ are edge-disjoint we know that we can partition the edges of $P$ into their constituent hammocks. By our above observation we know that for each $H_x \in \mcH$ it holds that $H_x \cap E(P)$ is a connected (possibly singleton) subpath of $P$. We let $P'_x := H_x \cap E(P)$ be each such subpath and let $\mcP'$ be all such induced non-empty subpaths by hammocks along $T_{\mcH}(H_i, H_j)$.
    
    We proceed to show that $P_x'$ is the $x$th path in $P$ among all paths in $\mcP'$. Let $H_m$ be the hammock at maximum height in $T_{\mcH}$ among all hammocks in $T_{\mcH}(H_i, H_j)$. By definition of a parent, we know that in any path incident to a vertex in $H_x$, if said path has a vertex of $H_x$'s parent's parent then such a path must pass through $H_x$'s parent. Moreover, as each hammock is connected and shares a vertex with its parent, we know that there is a path from $H_i$ to $H_j$. It follows that there is a path $P'$ from $H_i$ to $H_j$ which passes through every hammock in $T_{\mcH}(H_i, H_j)$ except for possibly $H_m$. If $P'$ passes through $H_m$ then our claim holds by \Cref{lem:uniquePath} and our above observation. If $P'$ does not pass through $H_m$ then by \Cref{lem:uniquePath} we must verify that the vertex in both $P'_{m-1}$ and $P_{m+1}'$ is in $H_m$. However, since $H_m$ is the parent of both $H_{m-1}$ and $H_{m+1}$ we know that any vertices shared by $H_{m-1}$ and $H_{m+1}$ must be in $H_m$ as it is easy to see that otherwise we could construct a cycle with edges in multiple hammocks by connecting $H_{m-1}$ and $H_{m+1}$ via this vertex as well as through $H_m$.
    
    Lastly, we note that each $b_{x(x+1)}$ is an articulation vertex and the only vertex in $V(H_x) \cap V(H_{x+1})$ by \Cref{lem:artPoint}.
\end{proof}

\subsection{Hammock Decompositions}

We now define a hammock decomposition. Roughly, a hammock decomposition is a forest of hammocks which both contains all shortest paths which start and end with cross edges and whose forest structure reflects the lca structure of the BFS tree $\TBFS$. Recall that, ultimately, we will show that every series-parallel graph has such a decomposition and use this decomposition to perturb our KPR chops.

We will call a path a cross edge path if its first and last edges are in the cross edges of $\TBFS$ (that is, in $E_c$) and a shortest cross edge path if it is both a cross edge path and it is a shortest path in $G$. For a hammock $H_i$, we will henceforth use $T_i$ and $T_i'$ to stand for the hammock trees of $H_i$ and use $r_i = \highV(T_i)$ and $r_i' = \highV(T_i')$ to stand for the hammock roots of $H_i$. In the below definitions we will without loss of generality (WLOG) assume that each $r_i$ satisfies certain properties and that $r_i'$ satisfies certain other properties.

We begin by formalizing the sense in which a forest of hammocks can reflect the lca structure of the BFS tree. Specifically, recall that given a hammock $H_i$ with hammock roots $r_i$ and $r_i'$, by definition we know $r_i$ and $r_i'$ are unrelated in $\TBFS$. Thus, we can naturally associate hammock $H_i$ with the vertex $\lca(r_i, r_i')$. Then the condition we would like to enforce is that if hammock $H_i$ is an ancestor of $H_j$ in a forest of hammocks $\mcH$ then $\lca(r_i, r_i')$ is an ancestor of $\lca(r_j, r_j')$ in $\TBFS$. In fact, we will be able to enforce an even stronger condition than this in our hammock decompositions, as formalized by the following notion of an lca-respecting forest of hammocks.
\begin{definition}[lca-Respecting Forest of Hammocks]\label{dfn:lcaRespecting}
    Let $\mcH = \{H_i\}_i$ be a rooted forest of hammocks . Then, we say that $\mcH$ is lca-respecting if for every pair $H_i, H_j$ where $H_i$ is the parent of $H_j$ in $\mcH$ then:
    \begin{enumerate}
        \item $V(H_i) \cap V(H_j) = \{r_j\}$;
        \item The parent of $r_j'$ in $\TBFS$ is $\lca(r_j, r_j')$ and $\lca(r_j, r_j') \in V(H_i) \cup \{\lca(r_i, r_i')\}$.
    \end{enumerate}
    Furthermore, we say that $\mcH$ is lca-respecting with base tree $T_0 \subseteq \TBFS$ if for each root hammock $H_k \in \mcH$ we also have
   \begin{enumerate}
        \item $r_k \in V(T_0)$;
        \item The parent of $r_k$ and $r_k'$ in $\TBFS$ is $\lca(r_k, r_k')$ and $\lca(r_k, r_k') \in V(T_0)$.
    \end{enumerate}
\end{definition}

Notice that it follows that if $\mcH$ is lca-respecting then if $H_i$ is an ancestor of $H_j$ in $\mcH$ then $\lca(r_i, r_i')$ is an ancestor of $\lca(r_j, r_j')$ in $\TBFS$. Even stronger, though, we know that if $H_i$ is a parent of $H_j$ then $\lca(r_j, r_j') \in V(H_i)$ or $\lca(r_j, r_j')$ is equal to $\lca(r_i, r_i')$.

Concluding we may now give our definition of a hammock decomposition. Since every hammock contains a cross edge we cannot, in general, expect to decompose a graph into hammocks. For example, if the input graph was just a tree then said graph would contain no hammocks. Thus, our hammock decomposition will partition a series-parallel graph into a forest of hammocks along with a tree $T_0$ and a ``parent edge'' for each hammock in our forest of hammocks. 

\begin{definition}[Hammock Decomposition]\label{dfn:PCHD}
    A hammock decomposition of graph $G = (V,E)$ with root $r$ and BFS tree $\TBFS$ is a partition of $E$ into $E(T_0) \sqcup E(\mcH) \sqcup E_p$ where:
    \begin{enumerate}
        \item $T_0$ is a subtree of $\TBFS$ containing $r$;
        \item $\mcH$ is an lca-respecting rooted forest of hammocks with base tree $T_0$ such that $G[\mcH]$ contains every shortest cross edge path in $G$;
        \item  $E_p := \{e_i : H_i \in \mcH\}$ where $e_i$ is the parent edge of $r_i'$ in $\TBFS$.
    \end{enumerate}
\end{definition}

Thus, roughly, a hammock decomposition consists of a base tree $T_0$ with trees of hammocks ``hanging off'' of $T_0$. We illustrate a hammock decomposition where $\mcH$ consists of two trees of hammocks in \Cref{sfig:hamDecomp}.

\begin{figure}
    \centering
    \begin{subfigure}[b]{0.49\textwidth}
        \centering
        \includegraphics[width=\textwidth,trim=0mm 0mm 0mm 0mm, clip]{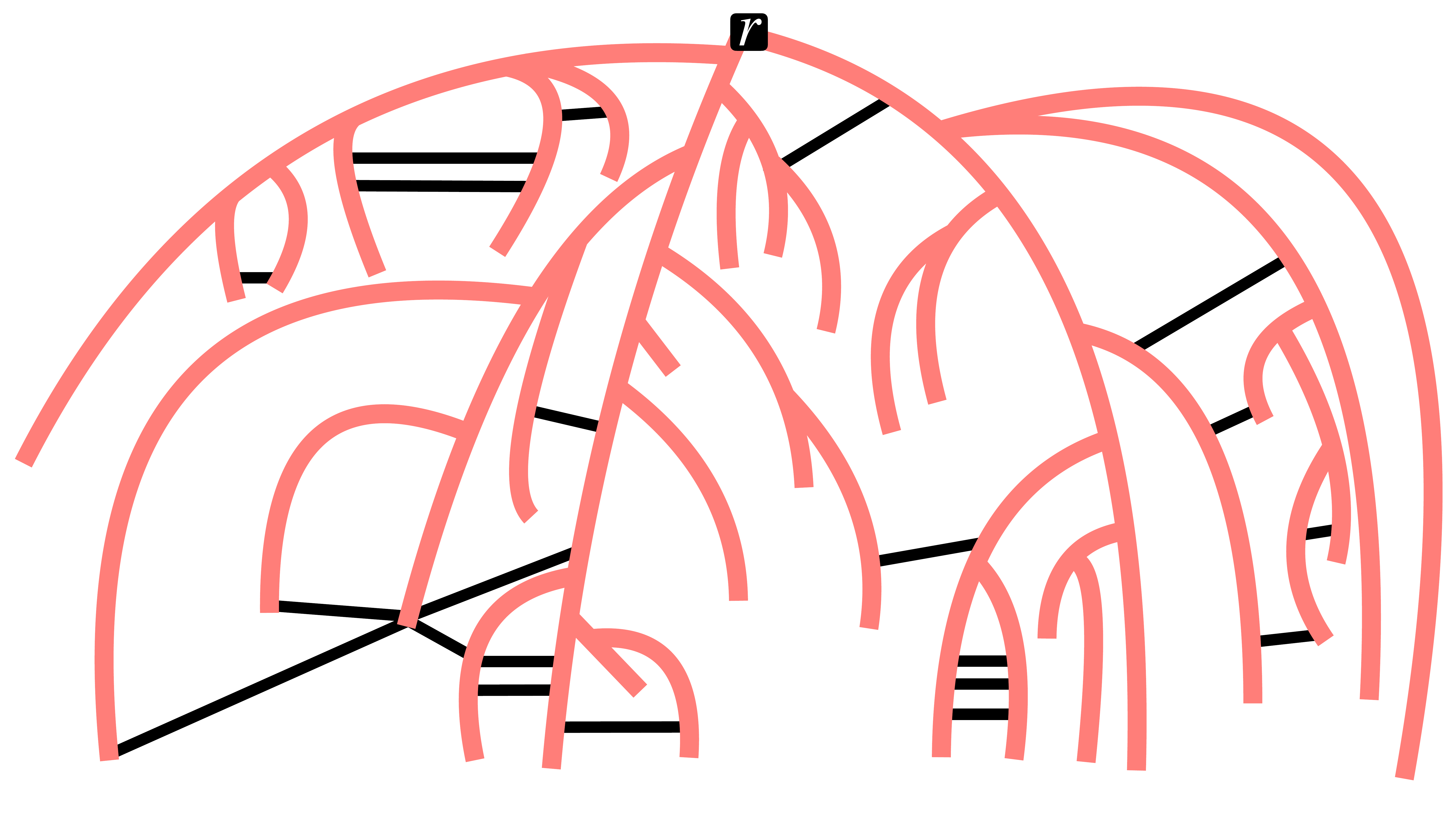}
        \caption{Series-parallel graph $G$.}\label{sfig:graphEG}
    \end{subfigure}
    \hfill
    \begin{subfigure}[b]{0.49\textwidth}
        \centering
        \includegraphics[width=\textwidth,trim=0mm 0mm 0mm 0mm, clip]{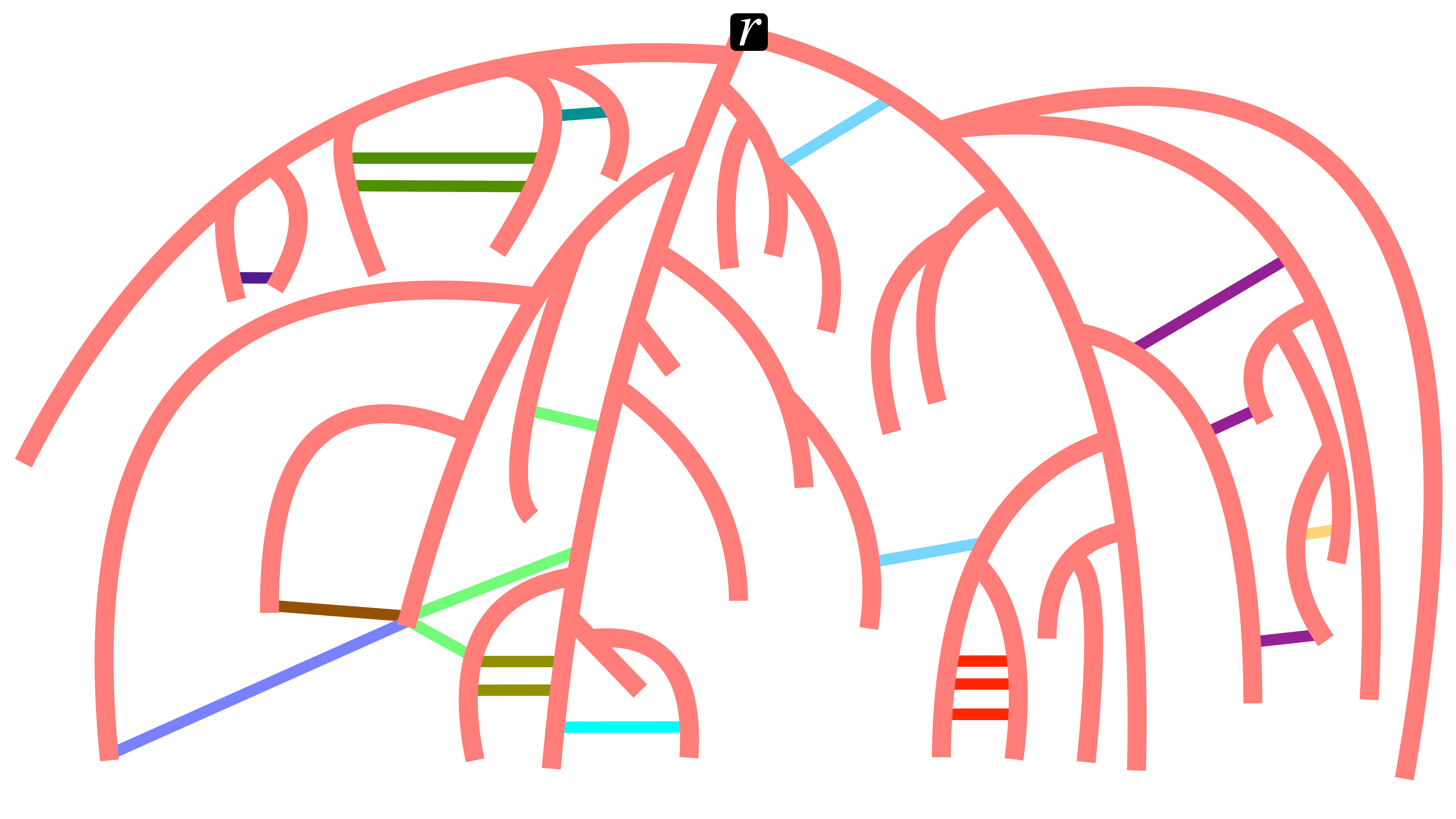}
        \caption{Lca-equivalence classes.}\label{sfig:lcaEquivClasses}
    \end{subfigure}
    \caption{An illustration of the lca-equivalence classes of a series-parallel graph $G$. Edges of $\TBFS$ in pink. Edges of $E_c$ in black in \Cref{sfig:graphEG} and colored according to their lca-equivalence class in \Cref{sfig:lcaEquivClasses}. Notice that edges with the same lca can belong to distinct equivalence classes.}\label{fig:constrPart}
\end{figure}

\section{Computing Hammock Decompositions for Series-Parallel Graphs}

In this section we show how to construct hammock decomposition for a series-parallel graph.  For the rest of this section we will assume we are working with a fixed series-parallel graph $G = (V,E)$ with a fixed but arbitrary root $r\in V$ and a fixed but arbitrary BFS tree $\TBFS$ rooted at $r$. As in the previous section we will extensively make use of the notation laid out in \Cref{sec:notationAssumptions}. The following theorem summarizes the main result of this section.
\begin{restatable}{theorem}{hamDecPC}\label{thm:hamDecPC}
    Every series-parallel graph has a hammock decomposition which can be computed in deterministic poly-time.
\end{restatable}


The main idea of our construction is to partition all cross edges into equivalence classes based on the behavior of the cross edges' least common ancestors. Notably, two cross edges with the same least common ancestors may end up in different equivalence classes. Each such equivalence class will correspond to one hammock in each of our hammock decomposition. 

More specifically, we build a series of forests of hammocks where at each step we add on to our forest of hammocks to guarantee one of the required properties of our hammock decomposition. We will use different notation for each of these forests of hammocks. For example $\hat{\mcH}$ will be our first forest of hammocks where a constituent hammock will be notated $\hat{H}_i$ with hammock trees $\hat{T}_i$ and $\hat{T}_i'$ and hammock roots $\hat{r}_i$ and $\hat{r}_i'$. We will also let $\hat{H}$ be the subgraph induced by $\hat{\mcH}$; and use symmetric notation for our other forests of hammocks. Our progression of forests of hammocks is as follows.
\begin{enumerate}
    \item $\hat{\mcH}$: we will initialize our forest of hammocks by connecting each cross edge that falls into the same equivalence class. The result of this will be one hammock for each of our equivalence classes.
    \item $\bar{\mcH}$: Next, we will extend each $\hat{H}_i$ to $\bar{H}_i$ by connecting these hammocks to one another along ``hammock-joining'' paths. The result of this will be a forest of hammocks $\bar{\mcH}$ which contains every shortest cross edge path.
    \item $\tilde{\mcH}$: Next we will extend each $\bar{H}_i$ along paths towards $\lca(\bar{r}_i, \bar{r}_i')$, resulting in $\tilde{H}_i$, to make our forest of hammocks lca-respecting
    \item $\mcH$: Lastly, we will add any edges of $\TBFS$ which do not appear in a hammock to an incident hammock (with the exception of the connected component of unassigned edges which contain $r$). This step will ensure that our hammock decomposition indeed partitions all edges of the graph.
\end{enumerate}


We will illustrate this process on the series-parallel graph given in \Cref{sfig:graphEG} throughout this section where \Cref{sfig:finalHD} gives the final hammock decomposition we compute for this graph. We also give all illustrations of the construction in a single figure in the appendix in \Cref{fig:allConstructionsAtOnce}.

Most of our proofs will revolve around finding clawed cycles when the above procedure fails to produce a forest of hammocks with the desired properties. As our proofs are quite lengthy, we will briefly highlight three of the major conceptual milestones of this section before proceeding.
\begin{enumerate}
    \item \textbf{Connected Below Subgraphs:} Connected below subgraphs will be a useful abstraction to help us find clawed cycles. An edge $e = \{u,v\} \in H$ where $u$ is the parent of $v$ in $\TBFS$ is connected below if there is a path from $r$ to $\TBFS(v)$ which only intersects $\TBFS(u)$ at $\TBFS(v)$. A subgraph will be connected below if each of its edges are. We will argue that the subgraph induced by our construction is connected below. As the paths guaranteed to exist for any connected below edges $\{u, v_l\}$ and $\{u, v_r\}$ (for $v_l$ and $v_r$ children of $u$ and $v_l \neq v_r$) have distinct endpoints in $\TBFS(u)$, we will often use these paths along with $\TBFS(r,u)$ as the paths of our clawed cycles. \Cref{lem:conBelGivesForest} gives a concise summary of this technique, showing that any connected below subgraph with at most one edge in each $C_i$ is a forest. The forest structure given by this lemma will be at the core of our proof of why our construction gives a \emph{forest} of hammocks.
    \item \textbf{Assignments of Hammock-Joining Path Components:} As alluded to above, after computing our $\hat{H}_i$s we will then connect these subgraphs to one another. We will do this by considering the graph induced by all ``hammock-joining paths''---roughly, the shortest paths which connect cross edges in different equivalence classes. We will then assign each connected component in this graph to one of its incident hammocks. The main idea here is to first argue that any assignment which is valid (in some later-described technical sense) will result in a forest of hammocks; see \Cref{lem:allAssignGiveForest}. Next, we will argue that no matter which valid assignment we use, for each collection of hammocks incident to one of the components which we are assigning there is some hammock which will always end up as an ancestor hammock of the other incident hammocks in the resulting forest of hammocks (\Cref{lem:someParent}). We will therefore assign each component to this always-ancestor hammock. That we may assume that each component is assigned to a hammock which is the ancestor of all other incident hammocks will further assist in arguing that our construction gives a forest of hammocks.
    \item \textbf{Tree of Hammock Ancestry $\leftrightarrow$ Lca Ancestry:} The idea that forms the foundation of how we extend our $\bar{H}_i$s to our $\tilde{H}_i$s to be lca-respecting is as follows. We will argue that if $\bar{H}_j$ is a descendant of $\bar{H}_i$ in $\bar{\mcH}$ then the lca corresponding to the cross edges of $\bar{H}_j$ must be a descendant of the lca corresponding to the cross edges of $\bar{H}_i$ in $\TBFS$. \Cref{lem:parChildEquiv} summarizes this fact. The fact that $\bar{\mcH}$'s forest structure reflects the lca structure of $\TBFS$ will be what allows us to ensure that our hammocks are lca-respecting.
\end{enumerate}

 

\subsection{Initial Hammocks $\hat{\mcH}$ by Connecting Equivalence Classes}
In this section we describe our initial forest of hammocks $\hat{\mcH}$. Roughly, we will define an equivalence relation for cross edges and then if two cross edges fall into the same equivalence we will connect the cross edges to one another. We will also introduce the notion of connected below subgraphs which will help us to argue that the result of this (and the next step in our construction) is a forest of hammocks.

\subsubsection{Lca-Equivalent Edges}

In this section we define the behavior of cross edges' least common ancestors which we will use to partition our cross edges into equivalence classes. Formally, we partition our cross edges based on ``lca-equivalence'' which we define as follows. In the following we WLOG distinguish between the endpoints of $e$ and $e'$.

\begin{definition}[Lca-Equivalent Edges]
    Let $e = \{u, v\}$ and $e' = \{u', v'\}$. Then we say that $e$ and $e'$ are lca-equivalent if $l(e) = l(e')$, $l(u, u') \prec l(e)$ and $l(v, v') \prec l(e)$.
\end{definition}

We emphasize that two cross edges with the same lca may end up in different equivalence classes. We illustrate these equivalence classes in \Cref{sfig:lcaEquivClasses}. Next, we verify that, indeed, these sets form an equivalence relation.

\begin{lemma}
    The set of lca-equivalent edges forms an equivalence relation.
\end{lemma}
\begin{proof}
    Reflexivity and symmetry are trivial. We prove transitivity. Suppose $e = \{u, v\}$ and $e' = \{u', v'\}$ are lca-equivalent and $e'$ and $e'' = \{u'', v''\}$ are lca-equivalent. We claim that $e$ and $e''$ are lca-equivalent. In particular, we have $l(e) = l(e')$ and $l(e') = l(e'')$ so $l(e) = l(e'')$. Now consider the (monotone) path from $u'$ to $l(e)$; this path contains both $l(u, u')$ and $l(u', u'')$; WLOG suppose $l(u', u'')$ occurs higher in $\TBFS$ in this path. It follows that $l(u', u'')$ has both $u$ and $u''$ as a descendant and so $l(u, u'') \preceq l(u', u'') \prec l(e)$. A symmetric argument shows $l(v, v'') \prec l(e)$ and so we conclude that $e$ and $e''$ are lca-equivalent.
\end{proof}

For the remainder of this section we will let $\mcC := \{C_i\}_i$ be the equivalence classes of the above equivalence relation. We will let $l(C_i)$ give the lca of any edge in $C_i$. Similarly, we will let $h(C_i)$ give the height of $l(C_i)$ in $\TBFS$. We will also slightly abuse notation and let $\preceq$ be the partial ordering of these equivalence classes according to their lcas: $C_j \preceq C_i$ iff $l(C_j) \preceq l(C_i)$ and $C_j \prec C_i$ iff $l(C_j) \prec l(C_i)$ where again $\prec$ and $\preceq$ give the ``ancestry'' partial ordering in $\TBFS$.

\subsubsection{Connected Below Subgraphs}

The crucial property of the edges which we use to connect edges in the same lca-equivalence class and to connect our $\hat{H}_i$s will be the following notion of ``connected below.'' This property of these edges will aid us in finding disjoint paths which will, in turn, allow us to construct clawed cycles when our hammock decomposition construction fails.

\begin{definition}[Connected Below]
    We say that $e = \{u, v\} \in \TBFS$ where $v \prec u$ is connected below if there is a path $P$ in $G$ between $r$ and $\TBFS(v)$ which satisfies $P \cap \TBFS(u) \subseteq \TBFS(v)$.  We say that a subgraph $G_U = (U, E_U) \subseteq G$ is connected below if each edge of $E_U \cap \TBFS$ is connected below.
\end{definition}

The following gives a slightly different but useful characterization of what it means for an edge to be connected below.
\begin{lemma}\label{lem:conBelowAlternateDfn}
    Let $e = \{u, v\} \in \TBFS$ be connected below where $v \prec u$ and let $F$ be a non-empty subgraph of $G$ with vertices contained in $\TBFS(v)$. Then there is a path $P$ in $G$ between $r$ and $F$ where $P \cap \TBFS(u) \subseteq \TBFS(v)$.
\end{lemma}
\begin{proof}
    By the definition of $e$ being connected below we know that there is a path $P'$ from $r$ to $\TBFS(v)$ satisfying $P' \cap \TBFS(u) \subseteq \TBFS(v)$. Extending this path through $\TBFS(v)$ to $F$ gives $P$.
\end{proof}

A simple proof by contradiction shows that two adjacent connected-below edges cannot also have a path connecting their children vertices in a series-parallel graph.
\begin{lemma}\label{lem:noFundCycle}
    There does not exist a vertex $x$ with distinct children $v_l$ and $v_r$ in $\TBFS$ such that $\{x, v_l\}$ and $\{x, v_r\}$ are connected below and there exists a path in $G$ contained in $\TBFS(x) \setminus \{x\}$ between $v_l$ and $v_r$
\end{lemma}
\begin{proof}
    Suppose for the sake of contradiction that the stated path exists and call it $P$. Let $C$ be the cycle created by taking the union of $P$, $\{x, v_l\}$ and $\{x, v_r\}$. Applying \Cref{lem:conBelowAlternateDfn} to the fact that $e_l$ and $e_r$ are connected below, we have that there exists a path $P_l$ from $r$ to $C$ which does not intersect $\TBFS(x) \setminus \TBFS(v_r)$. Symmetrically, there is a path $P_r$ from $r$ to $C_e$ which does not intersect $\TBFS(x) \setminus \TBFS(v_l)$. Letting $P_x := \TBFS(r, x)$, we have that $C_e$ with paths $P_x$, $P_l$ and $P_r$ forms a clawed cycle, a contradiction.
\end{proof}

Building on the previous lemma, we have our main fact for this section: any subgraph connected below with at most one edge from each equivalence class is a forest. That such a graph is a forest will form the basis of our proof that our hammock decompositions are indeed forests of hammocks.

\begin{figure}
    \centering
    \begin{subfigure}[b]{0.32\textwidth}
        \centering
        \includegraphics[width=\textwidth,trim=150mm 0mm 150mm 30mm, clip]{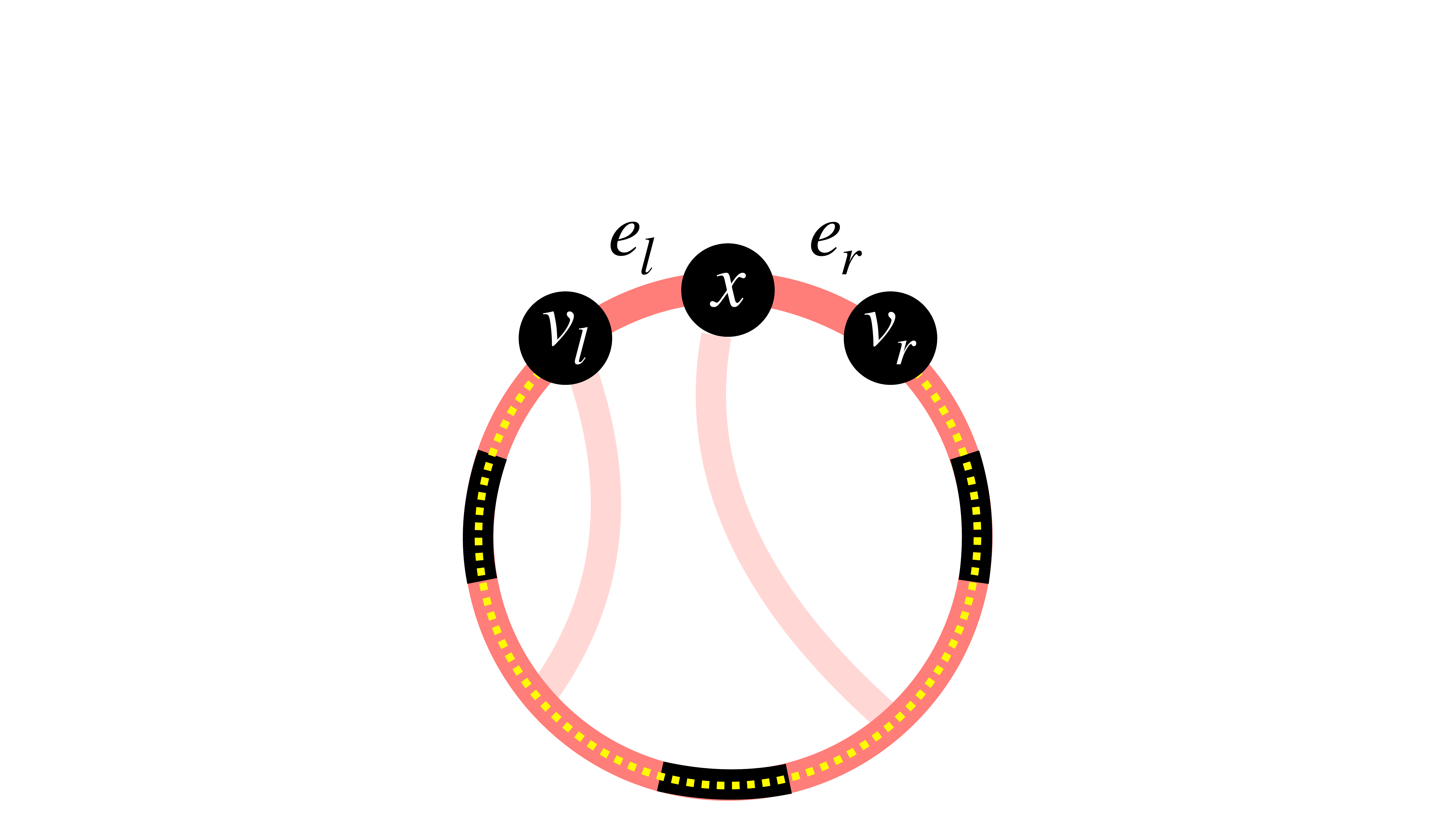}
        \caption{One local max.}\label{sfig:conBel1}
    \end{subfigure}\hfill
    \begin{subfigure}[b]{0.32\textwidth}
        \centering
        \includegraphics[width=\textwidth,trim=150mm 0mm 150mm 30mm, clip]{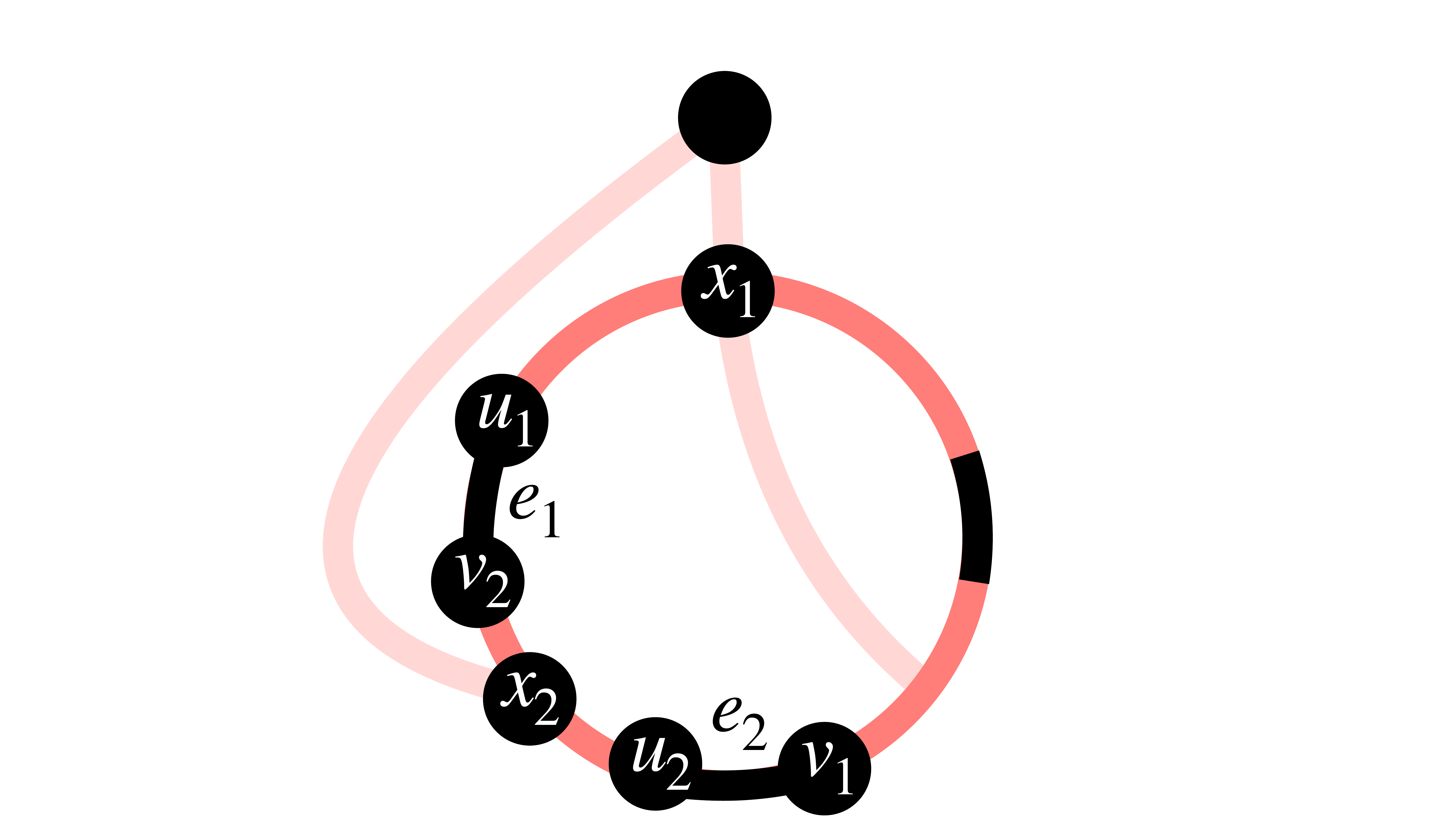}
        \caption{Two local maxes.}\label{sfig:conBel2}
    \end{subfigure}\hfill
    \begin{subfigure}[b]{0.32\textwidth}
        \centering
        \includegraphics[width=\textwidth,trim=150mm 0mm 150mm 30mm, clip]{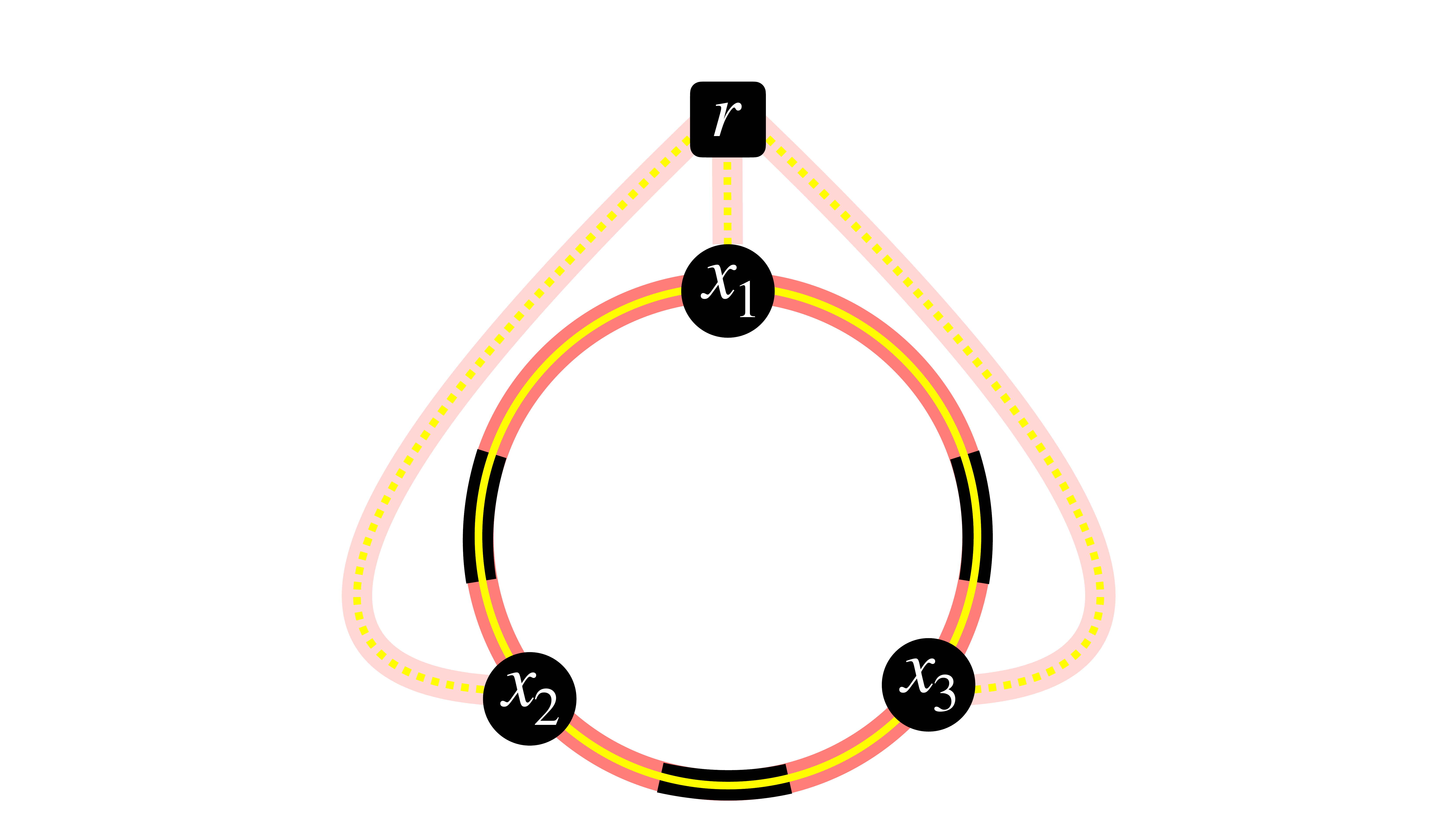}
        \caption{At least three local maxes.}\label{sfig:conBel3}
    \end{subfigure}
    \caption{The three cases of the proof of \Cref{lem:conBelGivesForest}. Edges in $C$ solid, edges outside of $C$ transparent. $\TBFS$ in pink and $E_c$ in black. In (a) we give the path between $v_l$ and $v_r$ contained in $\TBFS(x)\setminus\{x\}$ with a dotted yellow path. In (b) we illustrate give $l(e_1)=l(e_2)$ at the top. In (c) we highlight the cycle and the paths of the clawed cycle in solid and dotted yellow respectively.}
\end{figure}

\begin{lemma}\label{lem:conBelGivesForest}
    Suppose $G_U = (U, E_U)$ is a subgraph of $G$ which is connected below and where $|C_i \cap E_U| \leq 1$ for every $i$. Then $G_U$ is a forest.
\end{lemma}
\begin{proof}
    Assume for the sake of contradiction that $G_U$ has a cycle $C$. Since $G_U$ has at most one edge from each $C_i$, we know that  $|C_i \cap C| \leq 1$ for all $i$ and that $C$ contains at least one edge from $E_c$. Say that vertex $v \in C$ is a local max if no vertex in $C$ is an ancestor of $v$; clearly there is at least one such local max in $C$. We case on the number of local maxes in $C$. For each of our three cases we illustrate the contradiction we arrive at in Figures \ref{sfig:conBel1}, \ref{sfig:conBel2} and \ref{sfig:conBel3} respectively.
    
    \begin{enumerate}
        \item Suppose there is $1$ local max $x$ in $C$. Let $e_l=\{v_l, x\}$ and $e_r = \{v_r, x\}$ be the two edges of $C$ incident to $x$. Since every vertex in $C$ is a descendant of $x$ we know that $e_l$ and $e_r$ are in $\TBFS$ and are therefore connected below since otherwise we would have an edge in $E_c$ from a vertex to one of its descendants. However, it follows that the subgraph of $C$ connecting $v_r$ and $v_l$ which does not contain $x$---call it $P$---is contained in $\TBFS(x) \setminus \{x\}$. Thus $P$ along with $e_l$ and $e_r$ contradicts \Cref{lem:noFundCycle}.
        
        \item Suppose there are $2$ local maxes $x_1$ and $x_2$ in $C$. We will contradict the fact that $|C \cap C_i| \leq 1$ for every $i$. To do so, we first claim that there are distinct edges $e_1 = \{u_1, v_2\}, e_2 = \{u_2, v_1\} \in E_c \cap C$ where $u_1, v_1 \in \TBFS(x_1)$ and $v_2, u_2 \in \TBFS(x_2)$. To see this, notice that each subpath of $C$ in $\TBFS$ is contained in either $\TBFS(x_1)$ or $\TBFS(x_2)$ and so since $C$ is a cycle two such edges must exist. Next, notice that it therefore follows that $e_1$ and $e_2$ are lca-equivalent: $l(u_2, v_2) \in \TBFS(x_2)$ and $l(u_1, v_1) \in \TBFS(x_1)$ but $l(x_1, x_2) = l(e_1) = l(e_2)$ where $x_1, x_2 \prec l(x_1, x_2)$ by our assumption that $x_1$ and $x_2$ are local maxes. Thus, we know that $e_1$ and $e_2$ lie in the same $C_i$, a contradiction to the fact that $|C \cap C_i| \leq 1$ for every $i$.
        
        
        \item Suppose there are at least $3$ local maxes; let $x_1$, $x_2$ and $x_3$ be an arbitrary but distinct three of these maxes. Then cycle $C$ along with paths $P_1 := \TBFS(r, x_1)$, $P_2 := \TBFS(r, x_2)$ and $P_3 := \TBFS(r, x_3)$ form a clawed cycle since each $x_i$ is an ancestor of or unrelated to every other vertex in $C$. \qedhere
    \end{enumerate}
\end{proof}

\subsubsection{Constructing $\hat{\mcH}$}

We now define subgraph $\hat{H}_i$ and argue that $\hat{H}_i$ is indeed a hammock. Our initial hammocks will connect all ``lca-free minimal'' cross edge paths.

\begin{definition}[Minimal, lca-Free Cross Edge Paths]\label{dfn:lcaFree}
    We say that a cross edge path $P \subseteq G$ is lca-free if $l(e_f), l(e_l) \not \in P$. We say that $P$ is minimal if only its first and last edges are in $E_c$.
\end{definition}

\begin{definition}[$\hat{H}_i$]
    $\hat{H}_i$ is the subgraph of $G$ induced by all lca-free minimal cross edge paths between edges in $C_i$.
\end{definition}
As a reminder we let $\hat{T}_i$, $\hat{T}_i'$, $\hat{r}_i$ and $\hat{r}_i'$ refer to the two hammock trees and hammock roots of hammock $\hat{H}_i$. Similarly, we let $\hat{\mcH} := \{\hat{H}_i\}_i$ be the collection of all of our initial hammocks and let $\hat{H}$ be its corresponding induced subgraph of $G$. We illustrate $\hat{\mcH}$ in \Cref{sfig:initHams}.
\begin{figure}
    \centering
    \begin{subfigure}[b]{0.49\textwidth}
        \centering
        \includegraphics[width=\textwidth,trim=0mm 0mm 0mm 0mm, clip]{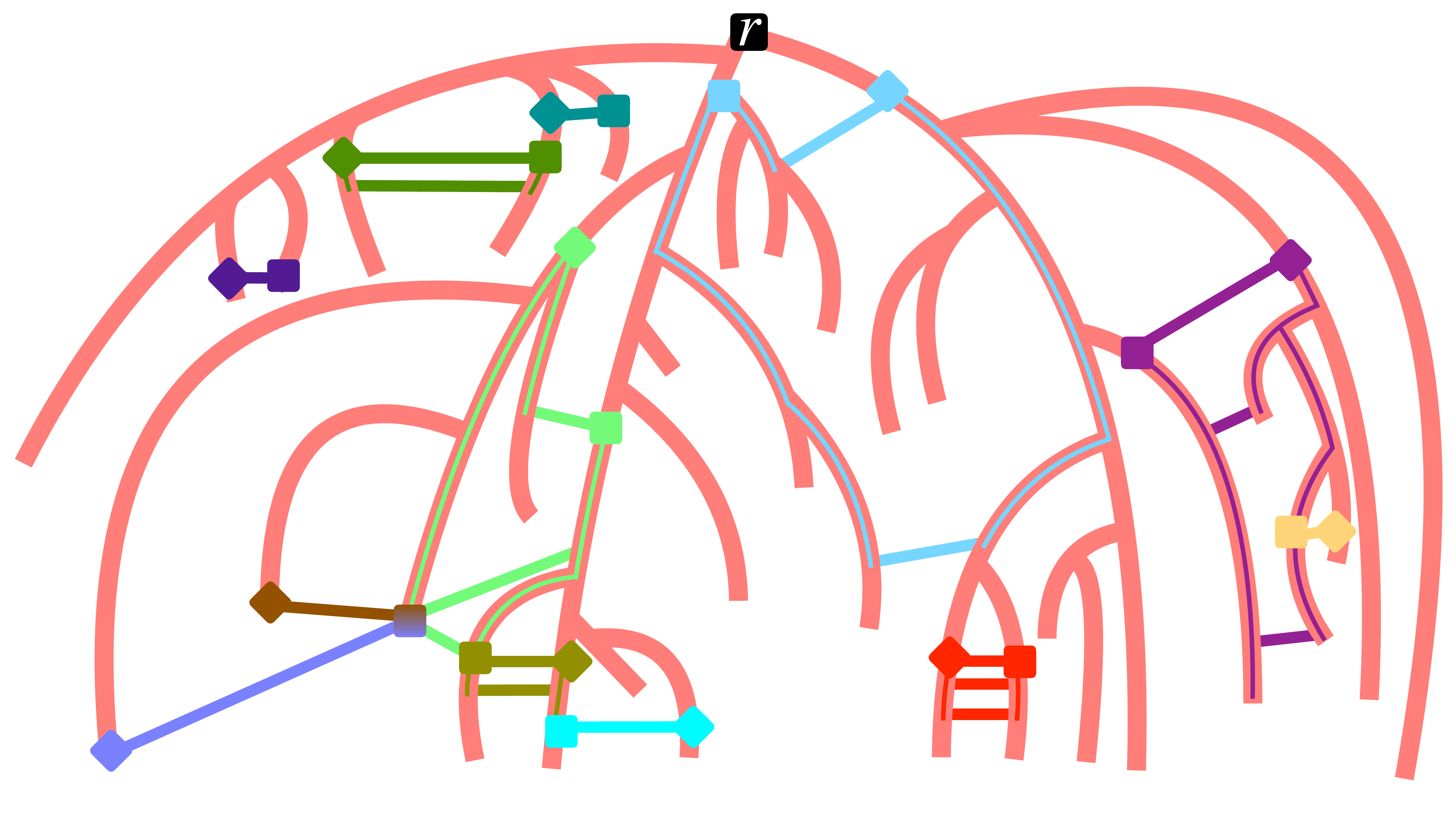}
        \caption{Initial hammocks $\hat{\mcH}$.}\label{sfig:initHams}
    \end{subfigure}
    \hfill
    \begin{subfigure}[b]{0.49\textwidth}
        \centering
        \includegraphics[width=\textwidth,trim=0mm 0mm 0mm 0mm, clip]{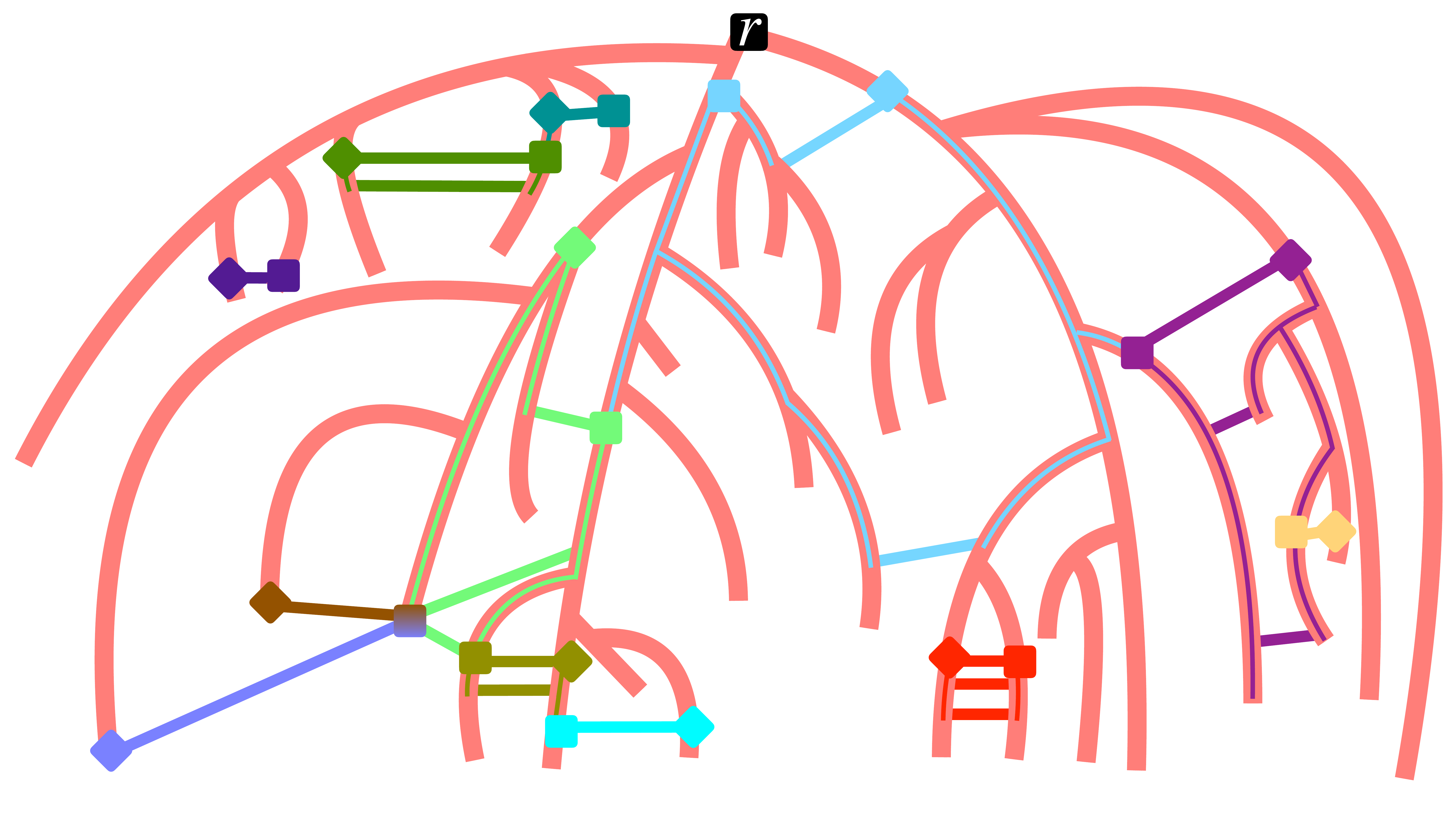}
        \caption{Extending $\hat{\mcH}$ to $\bar{\mcH}$.}\label{sfig:extendHams}
    \end{subfigure}
    \caption{An illustration of our initial hammocks $\hat{\mcH} = \{\hat{H}_i\}_i$ and how we extend them to our final hammocks $\bar{\mcH}$. Roots and edges of initial hammocks colored according to $i$. Notice that one vertex is the root of two hammocks (and so colored with two colors).}
\end{figure}

\begin{lemma}\label{lem:hatHHam}
    Each $\hat{H}_i \in \hat{\mcH}$ is a hammock.
\end{lemma}
\begin{proof}
    Consider a fixed $\hat{H}_i$. Fix an arbitrary edge $e_0 = (u_0, v_0) \in C_i$. Now consider an arbitrary $e = \{u, v\} \in C_i$. By definition of $C_i$ we know that (WLOG) $u_0$ and $u$ have an lca which is a descendant of $l(e) = l(e')$; thus, there is a minimal lca-free cross edge path of the form $(v_0, u_0, \ldots, u, v)$ which connects $e_0$ and $e$ which is included in $\hat{H}_i$. Symmetrically, there is a minimal cross edge connecting path of the form $(u_0, v_0, \ldots, v, u)$. We therefore know that $u$ and $v$ are in the same connected components in $\hat{H}_i \setminus E_c$ as $u_0$ and $v_0$ respectively. Moreover, since $\hat{H}_i \setminus E_c \subseteq \TBFS$, it follows that $\hat{H}_i \setminus E_c$ consists of at most two trees (the tree containing $u_0$ and the tree containing $v_0$). Even stronger, since $l(C_i)$ is not contained in any of the constituent paths of $\hat{H}_i$, we know that $\hat{H}_i \setminus E_c$ consists of exactly two trees and these trees are vertex-disjoint. We let these two trees $T$ and $T'$ be the hammock trees of $\hat{H}_i$. Next, to argue that $\hat{H}_i$ is a hammock we must argue that every cross edge between $T$ and $T'$ is included in $\hat{H}_i$. However, any such edge $e' = \{u', v'\}$ has $l(e') = l(C_i)$. Moreover, since $u', v' \neq l(C_i)$, we have that (WLOG) $l(u_0, u')$ and $l(v_0, v')$ are descendants of $l(C_i)$ and so $e' \in C_i$ and therefore $e' \in \hat{H}_i$. Lastly we must show that $\highV(T)$ and $\highV(T')$ are unrelated. However, notice that if (WLOG) $\highV(T')$ where a descendant of $\highV(T)$ then there would have been be a minimal lca-free cross edge path included in $\hat{H}_i$ which would have connected $T$ and $T'$, contradicting the fact that they are vertex disjoint.
\end{proof}

As our $H_i$s must ultimately be edge-disjoint, we will need that our $\hat{H}_i$s are edge-disjoint.

\begin{lemma}\label{lem:hatsDisjoint}
    $\hat{H}_i$ and $\hat{H}_j$ are edge-disjoint for $i \neq j$.
\end{lemma}
\begin{proof}
    By \Cref{lem:hatHHam} we may assume that $\hat{H}_i$ and $\hat{H}_j$ are hammocks. Suppose for the sake of contradiction that $E(\hat{H}_i) \cap E(\hat{H}_j) \neq \emptyset$ for $i \neq j$ and let $e = \{u,v\}$ be an edge included in $\hat{H}_i$ and $\hat{H}_j$. We illustrate the resulting situation, described below, in \Cref{sfig:edgeDis1}. We will argue that this situation leads to a clawed cycle, illustrated in \Cref{sfig:edgeDis2}. By definition of $\hat{H}_i$ and $\hat{H}_j$, such an edge must an edge in $\TBFS$ and so we assume WLOG $v \prec u$.
    
    We claim that in this case $l(C_i)$ and $l(C_j)$ are related; to see this notice that $l(C_i)$ and $l(C_j)$ must both be ancestors of or equal to $u$ by virtue of the fact that $e \in \hat{H}_i, \hat{H}_j$ and that $\hat{H}_i$ and $\hat{H}_j$ are hammocks. WLOG we assume that $C_j \preceq C_i$. Since $C_j$ contains $e$ we know that $|C_j| \geq 2$ and, in particular, that $\hat{H}_j$ contains a hammock-fundamental cycle containing $e$; let $F_j$ be this cycle with edges $e_j = \{u_j, v_j\}$ and $e_j' = \{u_j', v_j'\}$ in $E_c$ where we imagine WLOG that $u_j$ and $u_j'$ are in the same hammock tree of $\hat{H}_j$ as $e$; we will construct a clawed cycle with this cycle. Let $P_1: = \TBFS(r, l(u_j, u_j'))$ and $P_2 := \TBFS(r, l(v_j, v_j'))$ be our first two paths to $F_j$. By definition of a hammock we know that $l(u_j, u_j') \neq l(v_j, v_j')$; also notice that $u \preceq l(u_j, u_j')$ and $l(v_j, v_j') \not \in \TBFS(l(u_j, u_j'))$.

    Next, notice that since $e \in \hat{H}_i$, there must be some edge $e_i = \{u_i, v_i\} \in C_i$ where $v_i \in \TBFS(v)$. Let $P_3'$ be the path from $v_i$ to $F_j$ along $\TBFS(v_i, v)$ and let $P_3 := \TBFS(r, u_i) \oplus e_i \oplus P_3'$. 
    
    To show that $F_j$ along with $P_1$, $P_2$ and $P_3$ are a clawed cycle it suffices to argue that $P_3$ does not contain $l(u_j, u_j')$ or $l(v_j, v_j')$ and so we verify that none of $e_i$, $P_3'$ or $\TBFS(r, u_i)$ contain $l(u_j, u_j')$ or $l(v_j, v_j')$. 
    \begin{itemize}
        \item First, we argue that $l(u_j, u_j'), l(v_j, v_j') \not \in e_i$. Since $v_i \in \TBFS(v)$ and $l(v_j, v_j') \not \in \TBFS(v)$ we know that $v_i \neq l(v_j, v_j')$. Similarly, since $v_i \preceq v \prec u \preceq l(u_j, u_j')$, we know that $v_i \neq l(u_j, u_j')$. Now, consider $u_i$. It cannot be the case that $u_i \in \TBFS(l(u_j, u_j'))$ since otherwise we would have $l(C_i) \preceq l(u_j, u_j')$ and so $C_i \prec C_j$, a contradiction. It follows that $u_i \neq l(u_j, u_j')$. Similarly, we cannot have $u_i = l(v_j, v_j')$ since then we would have that $C_i = C_j$.
        \item Since $P_3' \subseteq \TBFS(v) \subseteq \TBFS(l(u_j, u_j')) \setminus \{l(u_j, u_j')\}$, we know that $P_3'$ contains neither $l(u_j, u_j')$ nor $l(v_j, v_j')$.
        \item Lastly, suppose for the sake of contradiction that $l(u_j, u_j') \in \TBFS(r, u_i)$. It follows that $l(C_i) \preceq l(u_j, u_j')$, contradicting the fact that $l(u_j, u_j') \prec l(C_j) \preceq l(C_i)$. Lastly, suppose for the sake of contradiction that $l(v_j, v_j') \in \TBFS(r, u_i)$. It follows that $l(C_i) = l(v_i, u_i) \preceq l(l(u_j, u_j'), l(v_j, v_j')) = l(C_j)$ and since $l(C_j) \preceq l(C_i)$, it follows that $l(C_i) = l(C_j)$. Since $u_i \in \TBFS(l(u_j, u_j'))$ and $v_i \in l(v_j, v_j')$ we then would have $C_i = C_j$, a contradiction.\qedhere
    \end{itemize}
     
%
\end{proof}

\begin{figure}
    \centering
    \begin{subfigure}[b]{0.49\textwidth}
        \centering
        \includegraphics[width=\textwidth,trim=0mm 0mm 0mm 0mm, clip]{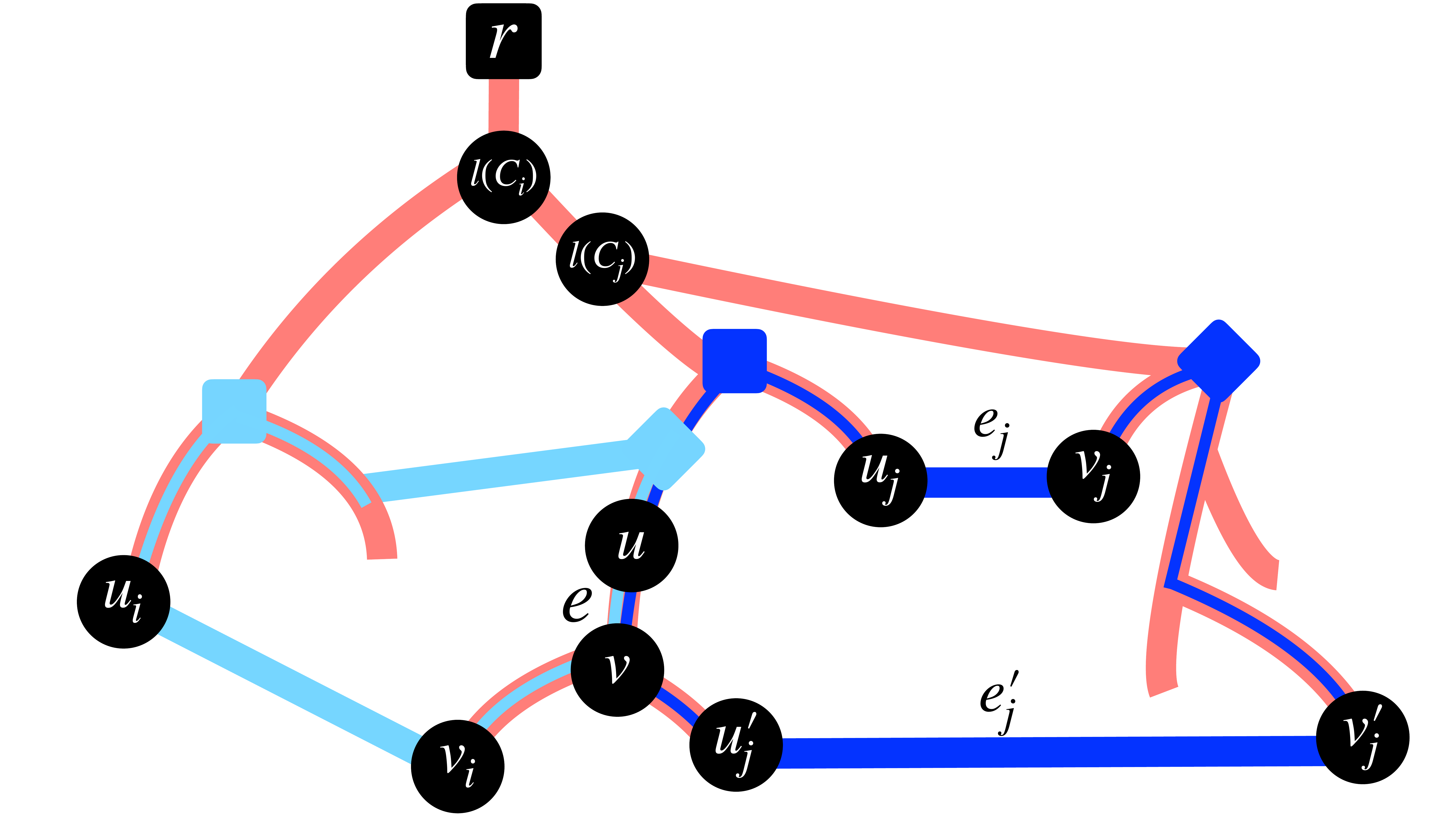}
        \caption{\Cref{lem:hatsDisjoint} proof setup.}\label{sfig:edgeDis1}
    \end{subfigure}
    \hfill
    \begin{subfigure}[b]{0.49\textwidth}
        \centering
        \includegraphics[width=\textwidth,trim=0mm 0mm 0mm 0mm, clip]{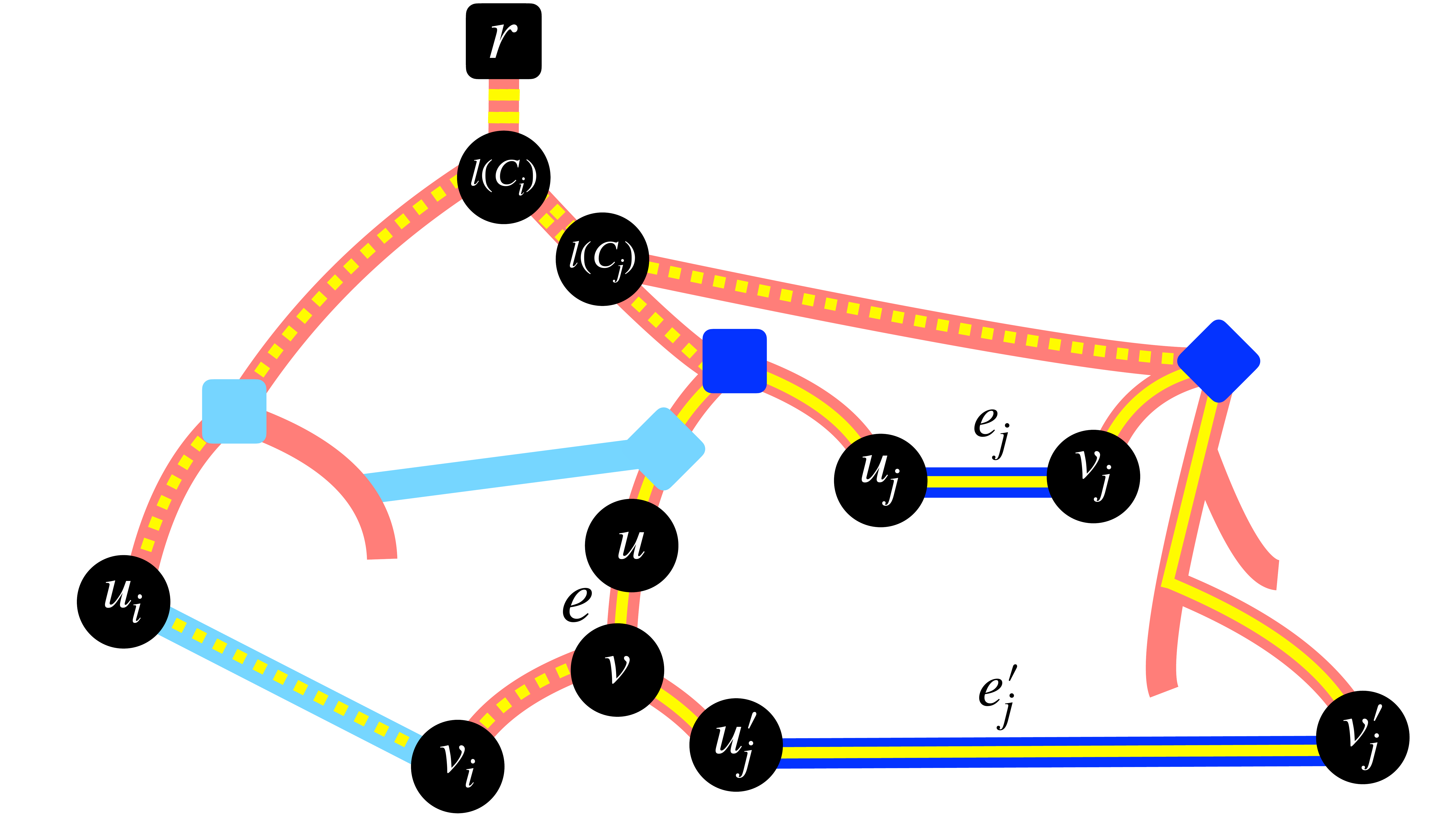}
        \caption{\Cref{lem:hatsDisjoint} proof contradiction.}\label{sfig:edgeDis2}
    \end{subfigure}
    \caption{The contradiction in the proof of \Cref{lem:hatsDisjoint}. In (a) we illustrate $\bar{H}_i$ in light blue and $H_j$ in dark blue. In (b) we highlight $F_j$ in solid yellow and $P_1$, $P_2$ and $P_3$ in dotted yellow.}
\end{figure}

Lastly, as alluded to above, we must establish that each of our $\hat{H}_i$s are connected below.

\begin{lemma}\label{lem:thirdPath}
    Every $\hat{H}_i$ is connected below for every $i$.
\end{lemma}
\begin{proof}
    Consider an edge $e \in \hat{H}_i \cap \TBFS$. By definition of $\hat{H}_i$, we know that $e \in \hat{H}_i$ because there are edges $e_1 = \{u_1, v_1\}, e_2 = \{u_2, v_2\} \in C_i$ which are connected by an lca-free minimal cross edge path $P_{12}$ between $e_1$ and $e_2$ containing $e$. Thus, one of $e_1$ and $e_2$ has an endpoint in $\TBFS(v)$; WLOG suppose $e_1$ does. We may also assume that $e_1$ has at most one endpoint in $\TBFS(u)$ (and so has exactly one endpoint in $\TBFS(v)$) since if both endpoints of $e_1$ were in $\TBFS(u)$ then we would have $l(e_1) \in \TBFS(u)$ and so $P_{12}$ would contain $l(e_1)$, contradicting its lca-freeness. Thus, we assume that $v_1 \in \TBFS(v)$ but $u_1 \not \in \TBFS(u)$. Next, we let $P := \TBFS(r, u_1) \oplus e_1$; we know that $P$ is from $r$ to $\TBFS(v)$ and is internally vertex disjoint from $\TBFS(u)$ since $u_1 \not \in \TBFS(u)$, thereby demonstrating that $e$ is connected below.
\end{proof}

Notice that it immediately follow from \Cref{lem:conBelGivesForest}, \Cref{lem:hatHHam}, \Cref{lem:hatsDisjoint} and \Cref{lem:thirdPath} that $\hat{\mcH}$ is a forest of hammocks.

\subsection{Extending $\hat{\mcH}$ to $\bar{\mcH}$ by Hammock-Joining Paths}

We now describe how we extend $\hat{\mcH}$ to $\bar{\mcH}$ via what we call ``hammock-joining'' paths. We do so to ensure that our final forest of hammocks contains all shortest cross edge paths. We illustrate this process in \Cref{sfig:extendHams}.

\begin{definition}[Hammock-Joining Paths]\label{dfn:hamJoin}
    We say that path $P \subseteq \TBFS$ is a hammock-joining path if
    \begin{enumerate}
        \item $P$ is between distinct hammocks $\hat{H}_i, \hat{H}_j \in \hat{\mcH}$;
        \item $l(C_i), l(C_j) \not \in P$.
    \end{enumerate}
\end{definition}

We let $\mcPHJ$ be the subgraph induced by all hammock-joining paths for the rest of this section.

\begin{lemma}\label{lem:pathsConBel}
    Every $e \in \mcPHJ$ is connected below.
\end{lemma}
\begin{proof}
    Consider an edge $e = \{u, v\}$ where $u$ is the parent of $v$ in $\TBFS$ which is part of a hammock-joining path $P_{ij}$ between $\hat{H}_i$ and $\hat{H}_j$. 
    We will construct a path $P$ from $r$ to $\TBFS(v)$ where $P \cap \TBFS(u) \subseteq \TBFS(v)$ as required by the definition of connected below. WLOG we assume $P_{ij}$'s endpoint in $\hat{H}_i$ is in $\TBFS(u)$. Since every vertex in $\hat{H}_i$ is a descendant of $l(C_i)$ and $l(C_i) \not \in P$ since $P$ is hammock-joining, we know that $l(C_i) \not \in \TBFS(u)$. Notice that, by construction, each leaf of each of $\hat{H}_i$'s hammock trees is incident to an edge of $C_i$. Thus, we know that there is some $e_i = \{u_i, v_i\} \in C_i$ where WLOG $v_i \in \TBFS(u)$ but $u_i \not \in \TBFS(u)$. Then, we can let $P$ be $\TBFS(r, u_i)$. We know that $P$ is internally vertex-disjoint from $\TBFS(u)$ since $u_i \not \in \TBFS(u)$. Thus, $P$ is from $r$ to $F$ and $P \cap \TBFS(u) \subseteq \TBFS(v)$ as required.
\end{proof}

Next, we describe how we will assign each collection of connected components in $\mcPHJ$ to one of our hammocks. We always assign such a component to an incident hammock, hence the following definition.

\begin{definition}[$I(T)$]
    For a connected component $T$ of $\mcPHJ \setminus \hat{H}$, we let $I(T) := \{i : V(\hat{H}_i) \cap V(T) \neq \emptyset\}$ be the indices of initial hammocks which intersect $T$.
\end{definition}

The following definition formalizes the notion of a valid assignment of connected components of $\mcPHJ$ to hammocks.

\begin{definition}[Valid Assignment of Connected Components]
    Let $\pi$ be a mapping from components of $\mcPHJ \setminus \hat{H}$ to non-negative integers. Then we say that $\pi$ is an assignment of the connected components of $\mcPHJ \setminus \hat{H}$ if it maps each component $T \subseteq \mcPHJ \setminus \hat{H}$ to an index in $I(T)$. We say that $\bar{\mcH} = \{\bar{H}_i\}_i$ results from $\pi$ if $\bar{H}_i = \hat{H}_i \cup \bigcup_{T : \pi(T) = i} T$. We say that $\pi$ is valid if it holds that when $\pi(T) = i$ then $l(C_i) \not \in T$ for every component $T \subseteq \mcPHJ \setminus \hat{H}$.
\end{definition}

We now show that, provided connected components of $\mcPHJ \setminus \hat{H}$ are assigned in a valid way, the result is a forest of hammocks.

\begin{lemma}\label{lem:allAssignGiveForest}
    Let $\pi$ be a valid assignment of the connected components of $\mcPHJ \setminus \hat{H}$. Then if $\bar{\mcH}$ results from $\pi$ then $\bar{\mcH}$ is a forest of hammocks.
\end{lemma}
\begin{proof}
    We let $\bar{H}$ be the subgraph induced by $\bar{\mcH}$ for the remainder of this proof. We first note that by \Cref{lem:thirdPath} and \Cref{lem:pathsConBel}, $\bar{H}$ is connected below; we will use this fact several times in this proof.
    
    We begin by verifying that every $\bar{H}_i$ is indeed a hammock. By \Cref{lem:hatHHam} $\bar{H}_i$ is a hammock when we initialize it to $\hat{H}_i$. Next, we claim that adding connected components of $\mcPHJ \setminus \hat{H}$ to $\bar{H}_i$ does not violate its hammockness. First, notice that $\bar{H}$ does not contain any fundamental cycles of $\TBFS$ since $\bar{H}$ is connected below so such a cycle would be a connected below subgraph violating \Cref{lem:noFundCycle}. Thus, since each of the connected components of $\mcPHJ \setminus \hat{H}$ we add to $\bar{H}_i$ intersect with $\hat{H}_i$ on some vertex, we know $\bar{H}_i$ continues to be spanned by two subtrees of $\TBFS$. To verify that $\bar{H}_i$ is a hammock we must also show that the edges between these two trees in $E_c$ are exactly $C_i$; since $C_i \subseteq \hat{H}_i$, it suffices to show that no additional edges from $E_c$ become incident to the two hammock trees of $\bar{H}_i$ as we add connected components of $\mcPHJ \setminus \hat{H}$ to $\bar{H}_i$. However, notice that the existence of such an edge would give us a cycle in $\bar{H}$ incident to two lca-equivalence classes, contradicting \Cref{lem:conBelGivesForest}. Lastly, it remains to argue that the roots of $\bar{H}_i$s hammock trees are unrelated. This is immediate from our assumption that $\pi$ is valid---i.e.\ if $\pi(T) = i$ then $l(C_i) \not \in T$  for every component $T \subseteq \mcPHJ \setminus \hat{H}$---since the roots of $\bar{H}_i$'s hammock trees would only become related if $l(C_i)$ were included in some $T$ assigned to $i$.
    
    Next, notice that the $\bar{H}_i$s are edge-disjoint by construction. In particular, the $\hat{H}_i$s are disjoint by \Cref{lem:hatsDisjoint}. Moreover, each connected component of $\mcPHJ \setminus \hat{H}$ is pair-wise edge-disjoint by definition and also disjoint from any $\hat{H}_i$ by construction. Thus, in constructing $\bar{H}_i$ by adding connected components from $\mcPHJ \setminus \hat{H}$ to  $\hat{H}_i$, our resulting $\bar{H}_i$s must be edge-disjoint.
    
    Lastly, we verify our cycle property: we must show that any cycle $C$ in $\bar{H}$ satisfies $|\bar{H}_i : E(C) \cap \bar{H}_i \neq \emptyset| = 1$. We use a ``shortcutting'' argument---which we illustrate in \Cref{fig:sCutting}---and \Cref{lem:conBelGivesForest}. In particular, suppose that $\bar{H}$ contains a cycle $C$ where $|\bar{H}_i : E(C) \cap \bar{H}_i \neq \emptyset| \geq 2$. Then, we claim that $\bar{H}$ also contains a cycle $C'$ where $|\bar{H}_i : E(C') \cap \bar{H}_i \neq \emptyset| \geq 2$ \emph{and} $|C_i \cap E(C)|  \leq 1$ for all $i$. To see this, notice that, since our $\bar{H}_i$s are edge-disjoint we may ``shortcut'' $C$ through the trees of $\bar{H}_i$. In particular, let $\bar{H}_i$ be one of our hammocks where $|C_i \cap E(C)| \geq 2$ and let $x$ and $y$ be the first and last vertices of $\bar{H}_i$ visited by $C$ (for some arbitrary cyclic ordering of the edges in $C$). Then, if $x$ and $y$ are in the same hammock tree of $\bar{H}_i$ then we replace the portion of $C$ between $x$ and $y$ with $\TBFS(x,y)$. If $x$ and $y$ are in different hammock trees, we replace the portion of $C$ between $x$ and $y$ with an arbitrary path in $\bar{H}_i$ which uses at most one edge of $C_i$. Doing these replacements does not affect $|\bar{H}_i : E(C) \cap \bar{H}_i \neq \emptyset|$ and reduces the number of $\bar{H}_i$ with $|C_i \cap E(C)| \geq 2$ by at least $1$; thus, after iterating a finite number of times we produce our desired $C'$.
    
    The existence of $C'$ allows us to arrive at a contradiction. In particular, since $\bar{H}$ is connected below we know that $C'$ is connected below where $|C_i \cap C'| \leq 1$; but $C'$ is a cycle, contradicting \Cref{lem:conBelGivesForest}.
\end{proof}

\begin{figure}
    \centering
    \begin{subfigure}[b]{0.49\textwidth}
        \centering
        \includegraphics[width=\textwidth,trim=0mm 0mm 0mm 0mm, clip]{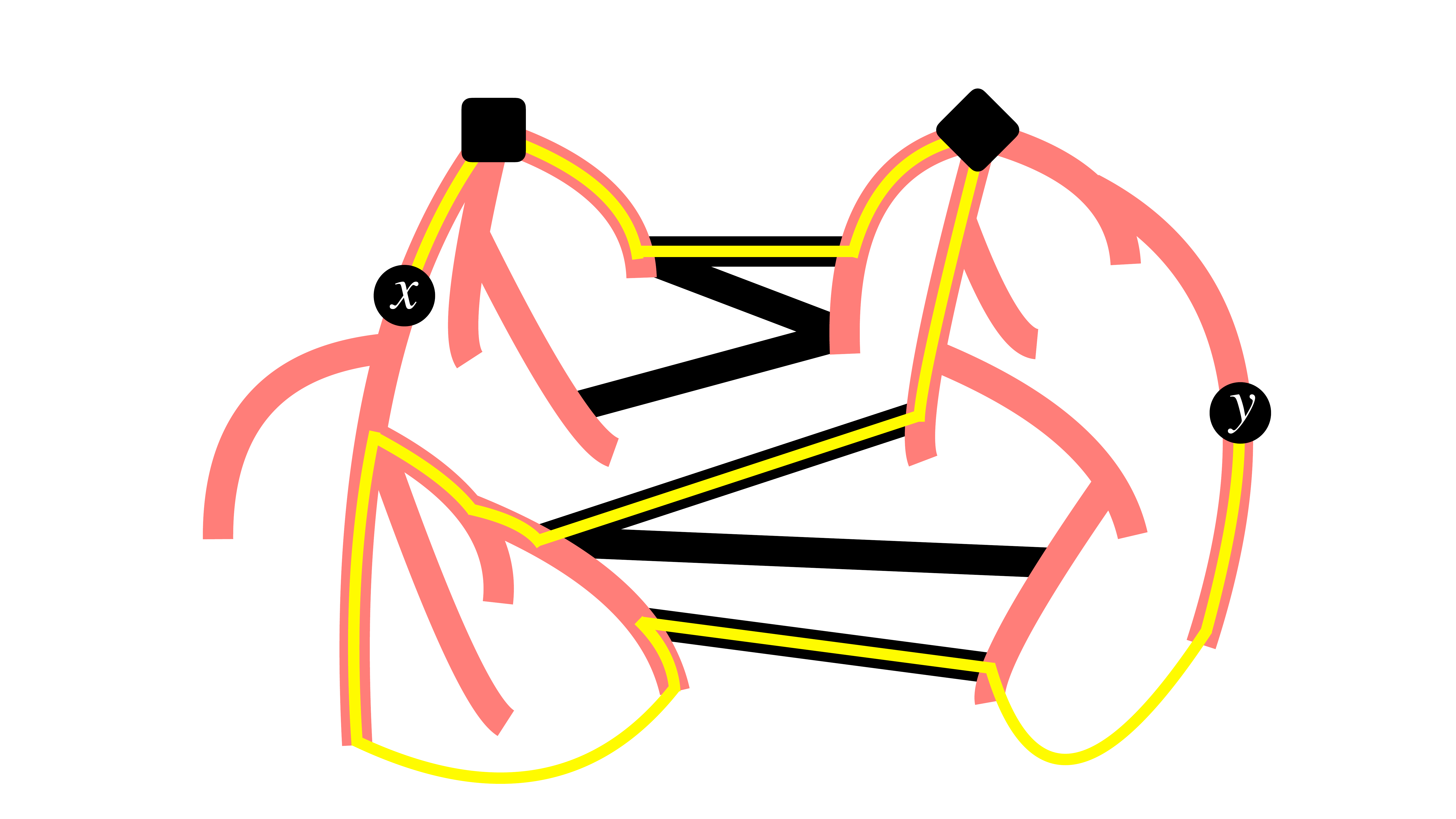}
        \caption{$C$ in $\bar{H}_i$.}\label{sfig:scut1}
    \end{subfigure}
    \hfill
    \begin{subfigure}[b]{0.49\textwidth}
        \centering
        \includegraphics[width=\textwidth,trim=0mm 0mm 0mm 0mm, clip]{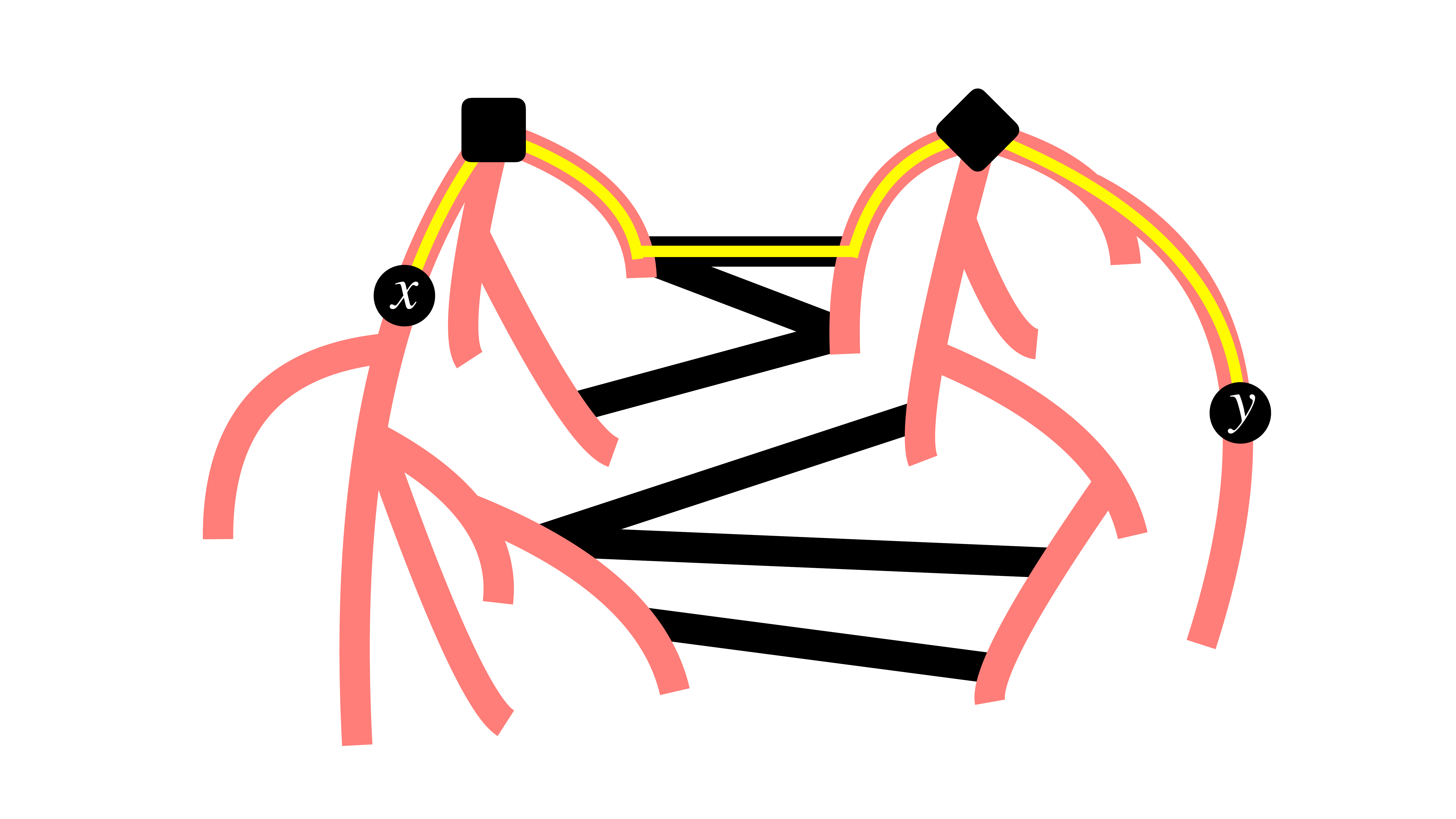}
        \caption{Shortcutting $C$ in $\bar{H}_i$.}\label{sfig:scut2}
    \end{subfigure}
    \caption{An illustration of how in the proof of \Cref{lem:allAssignGiveForest} we may assume that $|C \cap C_i| \leq 1$ by shortcutting $C$ in $\bar{H}_i$. We highlight $C$ in yellow both before and after shortcutting.}\label{fig:sCutting}
\end{figure}

Henceforth we will assume that the root hammock of each tree of hammocks $\mcT$ in $\bar{\mcH}$ is $\max_{i : \bar{H}_i \in \mcT}h(C_i)$ where we break ties arbitrarily. As a reminder, $h(C_i)$ is the height of $l(C_i)$ in $\TBFS$.

\begin{lemma}\label{lem:someParent}
    For any connected component $T$ of $\mcPHJ \setminus \hat{H}$ there is some $i \in I(T)$ such that for any valid assignment $\pi$ of the connected components of $\mcPHJ \setminus \hat{H}$ with $\bar{\mcH}$ as the resulting forest of hammocks we have $\bar{H}_j \preceq_{\bar{\mcH}} \bar{H}_i$ for all $j \in I(T)$.
\end{lemma}
\begin{proof}
    Let $\bar{H} := \hat{H} \cup \mcPHJ$ where $\hat{H} = G[\hat{\mcH}]$. Fix a connected component of $\bar{H}$; notice that we always choose the same root hammock for this component regardless of $\pi$; we let $\bar{H}_k$ be said root hammock. Fix a component $T$ of $\mcPHJ \setminus \hat{H}$ and let $I := I(T)$ for the rest of this proof. 
    
    If $\bar{H}_k$ is incident to $T$ then the claim trivially holds since $\bar{H}_j \preceq_{H} \bar{H}_k$ for any assignment and $j \in I \setminus \{k\}$ and so we will assume that $k \not \in I$. 
    
    Fix an arbitrary valid assignment $\pi$ of the connected components of $\mcPHJ \setminus \hat{H}$, let $\bar{\mcH}$ be the resulting forest of hammocks as per \Cref{lem:allAssignGiveForest} and say that $i$ is a local max with respect to $\pi$ if there is no $j \in I \setminus \{i\}$ where $\bar{H}_i$ is a descendant of $\bar{H}_j$ in $\bar{\mcH}$. Notice that there is at least $1$ local max since, by definition of $\mcPHJ$, there are at least two distinct hammocks among $\bar{\mcH}$ with a vertex in $T$. To prove our claim it suffices to show that there is $1$ local max under each assignment and this local max is always the same.
    
    We claim that the number of local maxes with respect to $\pi$ is at most $1$. To see this, assume for the sake of contradiction that there are $2$ local maxes $\bar{H}_i$ and $\bar{H}_j$ where $i, j \in I$. Since $\bar{H}_i$ and $\bar{H}_j$ are both local maxes neither of which is the root hammock $\bar{H}_k$ then by virtue of $\bar{\mcH}$ being a forest of hammocks we know that any path from $\bar{H}_i$ to $H_j$ in $G[\bar{\mcH}]$ must contain at least one vertex of the parent of $\bar{H}_i$ as per \Cref{lem:pathStruct}. On the other hand, since $i, j \in I$, we know there is a path $P \subseteq  T$ between $\bar{H}_i$ and $\bar{H}_j$. Letting $\bar{H}_{i'}$ be the parent of $\bar{H}_i$ in $\bar{\mcH}$, it follows that $V(\bar{H}_{i'}) \cap P \neq \emptyset$ and so $i' \in I$, contradicting the fact that $\bar{H}_i$ is a local max.
    
    Thus, for any assignment we know that the number of local maxes is $1$. We proceed to show that this is always the same local max. In particular, we show that if $\bar{H}_i$ is the local max under some assignment $\pi$, then $\bar{H}_i$ is the local max under any other assignment $\pi'$. Suppose for the sake of contradiction that $i$ is the local max under $\pi$ but $j$ for $j \neq i$ is the local max under some other $\pi'$. Let $\bar{\mcH}$ and $\bar{\mcH}'$ be the forest of hammocks resulting from $\pi$ and $\pi'$ as per \Cref{lem:allAssignGiveForest}. Let $\bar{H}_i$ and $\bar{H}_i'$ be $i$'s hammock in $\bar{\mcH}$ and $\bar{\mcH}'$ and let $\bar{H}_j$ and $\bar{H}_j'$ be $j$'s hammock in $\bar{\mcH}$ and $\bar{\mcH}'$. We emphasize that $\bar{H}_i$ and $\bar{H}_j'$ need not be edge-disjoint since they are hammocks in two different forests of hammocks. 
    
    Now by virtue of the fact that $i$ is a local max in $\bar{\mcH}$ and $i \neq k$, we know by \Cref{lem:pathStruct} that every path between $\bar{H}_k$ and $\bar{H}_j$ contains an edge of $\bar{H}_i$; it follows that every path between $\hat{H}_k$ and $\hat{H}_j$ contains an edge of $\bar{H}_i$. We additionally make the stronger claim that every path between $\hat{H}_k$ and $\hat{H}_j$ contains an edge of $\hat{H}_i$. To see this, suppose for the sake of contradiction that there was a path $P$ between $\hat{H}_k$ and $\hat{H}_j$ whose only edges in $\bar{H}_i$ are contained in $\bar{H}_i \setminus \hat{H}_i \subseteq \mcPHJ \setminus \hat{H}$. Let $\bar{H}_{i'}$ be the parent of $\bar{H}_i$ in $\bar{\mcH}$ and let $P'$ be the subpath of $P$ restricted to $\bar{H}_i$ where $P'$ is between vertices $u$ and $v$ where $u \in \bar{H}_{i'}$. Since $u \in \bar{H}_{i'}$ and $P' \subseteq \mcPHJ \setminus \hat{H}$ it follows that $i' \in I$, contradicting our assumption that $i$ is a local max under $\pi$. Thus, indeed, every path between $\hat{H}_k$ and $\hat{H}_j$ contains an edge of $\hat{H}_i$.
    
    On the other hand, by virtue of the fact that $j$ is a local max in $\bar{\mcH}'$, we know by \Cref{lem:pathStruct} that in $G[\bar{\mcH}']$ there is a path from $\bar{H}_k'$ to $\bar{H}_j'$ which does not have an edge in $\bar{H}_i'$. Extending this path through $\bar{H}_k'$ and $\bar{H}_j'$ on either end, we have that there is a path from $\hat{H}_k$ to $\hat{H}_j$ which does not have an edge in $\bar{H}_i'$ and therefore no edge in $\hat{H}_i$. However, this contradicts the above claim that every path between $\hat{H}_k$ and $\hat{H}_j$ contains an edge of $\hat{H}_i$.
\end{proof}

\begin{lemma}\label{lem:atLeastOneVal}
    There is at least one valid assignment of the components of $\mcPHJ \setminus \hat{H}$.
\end{lemma}
\begin{proof}
    Fix a connected component $T$ of $\mcPHJ \setminus \hat{H}$. Recall that $T$ is a connected subtree of $\TBFS$. Let $x$ be the highest vertex in $\TBFS$ in $T$. Since $x$ is included in $\mcPHJ$ it must lie on a hammock-joining path $P_{ij}$ between some $\hat{H}_i$ and $\hat{H}_j$. We know that neither $l(C_i)$ nor $l(C_j)$ lie on $P_{ij}$ and so neither $l(C_i)$ nor $l(C_j)$ are in $T$. Thus, assign $T$ to $i$. Doing so for each such $T$ results in a valid assignment.
\end{proof}

By \Cref{lem:atLeastOneVal} there is at least one valid assignment of the components of $\mcPHJ \setminus \hat{H}$. It follows that by \Cref{lem:someParent} there is some valid assignment $\bar{\pi}$ which by \Cref{lem:allAssignGiveForest} results in a forest of hammocks $\bar{\mcH}$ where we have $\bar{H}_j \preceq_{\bar{\mcH}} \bar{H}_i$ for all $j \in I(T)$. We use this assignment in our construction. In particular, henceforth we let $\bar{\mcH}$ be the forest of hammocks which results from $\bar{\pi}$ and let $\bar{H} := G[\bar{\mcH}]$ be its induced subgraph. We will let $\bar{T}_i$, $\bar{T}_i'$, $\bar{r}_i$ and $\bar{r}_i'$ refer to the two hammock trees and hammock roots of $\bar{H}_i$ henceforth.

%

We proceed to argue that $\bar{\mcH}$ is a rooted forest of hammocks. To do so, we first prove two simple technical lemmas. Recall from \Cref{lem:artPoint} that two hammocks share at most one vertex in a rooted tree of hammocks.
 
\begin{lemma}\label{lem:parPath}
    Let $\mcT$ be a rooted tree of hammocks and let $\bar{H}_i$ be the parent of $\bar{H}_j$ in $\mcT$ where $v_{ij}$ is the one vertex in $V(\bar{H}_i) \cap V(\bar{H}_j)$. Then there is a path $P$ from $r$ to $v_{ij}$ where $P \cap \bar{H}_j = \{v_{ij}\}$.
\end{lemma}
\begin{proof}
    
    Let $\bar{H}_k$ be the root hammock of $\mcT$. Notice that it suffices to prove that there is a path from $r$ to $\bar{H}_k$ which is vertex disjoint from $\bar{H}_j$ since we can continue such a path through $G[\mcT]$ to $v_{ij}$ and such a path will only intersect $\bar{H}_{j}$ at $v_{ij}$ by \Cref{lem:pathStruct}.
    
    Let $P_0 := \TBFS(r, l(C_k))$. We claim that $V(\bar{H}_j) \cap P_0 = \emptyset$; to see this notice that if $\bar{H}_j$ had a vertex in $P_0$ then we would have $l(C_k) \prec l(C_j)$, contradicting our choice of $\bar{H}_k$.
    
    Next, let $\bar{r}_k$ and $\bar{r}_k'$ be the roots of $\bar{H}_k$'s trees and let $P := \TBFS(l(C_k), \bar{r}_k)$ and $P' := \TBFS(l(C_k), \bar{r}_k')$ be the paths from the lca of $\bar{H}_k$ to its roots. We claim that either $P \cap V(\bar{H}_j) = \emptyset$ or $P' \cap V(\bar{H}_j) = \emptyset$; this is sufficient to show our claim since $P_0$ concatenated with the non-intersecting path will give us our required path to $\bar{H}_k$.
    
    Suppose for the sake of contradiction that $P \cap V(\bar{H}_j) \neq \emptyset$ and $P' \cap V(\bar{H}_j) \neq \emptyset$. Let $T_j$ and $T_j'$ be the hammock trees of $\bar{H}_j$ with respective roots $\bar{r}_j$ and $\bar{r}_j'$. We cannot have that both $P$ and $P'$ have a vertex in $T_j$ since then it would follow that $l(C_k) \in V(T_j)$ and so $l(C_k) \prec l(C_j)$ (since $u \prec l(C_j)$ for every $u \in T_j$), again contradicting our choice of $\bar{H}_j$. Thus, it must be the case that (WLOG) $P$ has a vertex in $T_j$ and $P'$ has a vertex in $T_j'$. It follows that $l(C_k) = l(C_j)$. However, we claim that it also follows that $C_k = C_j$. In particular, any edges $e_k = \{u_k, v_k\} \in C_k$ and $e_j = \{u_j, v_j\} \in C_j$ then satisfy (WLOG) $l(u_k, u_j) \preceq \bar{r}_j$ and $l(v_k, v_j) \preceq \bar{r}_j'$. Since $l(C_j) = l(C_k)$, we know that $\bar{r}_j \neq l(C_j)$ and $\bar{r}_j' \neq l(C_j)$ and so, indeed, $e_j$ and $e_k$ are lca-equivalent. This is a contradiction since we have assumed that $j \neq k$ in assuming that $\bar{H}_j$ has a parent $\bar{H}_i$ in $\mcT$.
\end{proof}
    
\begin{lemma} \label{lem:fundHatCycle}   
    Fix $i$. Suppose there is some $x \in \hat{H}_i$ in WLOG $\hat{T}_i$ where $x \neq \hat{r}_i$.  Then $\hat{H}_i$ has a hammock-fundamental cycle $C$ containing $x$ where $x \neq \highV(V(C) \cap V(\hat{T}_i))$.
\end{lemma}
\begin{proof}
    By construction of $\hat{H}_i$ we know that every leaf of $\hat{T}_i$ is incident to an edge of $E_c(\hat{T}_i, \hat{T}_i')$ and that $\hat{r}_i$ has at least two children. Thus, there is a path which contains $x$ from $\hat{r}_i$ to some vertex $u$ where $\{u, u'\} \in E_c(\hat{T}_i, \hat{T}_i')$ and another edge-disjoint path in $\hat{T}_i$ from $\hat{r}_i$ to some $v$ where $\{v, v'\} \in E_c(\hat{T}_i, \hat{T}_i')$. Connecting these paths in $\hat{T}_i'$ gives the stated fundamental cycle.
\end{proof}

%
%
Concluding, we have that $\bar{\mcH}$ is indeed a rooted forest of hammocks containing all shortest cross edge paths.
\begin{lemma}\label{lem:rootedHDForest}
    $\bar{\mcH}$ is a rooted forest of hammocks where if $\bar{H}_i$ is a parent of $\bar{H}_j$ then $\bar{H}_i \cap \bar{H}_j = \{\bar{r}_j\}$. Furthermore, $G[\bar{\mcH}]$ contains all shortest cross edge paths.
\end{lemma}
\begin{proof}
    Since by \Cref{lem:allAssignGiveForest} $\bar{\mcH}$ is a forest of hammocks, it remains to show that $v_{ij} = \bar{r}_j$ for any $i,j$ pair where $\bar{H}_i$ is a parent of $\bar{H}_j$ in $\bar{\mcH}$ and, as before, $v_{ij} \in V(\bar{H}_i) \cap V(\bar{H}_j)$ is the one vertex in both $\bar{H}_i$ and $\bar{H}_j$. Let $T_j$ and $T_j'$ be the two hammock trees of $\bar{H}_j$ with roots $\bar{r}_j$ and $\bar{r}_j'$ and let $\hat{T}_j \subseteq T_j$ and $\hat{T}_j' \subseteq T_j'$ be the corresponding hammock trees of $\hat{H}_j$ with roots (i.e.\ highest vertices in $\TBFS$) $\hat{r}_j$ and $\hat{r}_j'$. We assume WLOG that $v_{ij} \in T_j$ and so it suffices to show that $v_{ij} = \bar{r}_j$.
    
    Assume for the sake of contradiction that $v_{ij} \neq \bar{r}_j$. We claim that it follows that $v_{ij} \neq \hat{r}_j$: if $v_{ij} = \hat{\bar{r}_j}$ then since $v_{ij} \in \bar{H}_i$, by how we construct $\bar{\mcH}$ we know that no component of $\mcPHJ \setminus \hat{H}$ incident to $v_{ij}$ would be assigned to $j$ (because it could have been assigned to $i$) and so if $v_{ij}$ were $\hat{r}_j$ (which is to say it is the highest vertex in $\hat{T}_j$), then it would also be the highest vertex in $T_j$ as no component of $\mcPHJ \setminus \hat{H}$ assigned to $\bar{H}_j$ could have a vertex higher than $v_{ij}$, contradicting our assumption that $v_{ij} \neq \bar{r}_j$. On the other hand it must be the case that $v_{ij} \in \hat{H}_j$: no component of $\mcPHJ \setminus \hat{H}$ which $v_{ij}$ is incident to will be assigned to $\bar{H}_i$ and so it cannot be the case that $v_{ij} \in \bar{H}_j \setminus \hat{H}_j$.
    
    Applying \Cref{lem:fundHatCycle}, it follows that $\bar{H}_j$ contains a fundamental cycle $C$ with highest vertices $v_j$ and $v_j'$ in $T_j$ and $T_j'$ respectively where $v_j, v_j' \neq v_{ij}$. Such a cycle gives a contradiction since it follows that $C$ along with $\TBFS(r, v_j)$, $\TBFS(r, v_j')$ and the path from $r$ to $v_{ij}$ as guaranteed by \Cref{lem:parPath} gives a clawed cycle.
    
    It remains to verify that $G[\bar{\mcH}]$ indeed contains every shortest cross edge path. It suffices to show that $G[\bar{\mcH}]$ contains all lca-free cross edge paths (recall that an lca-free cross edge path is one between two cross edges which contains neither of the cross edges' lcas) since every shortest cross edge path is lca-free. Let $P$ be an arbitrary lca-free cross edge path. 
    \begin{enumerate}
        \item If $P$ is between two edges $e,e'$ where $e, e' \in C_i$ for some $i$ then we know that $E(P) \subseteq \hat{H}_i$ by definition of $\hat{H}_i$. 
        \item On the other hand, suppose $P$ is between $e \in C_i$ and $e' \in C_j$ for $i \neq j$. Let $P'$ be the edges of $P$ which are not in $\hat{H}_i$ or $\hat{H}_j$. $P'$ must be a connected subpath by definition of $\hat{H}_i$ and $\hat{H}_j$ and so $P'$ is a hammock-joining path between $\hat{H}_i$ and $\hat{H}_j$; thus every edge of $P$ will be included in $H$.
    \end{enumerate}
    
    It follows that $G[\bar{\mcH}]$ contains all shortest cross edge paths. \qedhere
    
\end{proof}

\subsection{Extending $\bar{\mcH}$ to $\tilde{\mcH}$ by Lca Paths}

The second to last step in the construction of our hammock decomposition is to extend $\bar{H}_i$ along the path between $\bar{r}_i'$ and $\lca(\bar{r}_i, \bar{r}_i')$ to the child of $\lca(\bar{r}_i, \bar{r}_i')$ in $\TBFS$ which contains $\bar{r}_i'$ in its $\TBFS$ subtree. This will allow us to argue that our final hammock decomposition is lca-respecting.

\subsubsection{Hammock Ancestry $\leftrightarrow$ Lca Ancestry in $\bar{\mcH}$}
The key to extending along lca paths will be to show that the lca structure of our lca equivalence classes reflects the ancestry structure of the hammocks in our forest of hammocks. As with the proofs of the previous section, we will leverage the connected belowness of certain edges, as summarized in the following lemma. 


\begin{lemma}\label{lem:conBelow}
    Suppose $V(\bar{H}_i) \cap V(\bar{H}_j) = \{v_{ij}\}$ where $C_j \prec C_i$ and let $e = \{l(C_j),v\}$ be the child edge of $l(C_j)$ in $\TBFS$ satisfying $v_{ij} \in \TBFS(v)$. Then $e$ is connected below.
\end{lemma}
\begin{proof}
    Our aim is to construct a path from $r$ to $\TBFS(v)$ whose only vertex in $\TBFS(l(C_j))$ is its endpoint in $\TBFS(v)$. We assume WLOG that $v_{ij} \in \bar{T}_i'$.
    
    If $e \in H$ then we know that $e$ is connected below since by \Cref{lem:thirdPath} and \Cref{lem:pathsConBel} $H$ is connected below. Thus, we may assume that $e \not \in H$. It follows that $l(C_j) \not \in \bar{H}_i$: if $l(C_j)$ were in $\bar{H}_i$ then it would have to be in $\bar{T}_i'$ since $v_{ij} \in \bar{T}_i'$ and no vertex in $\bar{T}_i$ is related to a vertex in $\bar{T}_i'$ but $v_{ij}$ and $l(C_j)$ are related; however if $l(C_j)$ were in $\bar{T}_i'$ then since $v_{ij}$ is also in $\bar{T}_i'$ and $\bar{T}_i'$ is a connected subtree of $\TBFS$, it would follow that $e \in \bar{H}_i$ and therefore $e \in H$ and connected below.
    
    Given that $l(C_j) \not \in \bar{H}_i$, we can construct our path $P$ from $r$ to $\TBFS(v)$ as follows. First, take the path $P_1 := \TBFS(r, \bar{r}_i)$ from $r$ to $\bar{r}_i$. Continue this through an arbitrary path $P_2 \subseteq \bar{H}_i$ to any vertex in $\TBFS(v)$, using a single edge of $E_c$; $P_2$ is well-defined since $v_{ij} \in \TBFS(v) \cap \bar{T}_i'$. Let $P = P_1 \oplus P_2$ be the resulting path. Since $v_{ij} \in \TBFS(v)$ this path is indeed from $r$ to $\TBFS(v)$. 
    
    It remains to verify that $P \cap \TBFS(l(C_j)) \subseteq \TBFS(v)$ and in particular we show that $P \cap \TBFS(l(C_j)) = \{v_{ij}\}$. It cannot be the case that $P_1$ intersects $\TBFS(l(C_j))$ since it would then follow that $l(C_j)$ is an ancestor of $\bar{r}_j$ and since $l(C_j)$ is an ancestor of $v_{ij}$ and also therefore an ancestor of $\bar{r}_i'$ so $l(C_i) \preceq l(C_j)$, contradicting our assumption that $C_j \prec C_i$. Similarly, we know that $\bar{T}_i$ does not contain any vertices of $\TBFS(l(C_j))$ since it would then follow that $l(C_i) \preceq l(C_j)$. Thus, edges of $P_2$ in $\bar{T}_i$ have no vertices from $\TBFS(l(C_j))$; similarly, edges of $P_2$ in $\bar{T}_i'$ cannot contain vertices of $\TBFS(l(C_j)) \setminus \TBFS(v)$ since $\bar{T}_i'$ is a connected subtree of $\TBFS(v)$ by our assumption that $l(C_j) \not \in \bar{H}_i$. Concluding, it follows that $P \cap \TBFS(l(C_j)) \subseteq \TBFS(v)$ as desired.
\end{proof}

In arguing that the lca structure reflects the hammock ancestry structure, we will distinguish between hammocks based on which of their parent's trees they connect to.
\begin{definition}[Left, right child]
    Let $\bar{H}_j$ be a child of $\bar{H}_i$ in $\bar{\mcH}$. Then $\bar{H}_j$ is a left child of $\bar{H}_i$ if $\bar{r}_j \in \bar{T}_i$ and a right child of $\bar{H}_i$ if $\bar{r}_j \in \bar{T}_i'$.
\end{definition}

The following will allow us to argue that the paths we construct for left children in our forest of hammocks interact with their parents in an appropriate way.
\begin{lemma}\label{lem:leftChild}
    If $\bar{H}_j$ is a left child of $\bar{H}_i$ and $\bar{H}_i$ is not a root hammock then $l(C_j) \in \TBFS(\bar{r}_i, \bar{r}_j) \subseteq \bar{T}_i$.
\end{lemma}
\begin{proof}
    First, notice that since $\bar{r}_j \in \bar{T}_i$ we know that both $l(C_j)$ and $l(C_i)$ are ancestors of $\bar{r}_j$ and so $l(C_j)$ and $l(C_i)$ are related. Thus, if $l(C_j) \not \in \TBFS(\bar{r}_i, \bar{r}_j)$ then it must be the case that $\bar{r}_i \prec l(C_j)$; suppose for the sake of contradiction that indeed $\bar{r}_i \prec l(C_j)$.
    
    We first claim that there is a path $P_1$ with edges contained in $\bar{T}_i$ from $\bar{r}_j$ to a vertex $v$ which is in a cycle $C$ in $\hat{H}_i$. Further, if $u = \highV( V(\bar{T}_i) \cap V(C))$ and $u' = \highV( V(\bar{T}_i') \cap V(C))$ then we will have $v \neq u, u'$. Let us argue why such a $P_1$ exists. It suffices to argue that there is some $v \in \hat{T}_i$ where $v \neq \hat{r}_i$ and a path between $\bar{r}_j$ and $\hat{T}_i$ ending in $v$ since by \Cref{lem:fundHatCycle} the existence of such an $v$ would give us our required cycle. By noting that $\hat{r}_i = \bar{r}_i$ we can further simplify what we must prove. In particular, let $\bar{H}_{i'}$ be the parent of $\bar{H}_i$ in $\bar{\mcH}$ which we know exists by our assumption that $\bar{H}_i$ is not a root hammock. Additionally, notice that we always know that $\bar{r}_i = \hat{r}_i$ since if $\bar{r}_i$ were strictly higher than $\hat{r}_i$ it would be because it lies in a connected component of $\mcPHJ \setminus \hat{H}$ incident to both $\bar{r}_i$ and $\hat{r}_i$. Such a component is incident to $\bar{H}_{i'}$ by definition of a forest of hammocks. Even stronger, such a component must be incident to $\hat{H}_{i'}$ since every component of $\mcPHJ \setminus \hat{H}$ incident to $\bar{H}_{i'}$ is also incident to $\hat{H}_{i'}$ which means that any such component would always be assigned to $i'$.
    
    Thus, it suffices to argue that that there is some $v \in \hat{T}_i$ where $v \neq \bar{r}_i$ and a path between $\bar{r}_j$ and $\hat{T}_i$ ending in $v$. Note that $\bar{r}_j$ was added to $\bar{T}_i$ either because a component of $\mcPHJ \setminus \hat{H}$ was assigned to $i$ or because it was already $\hat{T}_i$. In particular, either $\bar{r}_j \in \hat{T}_i$ or $\bar{r}_j \in \bar{T}_i \setminus \hat{T}_i$.
    
    \begin{enumerate}
        \item Suppose $\bar{r}_j \in \hat{T}_i$. In this case we can let our path from $\bar{r}_j$ to $\hat{T}_i$ be the trivial path consisting only of $\bar{r}_j$; it remains only to show that $\bar{r}_j \neq \bar{r}_i$. However, it cannot be the case that $\bar{r}_j = \bar{r}_i$ since otherwise we would have that $H_{j}$ is a child of $\bar{H}_{i'}$ not $H_{i}$ in $\bar{\mcH}$, a contradiction.
        \item Suppose $\bar{r}_j \in \bar{T}_i \setminus \hat{T}_i$. It follows that $\bar{r}_j$ is incident to a component $F \subseteq \TBFS$ of $\mcPHJ \setminus \hat{H}$ which is assigned to $i$ with exactly one vertex $v$ in $V(\hat{T}_i)$. In this case we can let our path from $\bar{r}_j$ to $\hat{T}_i$ be the path in $F$ from $\bar{r}_j$ to $v$; it remains only to argue that $v \neq \bar{r}_i$. However, if $v$ were equal to $\bar{r}_i$ then we would have that $F$ shares a vertex with $\bar{H}_{i'}$ and, even stronger, $F$ shares a vertex with $\hat{H}_{i'}$. It follows that $F$ would be assigned to $\bar{H}_{i'}$, not $\bar{H}_i$, a contradiction.
    \end{enumerate}
    
    Next, we claim that there is a path $P_2$ from $r$ to $\bar{r}_j$ such that $P_2 \cap V(\bar{H}_i) = \{\bar{r}_j\}$. Let $P_2'$ an arbitrary path contained in $\bar{H}_j$ from $\bar{r}_j'$ to $\bar{r}_j$. Since $\bar{H}_i$ and $\bar{H}_j$'s vertex sets only intersect on $\bar{r}_j$, it follows that $P_2' \cap V(\bar{H}_i) = \{\bar{r}_j\}$. Next we form $P_2$ by concatenating $P_2'$ with $\TBFS(l(C_j), \bar{r}_j')$. We claim that $\TBFS(l(C_j), \bar{r}_j') \cap V(\bar{H}_i) = \emptyset$. Suppose for the sake of contradiction that there is some $x \in \TBFS(l(C_j), \bar{r}_j') \cap V(\bar{H}_i)$. If $x \in \bar{T}_i$ then $\bar{r}_i$ is an ancestor of both $\bar{r}_j$ and $\bar{r}_j'$ and so $l(C_j) \prec l(C_i)$, contradicting our assumption that $\bar{r}_i \prec l(C_j)$. On the other hand, if $x \in \bar{T}_i'$ then we have that $\bar{r}_i'$ is an ancestor of $\bar{r}_j'$; since $\bar{r}_j \in \bar{T}_i$, it follows that $C_i = C_j$, contradicting our assumption that $C_i$ and $C_j$ are distinct lca-equivalence classes. Thus, $P_2$ is indeed from $r$ to $\bar{r}_j$ and satisfies $P_2 \cap V(\bar{H}_i) = \{\bar{r}_j\}$.
    
    We conclude by observing that the above gives us a clawed cycle with $C$ as the cycle with paths $\TBFS(r, u)$, $\TBFS(r, u')$ and $P_1 \oplus P_2$, a contradiction.
\end{proof}

We now demonstrate that the tree structure of $\bar{\mcH}$ reflects the structure of our lca-equivalence classes. As the previous lemma handled the left child case, most of this proof will be focused on the right child case.
\begin{lemma}\label{lem:parChildEquiv}
    $\bar{H}_j \preceq_{\bar{\mcH}} \bar{H}_i$ implies $C_j \preceq C_i$ for every $i,j$.
\end{lemma}
\begin{proof}
    It suffices to show that if $\bar{H}_j$ is a child of $\bar{H}_i$ in $\bar{\mcH}$ then $C_j \preceq C_i$; thus, let $\bar{H}_i$ and $\bar{H}_j$ be an arbitrary parent-child pair in $\bar{\mcH}$. Since $\bar{\mcH}$ is a rooted forest of hammocks by \Cref{lem:rootedHDForest} we know that $V(\bar{H}_i) \cap V(\bar{H}_j) = \bar{r}_j$ and since $l(C_i)$ and $l(C_j)$ are ancestors of every vertex in $\bar{H}_i$ and $\bar{H}_j$ respectively, we know that  $l(C_i)$ and $l(C_j)$ are both ancestors of $\bar{r}_j$ and therefore either $C_i \prec C_j$ or $C_j \preceq C_i$.
    
    Suppose for the sake of contradiction that $C_i \prec C_j$ and suppose minimality of this counter-example; in particular, suppose that there are no ancestors $\bar{H}_{i'}$ and $\bar{H}_{j'}$ of $\bar{H}_j$ in $\bar{\mcH}$ where $\bar{H}_{i'}$ is the parent of $\bar{H}_{j'}$ but $l(C_{i'}) \prec l(C_{j'})$.  It follows that for every pair of parent-child pairs $\bar{H}_{i'}$ and $\bar{H}_{j'}$ which are ancestors of $\bar{H}_j$ in $\bar{\mcH}$ we know that $C_{j'} \preceq C_{i'}$. Moreover, letting $\bar{H}_k$ be the root hammock of the tree of hammocks containing $\bar{H}_i$ and $\bar{H}_j$ in $\bar{\mcH}$, it follows that $C_{i'} \preceq C_k$ for any ancestor $\bar{H}_{i'}$ of $\bar{H}_j$. Additionally, since we chose $\bar{H}_k$ to maximize $h(C_k)$ and $C_j \prec C_i$, there must be some pair of ancestors $\bar{H}_{j'}$ and $\bar{H}_{i'}$ of $\bar{H}_j$ where $C_{j'} \prec C_{i'}$ (since otherwise we would have $C_k \prec C_{j}$).
    
    It follows that there is a sequence of hammocks $\bar{H}_0, \bar{H}_1, \bar{H}_2, \bar{H}_3, \ldots, \bar{H}_\alpha, \bar{H}_{\alpha+1}$ where $l(C_l) = l(C_{l'})$ for $l, l' \in [\alpha]$ and $\bar{H}_{l+1}$ is the child of $\bar{H}_{l}$ in $\bar{\mcH}$ for $0 \leq l \leq \alpha$ where $\alpha \geq 1$ but $l(C_1) = \ldots = l(C_{\alpha}) \prec l(C_0),l(C_{\alpha+1})$; here, we have slightly abused notation and relabeled $\bar{H}_{i'}$, $\bar{H}_{j'}$, $\bar{H}_i$ and $\bar{H}_j$ to $\bar{H}_0$, $\bar{H}_1$, $\bar{H}_{\alpha}$ and $\bar{H}_{\alpha + 1}$ respectively. We let $u := l(C_l)$ for $l \in [\alpha]$ and let $v_l$ be the child of $u$ in $\TBFS$ whose subtree contains $r_l$ Moreover, we claim that by \Cref{lem:leftChild} we know that for each $l \in [\alpha]$ we must have that $\bar{H}_{l+1}$ is a \emph{right} child of $\bar{H}_l$: to see why, notice that if $\bar{H}_{l+1}$ is a left child of $\bar{H}_l$ then \Cref{lem:leftChild} tells us that $l(C_{l+1}) \in T_l$ and so $l(C_{l+1}) \preceq r_l \prec u$; if $l < \alpha$ this contradicts our assumption that $l(C_{l+1}) = u$ and if $l = \alpha$ this contradicts our assumption that $l(C_{\alpha}) = u \prec l(C_{\alpha+1})$.
    
    We will use the existence of such a sequence of right children to contradict \Cref{lem:noFundCycle}.
    
    First, we claim that $v_1$ and $v_{\alpha+1}$ are connected in $\TBFS(u) \setminus \{u\}$. To see why, first notice that the graph induced by the union of $\bar{H}_1, \bar{H}_2, \ldots, \bar{H}_{\alpha}$ is connected by virtue of each hammock in this sequence being the parent of the next hammock in this sequence. $v_1$ and $v_{\alpha+1}$ are therefore connected in $\TBFS(u) \setminus \{u\}$ since we also know that $r_1$ and $r_{\alpha+1}$ are both contained in this graph and descendants of $v_1$ and $v_{\alpha+1}$ respectively.
    
    Notice that $\bar{H}_0$ and $\bar{H}_1$ intersect at $r_1 \preceq v_1$ and $\bar{H}_\alpha$ and $\bar{H}_{\alpha + 1}$ intersect at $r_{\alpha + 1} \preceq v_{\alpha + 1}$. Thus, since $\bar{H}_{0} \prec \bar{H}_1$ and $\bar{H}_{\alpha + 1} \prec \bar{H}_{\alpha}$, it follows by \Cref{lem:conBelow} that $\{u, v_1\}$ and $\{u, v_{\alpha + 1}\}$ are connected below.
    
    Thus, if we can show that $v_1 \neq v_{\alpha+1}$ then we will have contradicted \Cref{lem:noFundCycle}. Suppose for the sake of contradiction that $v_1 = v_{\alpha+1}$. Then, there must be some contiguous subsequence $(v_{\beta}, v_{\beta+1}, \ldots v_{\beta + \gamma})$ of $(v_1, \ldots, v_{\alpha+1})$ for $\gamma \geq 1$ where $v_{\beta} = v_{\beta + \gamma}$ but $v_{\beta + x} \neq v_{\beta + y}$ for all $x,y < \gamma$ where $y \neq x$. To simplify notation we assume that $\beta = 1$. We will show that all choices of $\gamma$ contradict the assumption that $v_1 = v_{\gamma + 1}$ and so it must be the case that $v_1 \neq v_{\alpha+1}$. We illustrate each of the following contradictions in Figures \ref{sfig:desc1}, \ref{sfig:desc2} and \ref{sfig:desc3}.
    \begin{enumerate}
        \item Suppose $\gamma = 1$ (i.e. there is a single hammock $\bar{H}_1$ we are considering with an lca of $u$). Then we would have that $r_1$ and $r_1'$ are both descendants of $v_1$ and $v_{\gamma + 1}$, meaning $l(C_1) \preceq v_1 = v_{\gamma+1} \prec u$, a contradiction to the fact that $l(C_1) = u$.
        \item Suppose $\gamma = 2$. Since $\bar{H}_2$ is a right child of $\bar{H}_1$, we know that $r_2 \preceq v_2$ \emph{and} $r_1' \preceq v_2$. Moreover, we know that $r_3 \preceq v_3$ and $r_3 \in T_2'$ and so $r_2' \preceq v_3 = v_1$. Summarizing, we have $r_1, r_2' \preceq v_1$ and $r_1', r_2 \preceq v_2$; however since $v_1, v_2 \neq u$, it follows that $C_1 = C_2$, contradicting our assumption that they are distinct lca-equivalence classes.
        \item Suppose $\gamma \geq 3$. Let $P$ be an arbitrary path connecting $r_{1}$ and $r_{\gamma+1}$ in the graph induced by $\bar{H}_1, \bar{H}_2, \ldots, \bar{H}_{\gamma}$. By our assumption that $\bar{H}_{l+1}$ is a right child of $\bar{H}_l$ for $l \in [\gamma]$, we know that such a path uses at least one edge in $C_l$ for $l \in [\gamma]$.  Moreover, we may assume WLOG that $P$ uses exactly one edge from $C_l$ for $l \in [\gamma]$: if $x_l$ and $y_l$ are the first and last vertices that $P$ visits in $C_l$ then by definition of a hammock there is a path from $x_l$ to $y_l$ which uses exactly one edge of $C_l$; we assume that $P$'s subpath restricted to $C_l$ is this path. Next, since we have assumed that $v_1$ and $v_{\gamma + 1}$ are equal, we know that both $r_{1}$ and $r_{\gamma+1}$ lie in $\TBFS(v_{1})$. By our choice of $\gamma$, we have that the graph induced by $\TBFS(r_1, r_{\gamma + 1}) \cup P$ contains a cycle $C$ with exactly one edge from at least three $C_l$. However, we can easily form a clawed cycle from such a cycle. In particular, let $P_1$, $P_2$ and $P_3$ be paths to $C$ from $u$ in $\TBFS$ via $v_1$, $v_2$ and $v_3$ respectively. Then $C$ with these paths is a clawed cycle, a contradiction.\qedhere
    \end{enumerate}
\end{proof}

\begin{figure}
    \centering
    \begin{subfigure}[b]{0.32\textwidth}
        \centering
        \includegraphics[width=\textwidth,trim=140mm 40mm 150mm 35mm, clip]{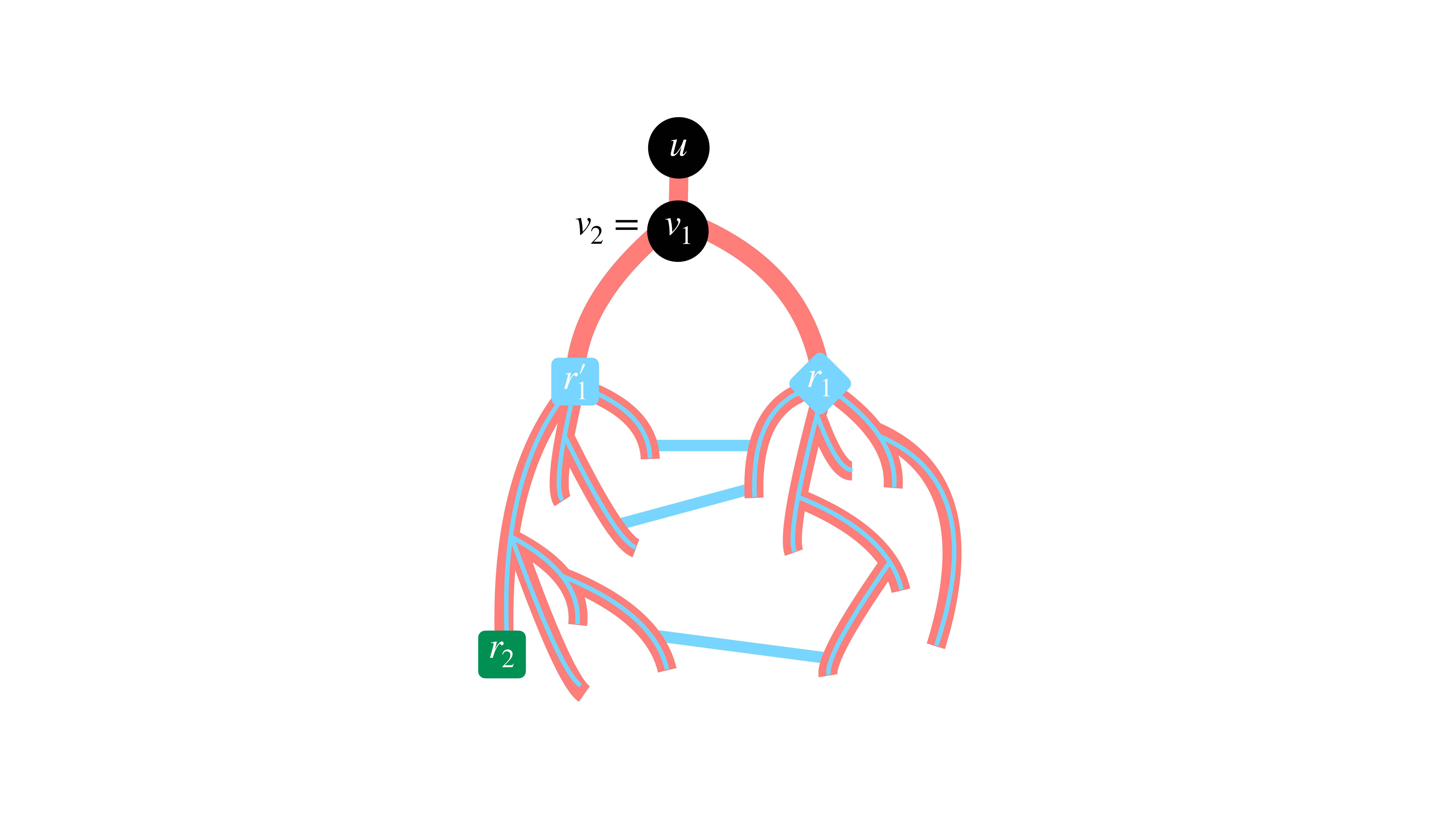}
        \caption{$\gamma=1$: $l(C_1) \prec u$}\label{sfig:desc1}
    \end{subfigure}
    \hfill
    \begin{subfigure}[b]{0.32\textwidth}
        \centering
        \includegraphics[width=\textwidth,trim=140mm 40mm 150mm 35mm, clip]{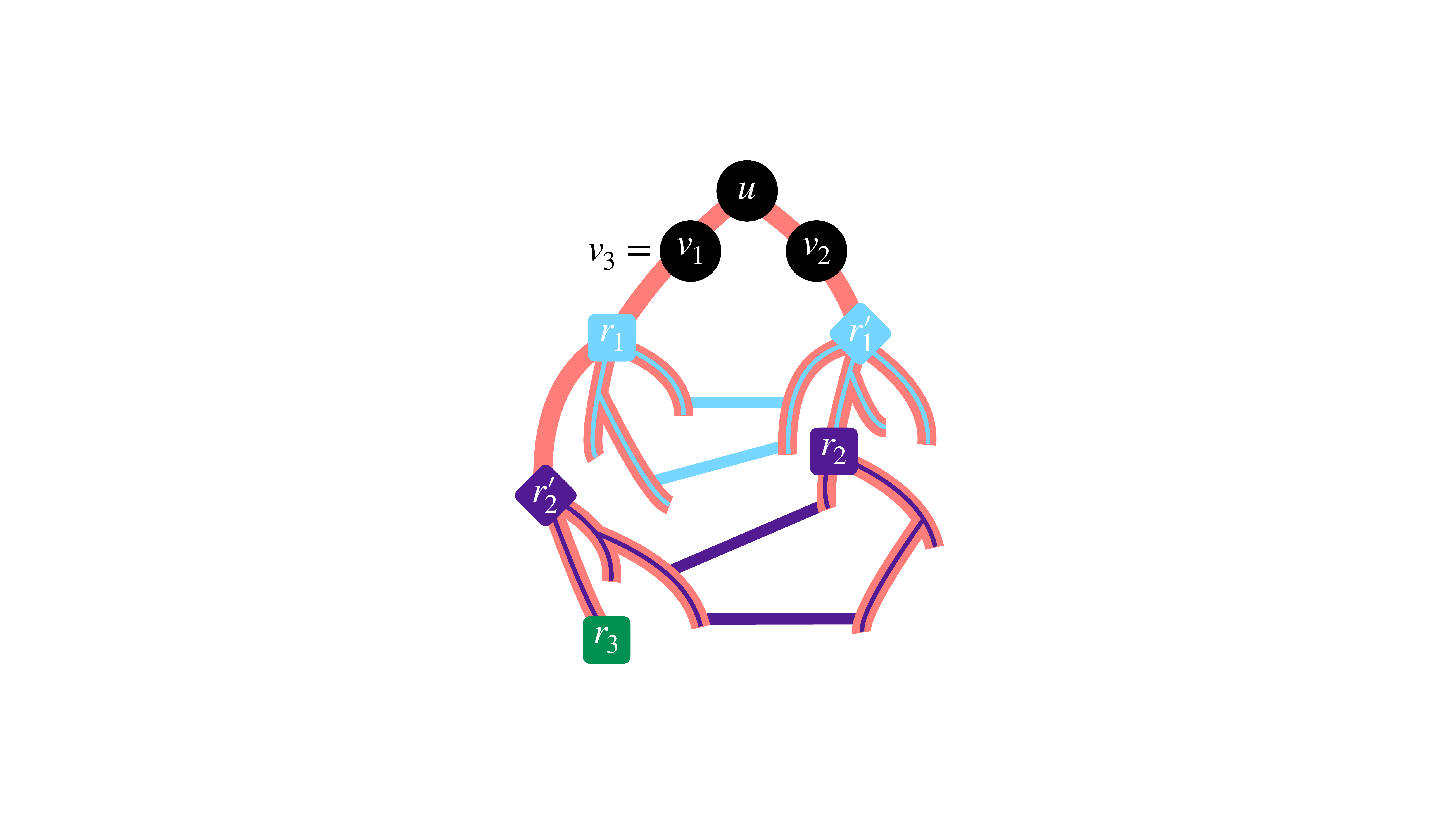}
        \caption{$\gamma = 2$: $C_1 = C_2$}\label{sfig:desc2}
    \end{subfigure}
    \begin{subfigure}[b]{0.32\textwidth}
        \centering
        \includegraphics[width=\textwidth,trim=140mm 40mm 150mm 35mm, clip]{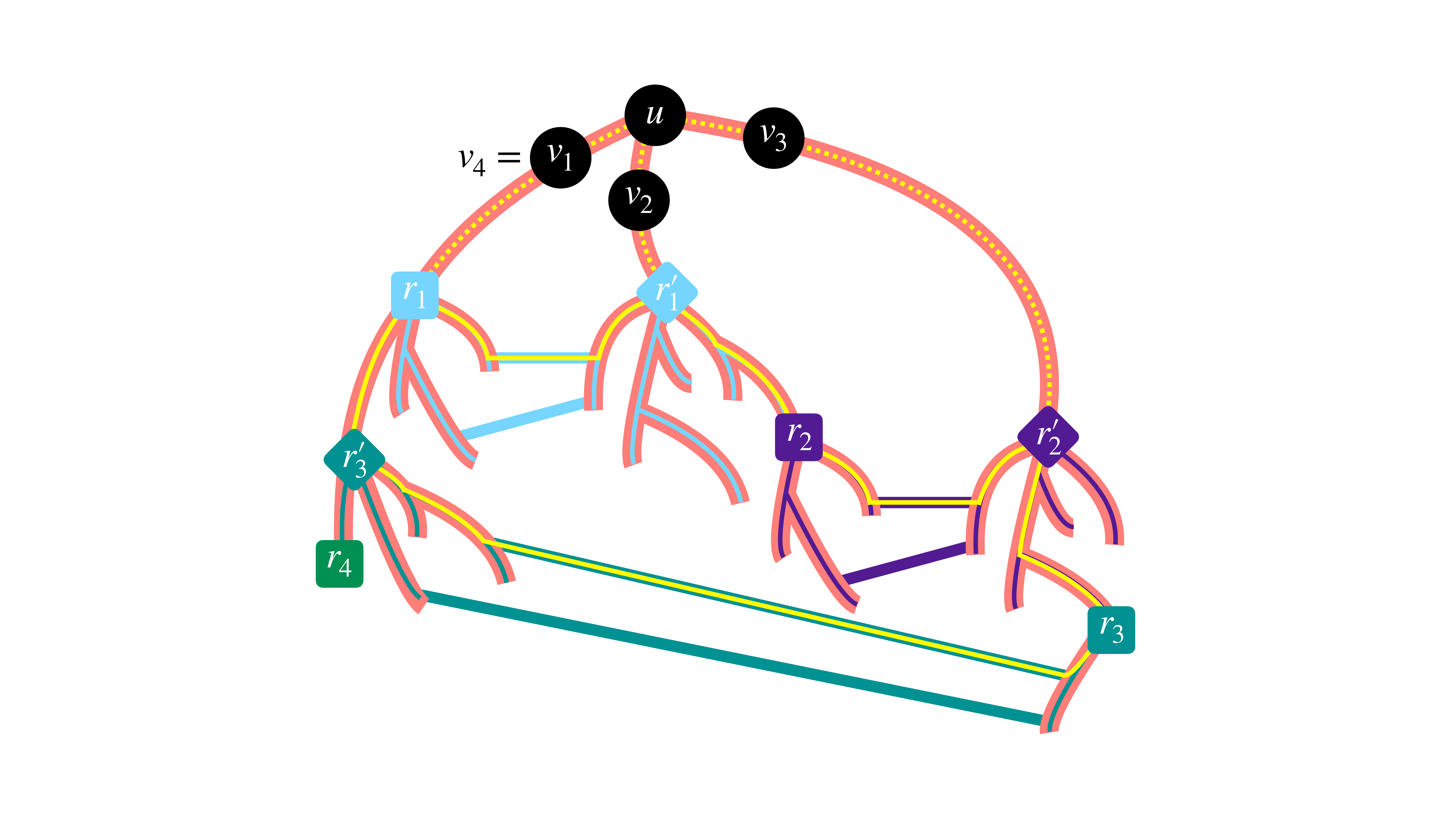}
        \caption{$\gamma \geq 3$: a clawed cycle}\label{sfig:desc3}
    \end{subfigure}
    \caption{Each of the three contradictions we arrive at in the proof of \Cref{lem:parChildEquiv} based on the value of $\gamma$. In (c) we highlight the cycle of our clawed cycle in solid yellow and each of its paths in dotted yellow.}
\end{figure}



\subsubsection{Extending Along Lca Paths}

In this section we show that the above mentioned paths are appropriately disjoint from one another as well as from the thus far computed hammocks. We will then use this disjointness to extend $\bar{\mcH}$ to $\tilde{\mcH}$. We formally define these paths---which we illustrate in illustrate in \Cref{sfig:hamConnPaths}---as follows.

\begin{definition}[Lca Paths $P_i$, $P_i'$, $\mcP_i$ $\mcP$]
    For each $i$ we let $P_i := \intV(\TBFS(\bar{r}_i, l(C_i)))$ and $P_i' := \intV(\TBFS(\bar{r}_i', l(C_i)))$ be the lca paths from $\bar{r}_i$ and $\bar{r}_i'$ respectively to $l(C_i)$, excluding the endpoints. We let $\mcP_i$ be the set which always contains $P_i'$ and which additionally contains $P_i$ if $\bar{H}_i$ is a root hammock in $\bar{\mcH}$. Lastly, we let  $\mcP := \bigcup_i \mcP_i $ be the collection of all relevant lca paths.
\end{definition}

\begin{figure}
    \centering
    \begin{subfigure}[b]{0.49\textwidth}
        \centering
        \includegraphics[width=\textwidth,trim=0mm 0mm 0mm 0mm, clip]{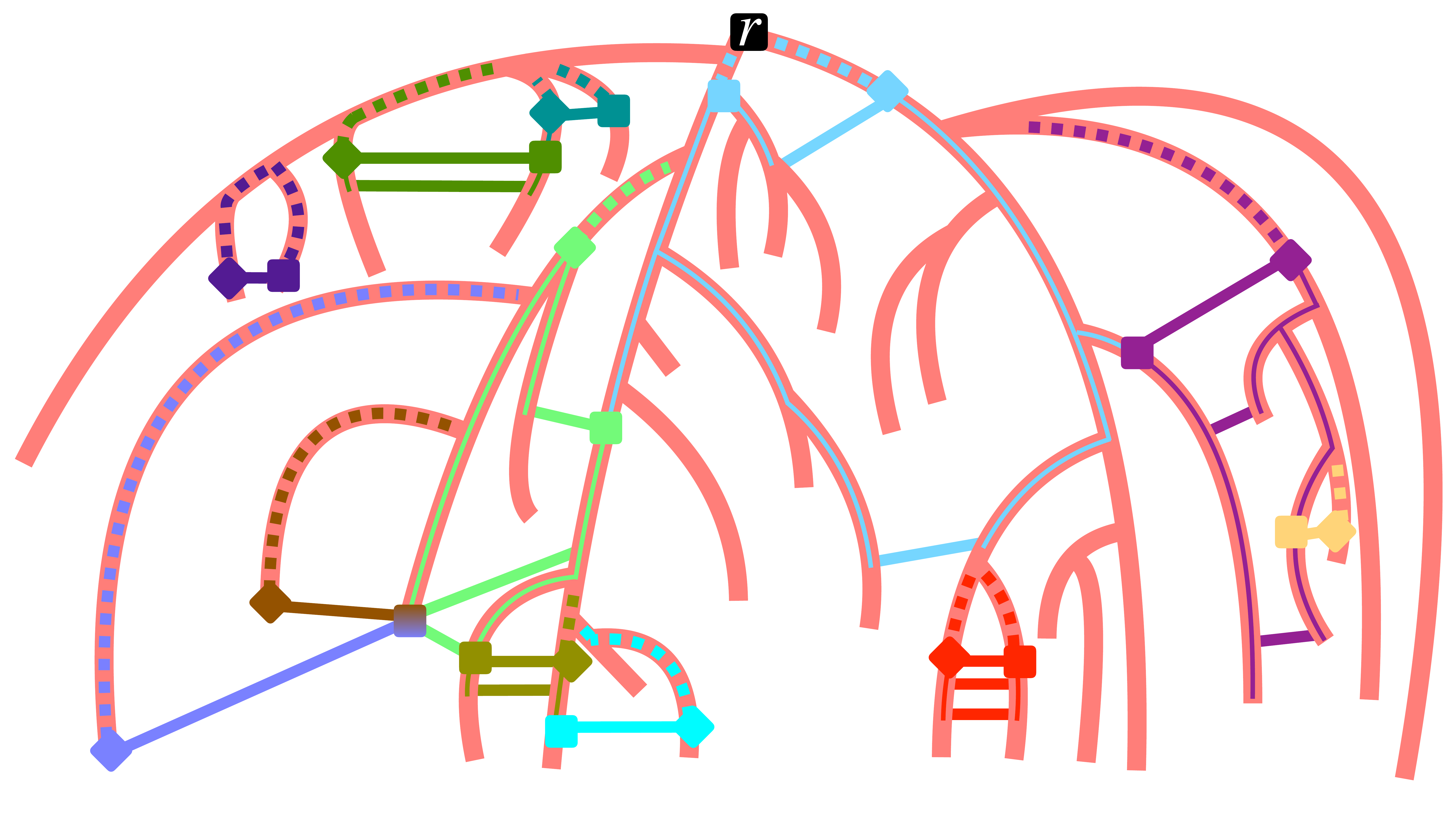}
        \caption{Lca paths.}\label{sfig:hamConnPaths}
    \end{subfigure}
    \hfill
    \begin{subfigure}[b]{0.49\textwidth}
        \centering
        \includegraphics[width=\textwidth,trim=0mm 0mm 0mm 0mm, clip]{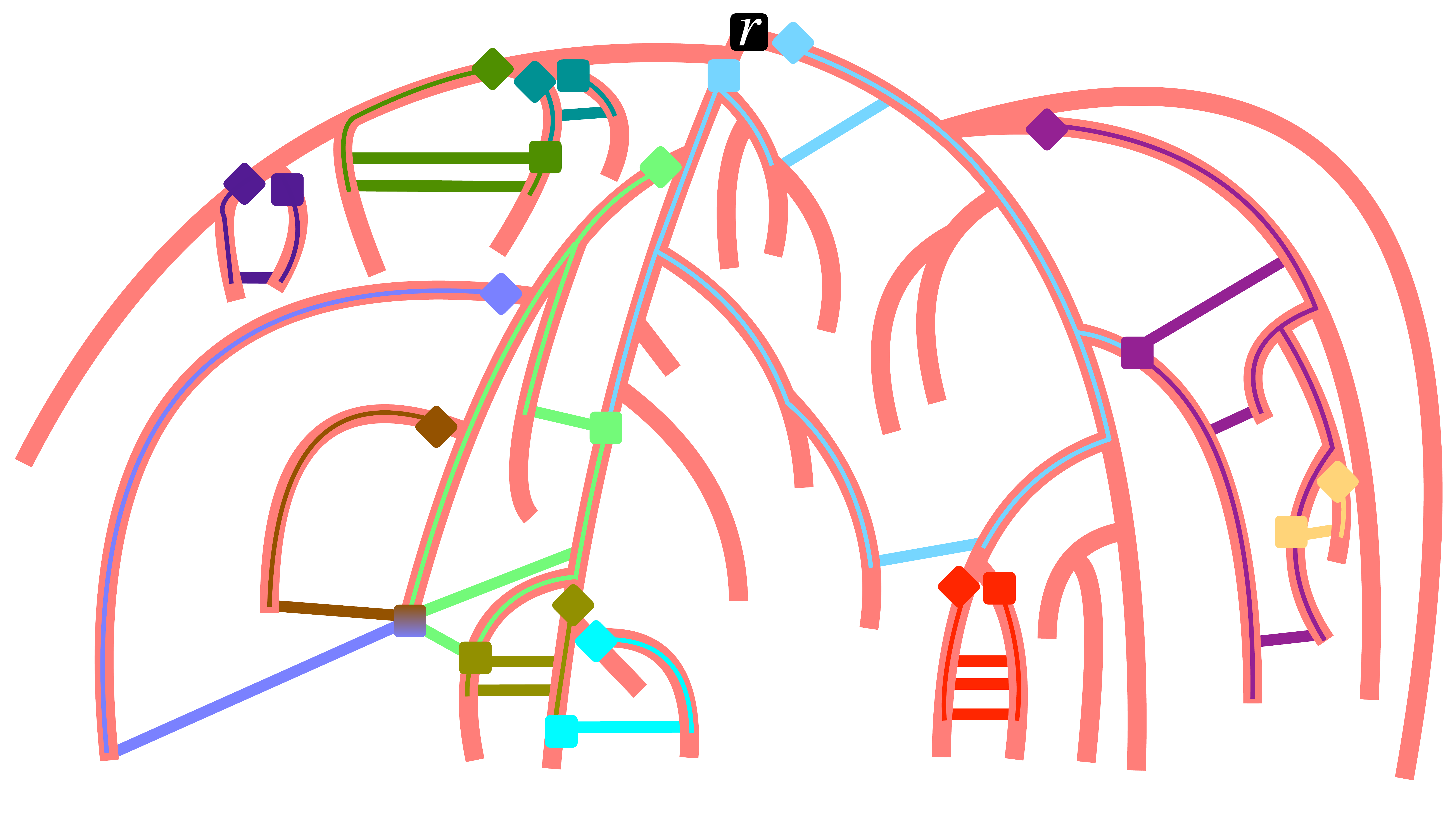}
        \caption{Extending $\bar{\mcH}$ to $\tilde{\mcH}$ via lca paths.}\label{sfig:addPaths}
    \end{subfigure}
    \caption{An illustration of the paths of our hammock decomposition. Each such path dotted and given in a color corresponding to its constituent hammock. On the left we give the paths and on the right we give the result of adding these paths to $\bar{\mcH}$, resulting in $\tilde{\mcH}$. Notice that each root hammock has two paths in $\mcP$.}
\end{figure}

Given the above paths we now describe how we construct $\tilde{\mcH}$ from $\bar{\mcH}$ be extending our hammocks along these paths. Formally, we let $\tilde{H}_i$ be the graph induced in $G$ by $V(\bar{H}_i)$ along with all vertices in paths in $\mcP_i$. We also let the root hammocks of $\tilde{\mcH}$ be the same as the root hammocks of $\bar{\mcH}$. We illustrate this in \Cref{sfig:addPaths}.

To show that this yields an lca-respecting tree of hammocks we first observe that all of the above paths are appropriately disjoint.

\begin{lemma}\label{lem:disjointPaths}
    For each $P \in \mcP$ and $i$ we have $V(P) \cap V(\bar{H}_i)= \emptyset$. Similarly, for distinct $P, P' \in \mcP$ we have $V(P) \cap V(P') = \emptyset$.
\end{lemma}
\begin{proof}
     First, let us show $V(P) \cap V(\bar{H}_i) = \emptyset$. Let $P \in \mcP_j$; we assume WLOG that $P = P_j'$. Suppose that there is some vertex $v \in V(P_j')$ which is contained in $V(\bar{H}_i)$ and let $v$ be the lowest such vertex in $P_j'$. It follows that the path between $\bar{r}_j'$ and $v$ must be a hammock joining path between $\bar{H}_i$ and $\bar{H}_j$.
    
    Continuing, let $e$ be the parent edge of $\bar{r}_j'$ in $\TBFS$; $e$ is in the above hammock-joining path and so is contained in some hammock $\bar{H}_l$ which intersects $\bar{H}_j$ on $\bar{r}_j'$. Since $\bar{H}_l$ and $\bar{H}_j$ intersect on $\bar{r}_j'$, $\bar{H}_l$ must be either the parent, sibling or child of $\bar{H}_l$ in $\bar{\mcH}$ as per \Cref{lem:pathStruct}. However, $\bar{H}_l$ cannot be the parent or sibling of $\bar{H}_j$ since it intersects $\bar{H}_j$ on $\bar{T}_j'$; thus $\bar{H}_l$ is the child of $\bar{H}_j$. We claim that $e \in \hat{H}_l$; indeed this is immediate from the fact that $e$ is incident to $\bar{r}_j'$ and since $\bar{H}_j$ is the parent of $\bar{H}_l$, $e$ would have been included in $\bar{H}_j$, not $\bar{H}_l$, if it weren't in $\hat{H}_l$. However, since $v$ is the parent of $\bar{r}_j'$ and $\bar{r}_j' \in \bar{H}_l$, we know that $\bar{r}_j' \neq \hat{r}_l, \hat{r}_l'$. 
    
    Thus, applying \Cref{lem:fundHatCycle} we know that there is a hammock-fundamental cycle $C \subseteq \hat{H}_l$ containing $\bar{r}_j'$ where $\bar{r}_j'$ is not the highest vertex in either $V(C) \cap \bar{T}_l$ or $V(C) \cap \bar{T}_l'$; let these highest vertices be $v_l$ and $v_l'$ and notice that to arrive at a contradiction it suffices to argue that there is a path $P$ from $r$ to $\bar{r}_j'$ which is internally vertex disjoint from $V(\bar{H}_l)$ since $C$ along with paths $\TBFS(r, v_l)$, $\TBFS(r, v_l')$ and $P$ would then form a clawed cycle. Let $P$ be the concatenation of $\TBFS(r, \bar{r}_j)$ with an arbitrary path in $\bar{H}_j$ connecting $\bar{r}_j$ and $\bar{r}_j'$ and $e$. $\bar{r}_j'$ is the only vertex which the latter shares with $\bar{H}_l$ so it suffices to show that $\TBFS(r, \bar{r}_j)$ shares no vertices with $\bar{H}_l$. Suppose for the sake of contradiction that $x \in \TBFS(r, \bar{r}_j) \cap V(\bar{H}_l)$. Then, if $x$ and $\bar{r}_j'$ were in the same hammock trees of $\bar{H}_l$ then we would have $l(C_j) \prec l(C_l)$, contradicting $l(C_l) \preceq l(C_j)$ by \Cref{lem:parChildEquiv} and the fact that $\bar{H}_l$ is a child of $\bar{H}_j$. On the other hand, if $x$ and $\bar{r}_j'$ were in different hammock trees of $\bar{H}_l$ then we would have that $C_j = C_l$, contradicting the fact that $C_j$ and $C_l$ are distinct. Thus, we conclude that $V(P) \cap V(\bar{H}_i) = \emptyset$.

    Next, let us show $V(P) \cap V(P') = \emptyset$ for distinct $P, P' \in \mcP$. We may assume that it is not the case that $P, P' \in \mcP_i$ for some $i$ since then $P$ and $P'$ are vertex-disjoint by construction. Thus, WLOG we assume $P = P_i'$ and $P' = P_j'$ for $i \neq j$. Suppose for the sake of contradiction that $P_i'$ and $P_j'$ share a vertex $x$. By definition of $P_i'$ and $P_j'$ we know that $x \neq \bar{r}_i', \bar{r}_j'$. Moreover, since in the first part of this proof we showed that $V(P_i), V(P_j) \cap V(\bar{H}) = \emptyset$, it follows that $\TBFS(\hat{r}_i',\bar{r}_i') \oplus \TBFS(\bar{r}_i', x) \oplus \TBFS(x, \bar{r}_j') \oplus \TBFS(\bar{r}_j', \hat{r}_j')$ is a hammock-joining path (\Cref{dfn:hamJoin}) and so would have been included in $\bar{H}$. However, then $x$ would have been included in $\bar{H}$ (by definition of $\bar{H}$). But, this contradicts the fact that $V(P_i'), V(P_j') \cap V(\bar{H}) = \emptyset$.
\end{proof}

We conclude that our forest of hammocks is indeed lca-respecting (\Cref{dfn:lcaRespecting}).
\begin{lemma}\label{lem:tildeFor}
    $\tilde{\mcH}$ is an lca-respecting rooted forest of hammocks where for each root hammock $\tilde{H}_k \in \tilde{\mcH}$ we have $\lca(\tilde{r}_k, \tilde{r}_k')$ is the parent of $\tilde{r}_k$ and $\tilde{r}_k'$ in $\TBFS$.
\end{lemma}
\begin{proof}
    It is easy to see by the disjointness of the paths in $\mcP$ that adding all vertices in $\mcP_i$ (along with all induced edges) to $\bar{H}_i$ indeed results in a hammock. To see that the result is a rooted forest of hammocks we need only verify that any cycle $C$ continues to only be incident to at most $1$ hammock. However, notice that the disjointness properties of our paths guarantee that no cycle can use an edge incident to a vertex in a path in $\mcP$; thus, since $\bar{\mcH}$ was a forest by \Cref{lem:rootedHDForest}, so too is $\tilde{\mcH}$.
    
    Lastly, to see that the resulting rooted forest of hammocks is lca-respecting notice that $\tilde{H}_j$ is a child of $\tilde{H}_i$ in $\tilde{\mcH}$ iff $\bar{H}_i$ is a parent of $\bar{H}_j$ in $\bar{\mcH}$ and so the first condition of lca-respecting---$V(\tilde{H}_i) \cap V(\tilde{H}_j) = \{\tilde{r}_j\}$ whenever $\tilde{H}_i$ is a parent of $\tilde{H}_j$ in $\tilde{\mcH}$---still holds by \Cref{lem:rootedHDForest}. The second condition---the parent of $\tilde{r}_j'$ in $\TBFS$ is $\lca(\tilde{r}_j, \tilde{r}_j')$ and $\lca(\tilde{r}_j, \tilde{r}_j') \in \tilde{H}_i \cup \{\lca(r_i, r_i')\}$ where $\tilde{H}_i$ is the parent of $\tilde{H}_j$---holds by definition of the paths in $\mcP$ and \Cref{lem:parChildEquiv}. The condition on each hammock root similarly follows.
\end{proof}

\subsection{Extending $\tilde{H}$ to $\mcH$ by Adding Dangling Subtrees}
The last step of our hammock decomposition involves adding subtrees of $\TBFS$ which consist of unassigned nodes to hammocks. We do this to ensure that our hammock decompositions indeed partitions all edges of the input series-parallel graph.


Formally, we construct $\mcH$ from $\tilde{\mcH}$ as follows. Let $E_p$ be the collection of the parent edges in $\TBFS$ of $\tilde{r}_i'$ for every hammock $H_i$ (where as usual we use $\tilde{r}_i$ and $\tilde{r}_i'$ to denote the hammock roots of $\tilde{H}_i$). We let $T_0$ be the connected component of $\TBFS \setminus (E_p \cup E(G[\tilde{\mcH}]))$ which contains $r$. For each connected component (i.e.\ dangling tree) $T \neq T_0$ of $\TBFS \setminus (E_p \cup E(G[\tilde{\mcH}]))$ we add $T$ to an arbitrary hammock which contains $\highV(T)$. As a reminder, $\highV(T)$ gives the vertex in $T$ which is highest in $\TBFS$. Lastly we designate each hammock $H_k$ in $\mcH$ such that $r_k \in V(T_0)$ and $\tilde{H}_k$ was designated a root hammock in $\tilde{\mcH}$ as a root of $\mcH$. We illustrate the result of adding the dangling trees in \Cref{sfig:dangTrees} and the final hammock decomposition in \Cref{sfig:finalHD}.

\begin{figure}
    \centering
    \begin{subfigure}[b]{0.49\textwidth}
        \centering
        \includegraphics[width=\textwidth,trim=0mm 0mm 0mm 0mm, clip]{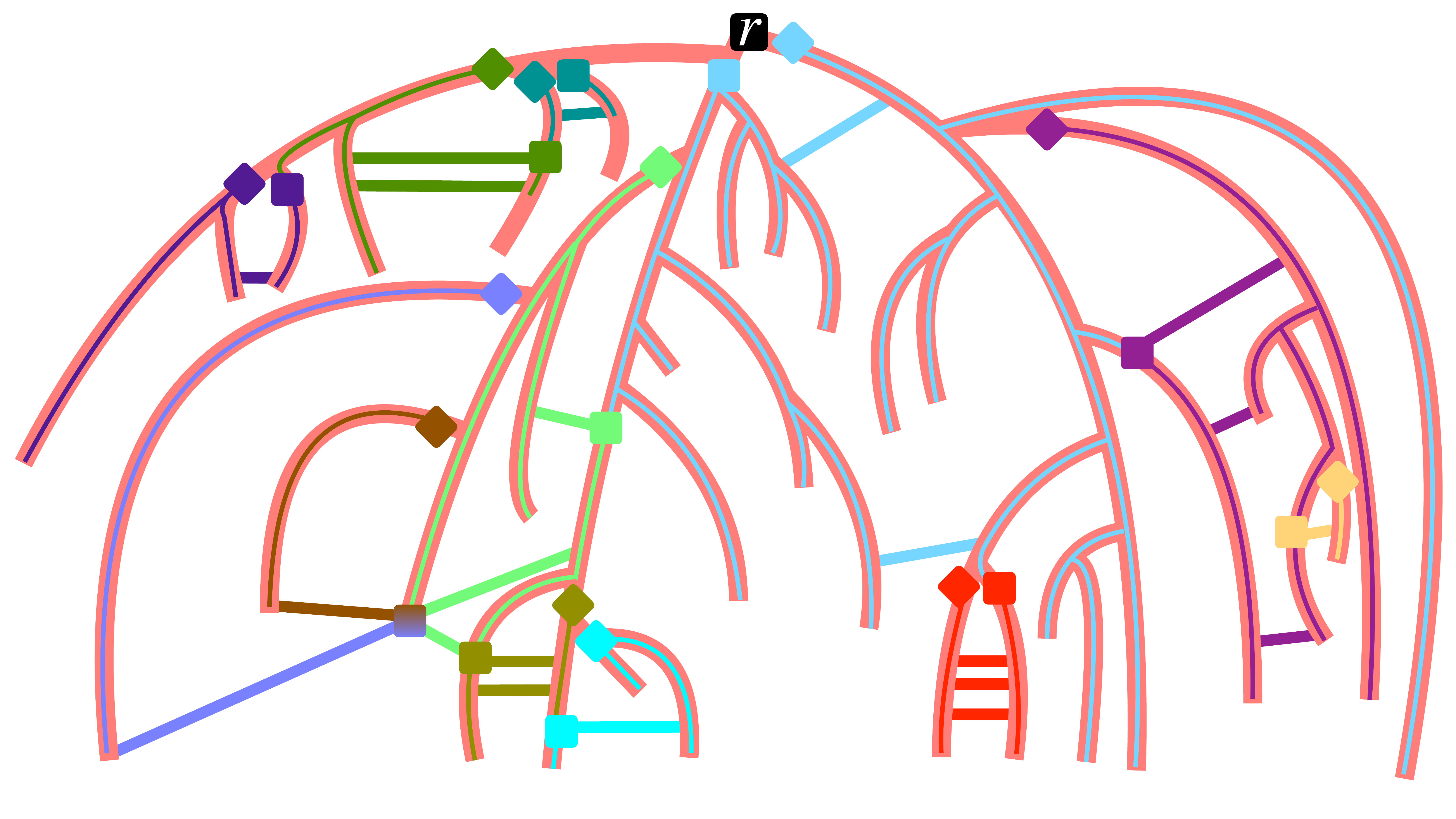}
        \caption{Adding dangling trees.}\label{sfig:dangTrees}
    \end{subfigure}
    \hfill
    \begin{subfigure}[b]{0.49\textwidth}
        \centering
        \includegraphics[width=\textwidth,trim=0mm 0mm 0mm 0mm, clip]{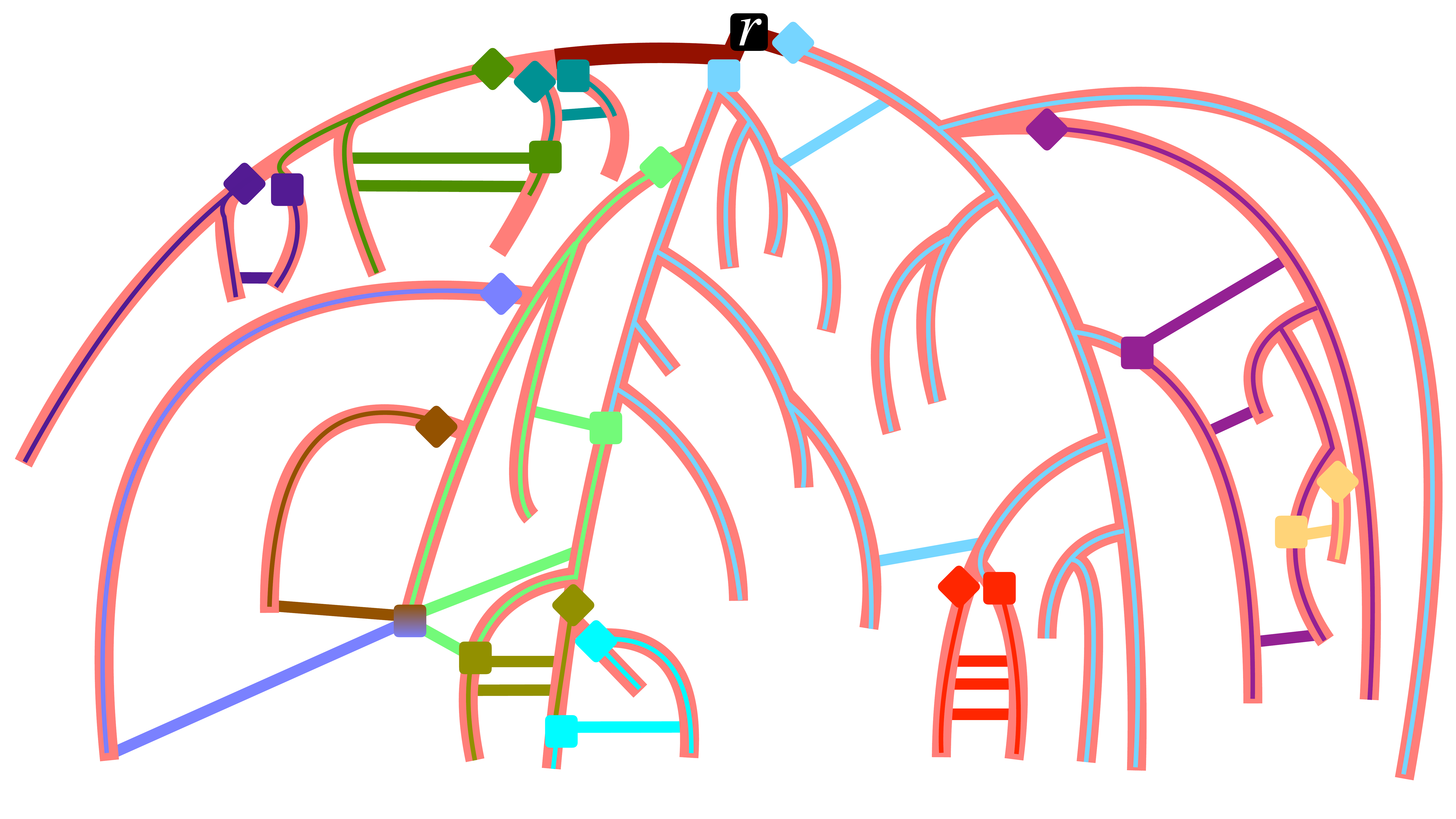}
        \caption{Final hammock decomposition.}\label{sfig:finalHD}
    \end{subfigure}
    \caption{The result of adding the dangling trees to each of our hammocks and the final hammock decomposition. On the right we give $T_0$ in dark red.}
\end{figure}

We proceed to show that this indeed results in a hammock decomposition (\Cref{dfn:PCHD}). We begin by observing that $\mcH$ is an appropriate forest of hammocks. 

\begin{lemma}\label{lem:finalFH}
    $\mcH$ is an lca-respecting rooted forest of hammocks with base tree $T_0$.
\end{lemma}
\begin{proof}
    We first observe that each $H_i$ is indeed a hammock: since we only add a dangling tree $T$ to $\tilde{H}_i$ if $\highV(T)$ is contained in $\tilde{H}_i$ and so the roots of $H_i$ are unrelated. Moreover all the $H_i$ must be edge disjoint by the edge disjointness of the $\tilde{\mcH}$ and how we construct our $H_i$.
    
    Next, we will argue the cycle property of forests of hammocks; namely that any cycle's edges are fully contained within a single hammock, thereby showing that $\mcH$ is a forest of hammocks.
    
    First, we observe that if $\{u, v\}$ is an edge in a dangling tree $T$ where $v \prec u$ then there does not exist a cross edge $e_c = \{u_c, v_c\} \in E_c$ such that $u_c \in \TBFS(v)$ but $v_c \not \in \TBFS(u)$. To see why this is the case suppose that such an $e_c$ existed and belonged to hammock $\tilde{H}_i \in \tilde{\mcH}$ and recall that by \Cref{lem:tildeFor} $\tilde{\mcH}$ is an lca-respecting rooted forest of hammocks. By definition of $e_c$ we know that $ u \prec \lca(\tilde{r}_i, \tilde{r}_i')$ which is to say that $\lca(\tilde{r}_i, \tilde{r}_i')$ lies strictly higher in $\TBFS$ than $u$. But, since $\tilde{\mcH}$ is lca-respecting and $\lca(\tilde{r}_k, \tilde{r}_k')$ is the parent of $\tilde{r}_k'$ for any root hammock $\tilde{H}_k$ by \Cref{lem:tildeFor}, it follows that $\TBFS(u_c, \lca(\tilde{r}_i, \tilde{r}_i')) \setminus \lca(\tilde{r}_i, \tilde{r}_i')$ is contained in $\tilde{H}$ and so $\{u, v\}$ is contained in $\tilde{H}$, contradicting our assumption that it is in a dangling tree.
    
    It follows that if a cycle uses two edges $\{w, u\}$ and $\{u, v\}$ and $\{u, v\}$ lies in a dangling tree then it must be the case that both edges are child edges of $u$, i.e. $w,v \prec u$.
    
    On the other hand, notice that if we have a cycle which uses $\{w, u\}$ and $\{u, v\}$ then there must be some cross edge between $\TBFS(w)$ and $\TBFS(u)$. Such a cross edge will belong to some hammock $\tilde{H}_i$ for which either $w$ or $v$ is $\tilde{r}_i'$ and so either $\{w, u\}$ or $\{u, v\}$ will be in $E_p$, contradicting the fact that our cycle cannot use any edges of $E_p$. Thus, no cycle uses an edge of a dangling tree and so $\mcH$ is a forest of hammocks.
    
    Next, observe that by construction every connected component of $G[\mcH]$ will contain exactly one hammock $\tilde{H}_k$ of $\tilde{\mcH}$ which was designated a root in $\tilde{\mcH}$ where $\tilde{r}_k = r_k \in T_0$ and so $\mcH$ is a \emph{rooted} forest of hammocks (i.e.\ it assigns exactly one root per connected component of $G[\mcH]$).
    
    It remains to argue that $\mcH$ is lca-respecting with base tree $T_0$. In particular, we must show if $H_j$ is the child of $H_i$ then $V(H_i) \cap V(H_j) = \{r_j\}$ and the parent of $r_j'$ in $\TBFS$ is $\lca(r_j, r_j')$ and $\lca(r_j, r_j') \in V(H_i)$ or $\lca(r_j, r_j') = \lca(r_i, r_i')$. Notice that if $\tilde{H}_j$ is the child of $\tilde{H}_i$ in $\tilde{\mcH}$ then we know that $H_j$ will be the child of $H_i$ in $\mcH$ and so this is immediate from \Cref{lem:tildeFor}. On the other hand, if $H_j$ is the child of $H_i$ but $\tilde{H}_j$ was not the child of $\tilde{H}_i$ then it must be the case that there is some dangling tree $T$ which connects $\tilde{r}_j$ to $\tilde{H}_i$ which was added to $\tilde{H}_i$. It follows that $\tilde{H}_j$ was a root hammock in $\tilde{\mcH}$ and so by \Cref{lem:tildeFor} we know that $V(H_i) \cap V(H_j) = \{r_j\}$ and the parent of $r_j' = \tilde{r}_j'$ in $\TBFS$ is $\lca(r_j, r_j')$ and $\lca(r_j, r_j') \in V(H_i)$. Additionally we must show that for each root hammock $H_k$ we have $r_k \in V(T_0)$ and the parent of $r_k$ and $r_k'$ in $\TBFS$ is $\lca(r_k, r_k')$ and $\lca(r_k, r_k') \in V(T_0)$. The former is immediate by how we choose our roots and the latter follows from \Cref{lem:tildeFor}.
\end{proof}

%
%
%
%
%
%
%
%
%

Concluding, we have our hammock decomposition.
\begin{lemma}
    $(T_0, \mcH, E_p)$ is a hammock decomposition which can be computed in deterministic poly-time.
\end{lemma}
\begin{proof}
    First, notice that $(T_0, \mcH, E_p)$ indeed partitions all edges by construction. Moreover, by \Cref{lem:finalFH} we know that $\mcH$ is an lca-respecting rooted forest of hammocks with base tree $T_0$. $G[\mcH]$ contains all shortest cross edge paths since $G[\bar{\mcH}]$ contains all shortest cross edge paths by \Cref{lem:rootedHDForest} and $G[\bar{\mcH}]$ is a subgraph of $G[\mcH]$.
    
    Lastly, we note that the above hammock decomposition is easily computable in deterministic poly-time. In particular $\hat{\mcH}$ is trivial to compute, $\bar{\mcH}$ can be computed by applying one valid assignment, seeing which is the hammock guaranteed to exist by \Cref{lem:someParent} in $I(T)$ for each $T$ in $\mcPHJ \setminus \hat{H}$ and then assigning $T$ to this hammock. $\tilde{\mcH}$ is trivially computable from $\bar{\mcH}$ and $\mcH$ is trivially computable from $\tilde{\mcH}$.
\end{proof}
The above lemma immediately gives \Cref{thm:hamDecPC}.

\section{Scattering Chops via Hammock Decompositions}\label{sec:scatChopsWithHDs}
In this section, we prove \Cref{thm:main}. In particular, we show that every series-parallel graph is $O(1)$-scatter-choppable (\Cref{dfn:scatChop}) which by \Cref{lem:scatterChop} demonstrates that every series-parallel graph is $O(1)$-scatterable (in polynomial time). Combining this fact with \Cref{thm:SPRReduction} will give our SPR solution. We will demonstrate that every series-parallel graph is $O(1)$-scatter-choppable by using our hammock decompositions and, in particular, \Cref{thm:hamDecPC}.

The concrete lemma we will prove in this section is as follows. Clearly, the lemma below combined with \Cref{lem:scatterChop} proves \Cref{thm:main}.

\begin{lemma}\label{lem:sec9-scattering}
Consider a graph $G$ with a hammock decomposition $\{ H_j \}$. Then, there exists a $c$-fuzzy, $O(1)$-scattering, $\Delta$-chop of $G$ with respect ot an arbitrary root $r\in V$.
\end{lemma}

\begin{figure}\centering
\includegraphics[scale=.8]{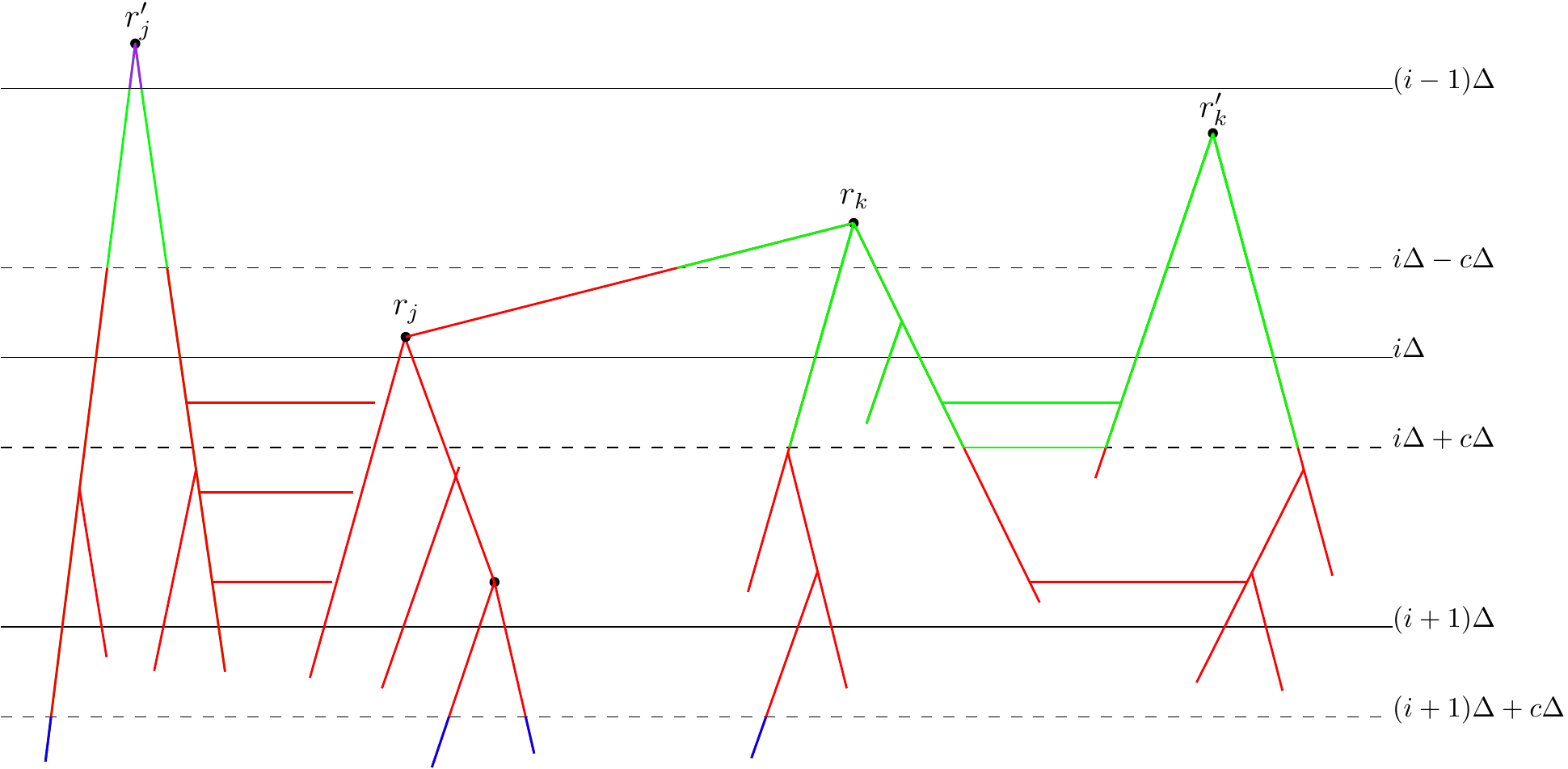}
\caption{Construction of a $c$-fuzzy, $O(1)$-scattering, $\Delta$-chop. There are two hammocks in the picture, $H_j$ and $H_k$.}
\label{fig:chops1}
\end{figure}

Our $c$-fuzzy, $O(1)$-scattering, $\Delta$-chop construction is as follows; see also Figure~\ref{fig:chops1}.

\begin{enumerate}
\item Compute the $\Delta$-chops $A_i=\{v\in V:(i-1)\Delta\le d(r,v)<i\Delta\}$, and initialize $A'_i=A_i$ for all $i$
\item For each hammock $H_j$, independently in parallel:
  \begin{enumerate}
  \item For each $A_i$ with $V(H_j)\cap A_i\ne\emptyset$, independently in parallel:
    \begin{enumerate}
    \item If $r_j\in A_i$ and $d(r,r_j)\ge i\Delta- c\Delta$, then move all of $\{v\in V(H_j')\cap A_i : i\Delta-c\Delta\le d(r,v)<i\Delta\}$ from $A'_i$ to $A'_{i+1}$
    \item If $d(r,r_j)<i\Delta-c\Delta$, then move all $\{v\in V(H_j')\cap A_i:d(r,v)\le (i-1)\Delta+c\Delta\}$ from $A'_i$ to $A'_{i-1}$
    \end{enumerate}
  \end{enumerate}
In other words, $\mathcal A=\{A_i\}$ are the original $\Delta$-chops and $\mathcal A'=\{A_i'\}$ are the sets after the modifications above. All steps can be performed independently in parallel because $c=1/3$ means that there is no interference between different $A_i$.
  \item For the base tree $T_0$, treat it as a subgraph $H_j$ with root $r_j=r$ (the root of $T_{BFS}$), and apply step~2(a) to it.
\end{enumerate}

We begin with the following auxiliary lemmas:
\begin{lemma}\label{lem:single-hammock}
Let $G$ be a series-parallel graph, and consider a (unique) shortest $u$--$v$ path in $G$ that is contained in hammock $H_i$. Then, $P$ contains at most $2$ cross edges in $H_i$.
\end{lemma}
\begin{figure}\centering
\includegraphics[scale=.6]{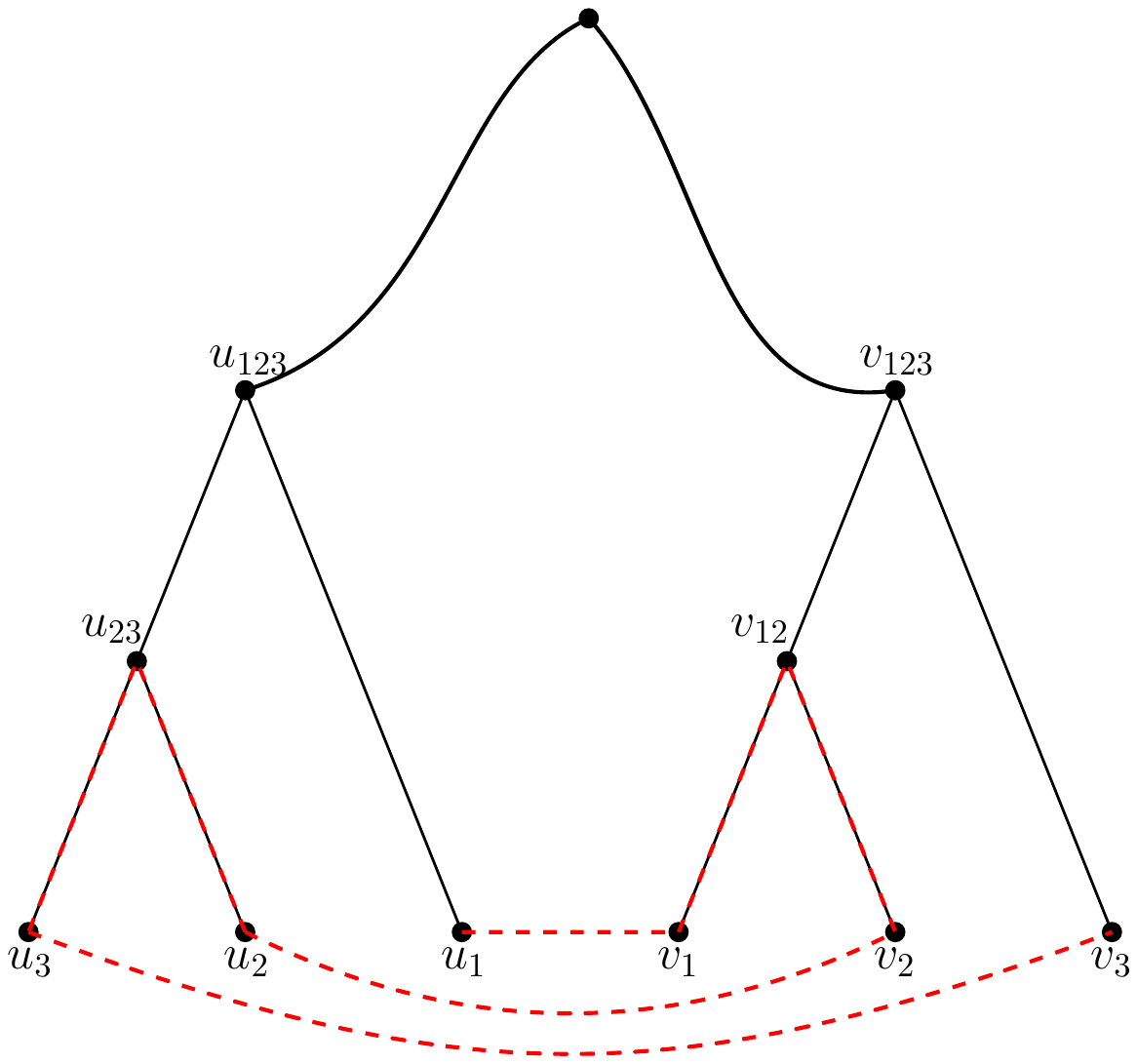} \quad
\includegraphics[scale=.6]{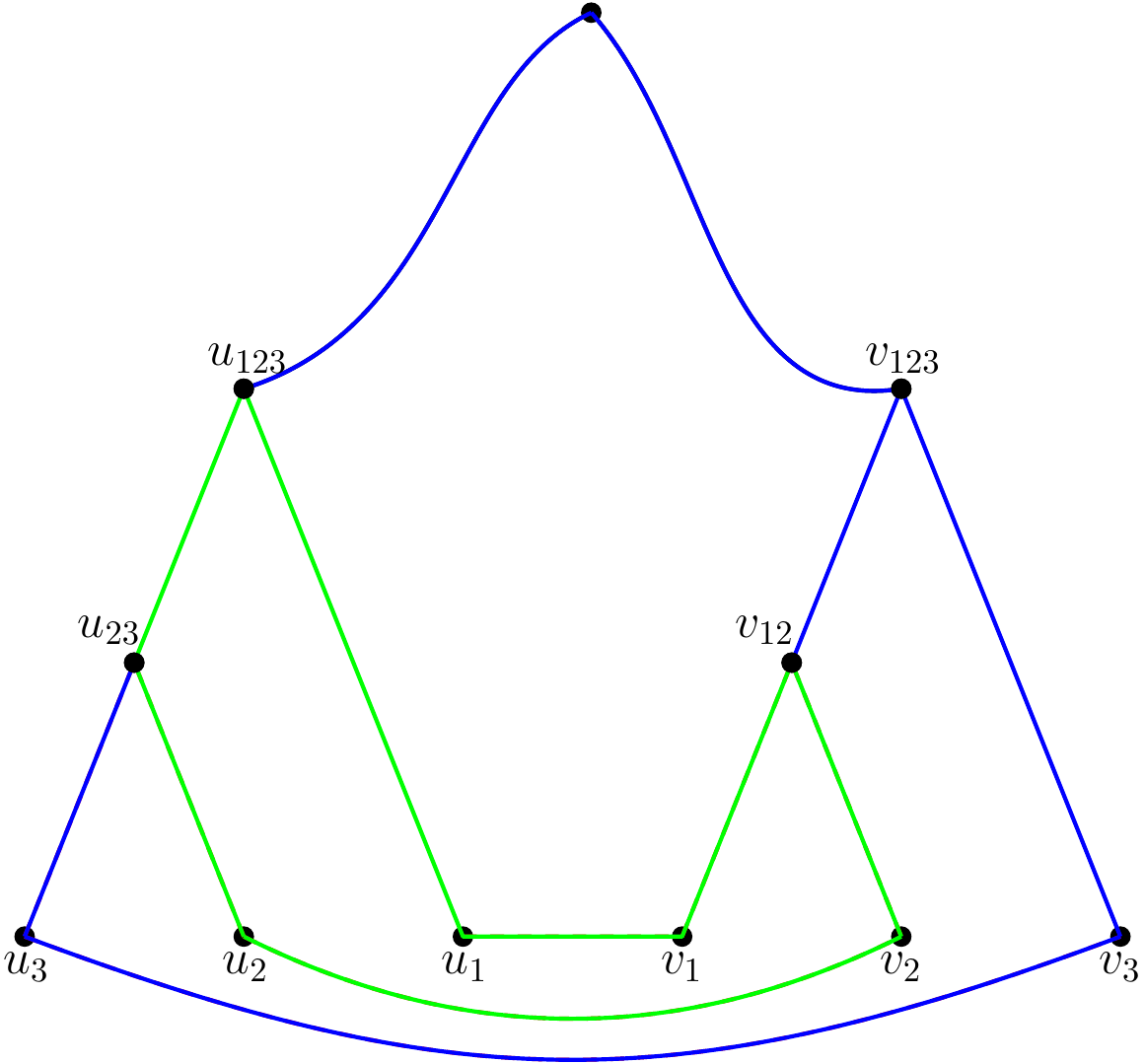}
\caption{Setting for the proof of \Cref{lem:single-hammock}. The path $P$ is the red, dotted path in the left. The green cycle and the blue claw on the right together form a clawed cycle.}
\label{fig:chops2}
\end{figure}
\begin{proof}
Assume for contradiction that $P$ contains at least three cross edges. Then, there exist $u_1,u_2,u_3$ on one tree of the hammock and $v_1,v_2,v_3$ on the other tree such that the path $P$ contains as a subpath the following edges/paths in order (see Figure~\ref{fig:chops2}): the edge $(u,v)$, the $v_1$--$v_2$ path in $T_{BFS}$, the edge $(v_2,u_2)$, the $u_2$--$u_3$ path in $T_{BFS}$, and the edge $(u_3,v_3)$. We first claim that it is impossible for the following two conditions to hold simultaneously, since they would imply that $G$ contains a $K_4$ minor:
 \begin{enumerate}
 \item $\lca(v_1,v_2)$ is a descendant of $\lca(v_1,v_2,v_3)$ (the least common ancestor of $v_1,v_2,v_3$), and
 \item $\lca(u_2,u_3)$ is a descendant of $\lca(u_1,u_2,u_3)$.
 \end{enumerate}
Suppose for contradiction that both of these conditions hold. For ease of notation, define $u_{23}=\lca(u_2,u_3),u_{123}=\lca(u_1,u_2,u_3),v_{12}=\lca(v_1,v_2),v_{123}=\lca(v_1,v_2,v_3)$. Then, we can find a clawed cycle as follows (see Figure~\ref{fig:chops2}). Take the cycle composed of edge $(u_1,v_1)$, the $v_1$--$v_2$ path in $T_{BFS}$, the edge $(v_2,u_2)$, and the $u_2$--$u_1$ path in $T_{BFS}$, which contains $u_{23}$ and $u_{123}$ as distinct vertices. Take the claw rooted at $v_{123}$ with following branches: the $v_{123}$--$u_{123}$ path in $T_{BFS}$, the $v_{123}$--$v_{12}$ path in $T_{BFS}$, and the path composed of the $v_{123}$--$v_3$ path in $T_{BFS}$, the edge $(v_3,u_3)$, and the $u_3$--$u_{23}$ path in $T_{BFS}$. This clawed cycle identifies a $K_4$, contradicting the assumption that $G$ is series-parallel.

Therefore, one of the two conditions is false. Suppose first that (1) is false, which means that $v_{12}=v_{123}$. Replace the segment of $P$ from $v_{12}$ to $v_3$ with the $v_{123}$--$v_3$ path in $T_{BFS}$, which is a shortest $v_{123}$--$v_3$ path, so the length of the path cannot increase after the replacement. It follows that the replacement path is also a shortest $u$--$v$ path, contradicting the assumption that the shortest path is unique.

If (2)\ is false, then we can similarly replace the segment of $P$ from $u_{23}$ to $u_3$ with the $u_{123}$--$u_3$ path in $T_{BFS}$,
\end{proof}

\begin{lemma}\label{lem:monotone}
For each monotone path $P\in T_{BFS}$ of length at most $c\Delta$, the path has $O(1)$ edges cut by $\mathcal A'$.
\end{lemma}
\begin{proof}
First, suppose that the path $P$ is disjoint from the base tree $T_0$. We handle the base tree $T_0$ at the end of the proof.

Suppose that for some $i$, the path $P$ has an edge $(u,v)$ cut by $\mathcal A'$ that satisfies $d(r,u),d(r,v)\in[i\Delta,i\Delta+c\Delta]$. The only way for  this to happen is if $u$ and $v$ belong to two different hammocks $H_j,H_\ell$, and step~(a)ii.\ was applied to one hammock (say $H_j$) but not the other. In that case, we have $r_j\notin A_j$, so $d(r,r_j)<i\Delta$. Let $v'$ be the first vertex on $P$ after exiting $H_j$. Then, we must have $d(r,v')<d(r,r_j)$ regardless of whether $P$ exits $H_j$ via $r_j$ or via $\lca(r_j,r_j')$. It follows that $d(r,v')<d(r,r_j)<i\Delta$, and since $P$ has length at most $c\Delta$ for $c=1/3$, it will never again visit a vertex whose distance from $r$ is in $[i'\Delta,i'\Delta+c\Delta]$ for some integer $i'$.

Suppose next that $P$ has an edge $(u,v)$ cut by $\mathcal A'$ that satisfies $d(r,u),d(r,v)\in[i\Delta-c\Delta,i\Delta)$. The only way for this to happen is if $u$ and $v$ belong to two different hammocks $H_j,H_\ell$, and step~(a)i.\ was applied to one hammock (say $H_j$) but not the other.  In that case, we have $d(r,r_\ell)<i\Delta-c\Delta$. By the same argument from before, $P$ must exit $H_\ell$ at some vertex $v'$ satisfying $d(r,v')<d(r,r_\ell)<i\Delta-c\Delta$, and it will never again visit a vertex whose distance from $r$ is in $[i'\Delta-c\Delta,i'\Delta)$ for some integer $i'$.

Finally, it is easy to see by steps~(a)i.~and~(a)ii. that $P$ cannot have any edge $(u,v)$ cut by $\mathcal A'$ that satisfy $d(r,u),d(r,v)\in(i\Delta+c\Delta,(i+1)\Delta-c\Delta)$.

It follows that $P$ cuts at most one edge in each of the ranges $\bigcup_{i\in\mathbbm Z}[i\Delta,i\Delta+c\Delta]$, $\bigcup_{i\in\mathbbm Z}[i\Delta-c\Delta,i\Delta)$, and $\bigcup_{i\in\mathbbm Z}(i\Delta+c\Delta,(i+1)\Delta-c\Delta)$, and at most $O(1)$ edges in between the ranges.

Finally, we consider the case when $P$ intersects the base tree $T_0$. Let $P_0=P\cap T_0$ be the subpath of $P$ inside $T_0$, and let $P'=P-P_0$ be the remainder of $P_0$, the latter of which has $O(1)$ edges cut by $\mathcal A'$ from before.  Since $T_0$ is a tree, $P_0$ cuts at most one edge in $T_0$ by construction, so in total, $P=P_0\cup P'$ has $O(1)$ edges cut by $\mathcal A'$.
\end{proof}

\begin{lemma}\label{lem:cross edge path}
For each shortest cross edge path $P$ of length less than $c\Delta$, the path has at most $O(1)$ edges cut by $\mathcal A'$.
\end{lemma}
\begin{proof}
By property~(1) of \Cref{dfn:PCHD}, the shortest cross edge path $P$ is contained in $G[\mathcal H]$, where $\mathcal H$ is the set of hammocks.
Let $H_m$ be the hammock of lowest depth (in $\mathcal T$) that shares an edge with $P_m$. Remove the subpath $P_m\cap H_m$ from $P_m$, which splits it further into two paths $P_1,P_2$, each of which is monotone with respect to the tree structure $\mathcal T$ in the following sense: the set of hammocks that share edges with $P_1$ form a monotone path in the tree $\mathcal T$ of hammocks, and the same for $P_2$. By \Cref{lem:single-hammock}, $P_m\cap H_m$ is cut at most twice by $\mathcal A'$. It suffices to argue that $P_1$ is cut $O(1)$ times by $\mathcal A'$; the argument for $P_2$ will be symmetric.

Let $(x,y)$ be the first edge on $P_1$ that is cut by $\mathcal A'$. (If none exist, then we are done.) Suppose that $x\in A'_{i+1}$ and $y\in A'_{i}$, and let $H_j$ be a hammock containing $y$. Our goal is to show that all edges on $P_1$ that are cut by $\mathcal A'$ must belong to $H_j$. Assuming this, we are done by \Cref{lem:single-hammock}.


We have two cases:
\begin{enumerate}
\item $x=r_\ell$ is an endpoint for a child hammock $H_\ell$ of $H_j$. In this case, regardless of whether step~(a)i.\ is applied, we must have $i\Delta-c\Delta\le d(r,r_\ell)\le i\Delta$. We first claim that $P_1$ cannot cross between $A'_{i}$ and $A'_{i-1}$. Suppose otherwise that it contains an edge $(u,v)$ with $u\in A'_{i}$ and $v\in A'_{i-1}$. Then, we must have $d(r,v)\le(i-1)\Delta+c\Delta$ (regardless if step (a)ii.\ was applied). So the path $P_1$ has length at least $d(r,r_\ell)-d(r,v)\ge\Delta-2c\Delta = c\Delta$, contradicting the assumption that $P_1$ has length less than $c\Delta$.

We now claim that $P_1$ cannot cross between $A'_i$ and $A'_{i+1}$ in hammock $H_j$ or beyond. In this case, we must have $d(r,r_j)<i\Delta-c\Delta$, since otherwise step~(a)i.\ would add all vertices in $V(H_j)\cap A_i$ to $A_{i+1}'$. Therefore, in order to cross between $A'_i$ and $A'_{i+1}$ in $H_j$, the path $P_1$ must reach a vertex $v$ with $d(r,v)>i\Delta+c\Delta$ since step~(a)ii.\ was applied to $A_i$. But then the path $P_1$ has length at least $d(r,v)-d(r,r_\ell)\ge c\Delta$, contradicting the assumption that $P_1$ has length less than $c\Delta$.
\item Hammock $H_j$ contains both $x,y$. In this case, we must have $d(r,r_j)<i\Delta-c\Delta$, since otherwise step~(a)i.\ would add all vertices in $V(H_j)\cap A_i$ to $A_{i+1}'$. The argument that $P_1$ cannot cross between $A'_i$ and $A'_{i-1}$ is identical. We now claim that $P_1$ cannot cross between $A'_i$ and $A'_{i+1}$ in any hammock after $H_j$. We have $d(r,x)\ge i\Delta$ (regardless if step~(a)ii.\ was applied), and in order for $P_1$ to travel to any hammock after $H_j$, it must reach either $r_j$ or $\lca(r_j,r_j')$. But then it has length at least $d(r,x)-d(r,r_j)\ge\Delta$, a contradiction.
\end{enumerate}
\end{proof}

With the above lemmas in hand, we finally prove \Cref{lem:sec9-scattering}.
Consider two vertices $u,v$ of distance at most $\Delta$ in $G$. We want to show that the (unique) shortest $u$--$v$ path $P_0$ has $O(1)$ edges cut by $\mathcal A'$. Partition $P_0$ into $O(1)$ many subpaths of length less than $c\Delta$ each. We will transform each subpath into a different shortest path between the two endpoints, and which has $O(1)$ edges cut by $\mathcal A'$.

We further subdivide path $P$ as follows. First, if $P$ contains no cross edges, then it consists of at most two monotone paths, so by \Cref{lem:monotone}, it is cut $O(1)$ times by $\mathcal A'$. Therefore, for the rest of the proof, we assume that $P$ has at least one cross edge. Let $P_l$ be the subpath of $P$ from the beginning to just before the first cross edge of $P$, and let $P_r$ be the subpath that begins just after the last cross edge of $P$ and continues until the end. Let $P_m$ be the remaining subpath after removing $P_l,P_r$, so that $P_l,P_m,P_r$ partition the edges of $P$. The subpaths $P_l,P_r$ consist of at most two monotone paths in $T_{BFS}$, so by \Cref{lem:monotone}, they are each cut $O(1)$ times by $\mathcal A'$. The subpath $P_m$ is a cross edge path, so by \Cref{lem:cross edge path}, it is cut $O(1)$ times.

%

\section{Future Work}
We conclude with directions for future work. First, we reiterate the main open question in this area as previously stated by many prior works:
\begin{enumerate}
    \item For fixed $h$ does every $K_h$-minor-free graph admit an $O(1)$-SPR solution?
\end{enumerate}
The most exciting next step in settling this long-standing open question would be to tackle the planar case. In particular, if one could demonstrate the existence of a forest-like structure similar to hammock decompositions for planar graphs then the techniques introduced in this work would solve the planar case. Notably, it is not too hard to see that the reduction of \citet{filtser2020scattering} of $O(1)$-SPR to $O(1)$-scattering partitions works even if every shortest path is only approximately preserved by the scattering partition. In particular, the reduction of \citet{filtser2020scattering} still works if one can only provide a low-diameter partition where for any vertices $u$ and $v$ there is \emph{some} path between $u$ and $v$ with length at most $O(1) \cdot d_G(u,v)$ which is incident to at most $O(1)$ parts. Consequently, the arguments presented in this paper show that one need only demonstrate the existence of a forest-like structure which \emph{approximates} the distances between cross edges up to a constant to make use of the scattering chops introduced in this work. Thus, an extremely promising next avenue would be to show that every planar graph has a forest-like subgraph which approximately preserves the distances between all cross edges.

\bibliographystyle{plainnat}
\bibliography{abb,main}

\appendix

\section{Metric Nested Ear Decompositions from Hammock Decompositions}\label{sec:metricNestedEar}

A classic result of \citet{eppstein1992parallel} shows that a graph is a 2-vertex-connected series-parallel graph if it has a ``nested ear decomposition.'' These nested ear decompositions need not be unique and, in general, may have little to do with the metric induced by the series-parallel graph. In this section we show how to apply our hammock decompositions to find, among all possible nested ear decompositions, a nested ear decomposition with strong properties regarding how the metric and the nested ear decomposition interact.

We now give a series of definitions and theorems which formalize the nested ear decompositions of \citet{khuller1989ear} and \citet{eppstein1992parallel}.
\begin{definition}[Open Ear Decompositions]
    An ear is a path whose endpoints may coincide. An ear decomposition is a partition of the edges of a graph into (the edges of) ears $P_1, P_2, \ldots$ where for each $i$ $\intV(P_i)$ is disjoint from all vertices in $P_1, \ldots, P_{i-1}$ and the two endpoints of $P_i$ are contained among the vertices of $P_1, \ldots, P_{i-1}$. An ear decomposition is open if the two endpoints of each $P_i$ for $i \geq 2$ are distinct.
\end{definition}
\citet{khuller1989ear} introduced the notion of a tree ear decomposition and \citet{eppstein1992parallel} strengthened this to the idea of a nested ear decomposition.
\begin{definition}[Tree, Nested Ear Decomposition]
    A tree ear decomposition is an open ear decomposition where for $i > 1$ we have that $P_i$ has both endpoints in the same $P_j$. A nested ear decomposition is a tree ear decomposition where for each $P_j$ the collection of ears with both endpoints in $P_i$ form a collection of nested intervals.
\end{definition}

Recall that a graph is 2-vertex connected if the deletion of any one vertex still leaves the graph connected. 
\begin{theorem}[\citet{eppstein1992parallel}]
    A 2-vertex-connected graph is series-parallel if and only if it has a nested ear decomposition.
\end{theorem}

Concluding, we apply our hammock decompositions to strengthen the result of \citet{eppstein1992parallel} to respect the input metric in the following way. Recall that a shortest cross edge path is a path that starts and ends with a cross edge and is also a shortest path.
\begin{theorem}
    Let $G = (V, E)$ be a 2-vertex-connected series-parallel graph with unit weights and unique shortest path lengths. Fix a root $r \in V$ and a BFS tree $\TBFS$ with cross edges $E_c := E \setminus E(\TBFS)$. Then there is a collection of edges $E_p \subseteq E(\TBFS)$ such that $G$ has a nested ear decomposition $P_1, P_2, \ldots$ where
    \begin{enumerate}
        \item Each $P_i$ is of the form $P_u \oplus \{u, v\} \oplus P_v$ where $P_u$ and $P_v$ are monotone paths from $u$ and $v$ towards $r$ and $\{u,v\} \in E_c$.
        \item $|E(P_i) \cap E_p| \leq 1$ and $(V, E \setminus E_p)$ contains all shortest cross edge paths. Furthermore after deleting all ears which contain an edge of $E_p$ any two ears in the same connected component contain cross edges with the same lca.
    \end{enumerate}
\end{theorem}
\begin{proof}
    Let $(T_0, \mcH, E_p)$ be the hammock decomposition (\Cref{dfn:PCHD}) as guaranteed by \Cref{thm:hamDecPC}.
    
    We claim that since $G$ is 2-vertex-connected it must be the case that $T_0$ is a star with center $r$ and $\mcH$ consists of a single tree of hammocks where $r_k$ and $r_k'$ are children of $r$. To see why, notice that since $G$ is $2$-vertex connected, for any children $u$ and $v$ of $r$ there must be some cross edge with one endpoint in $\TBFS(u)$ and another in $\TBFS(v)$ (since otherwise the deletion of $r$ would separate the graph into 2 connected components). It follows that such a cross edge will belong to a tree of hammocks which contains both $u$ and $v$ as hammock roots. Applying this to all pairs of children of $r$ shows that all children of $r$ must belong in the same hammock, thereby showing $T_0$ is a star and $\mcH$ consists of a single tree of hammocks.
    
    Continuing, we can construct the claimed nested ear decomposition as follows. We let $D$ be our nested ear decomposition so far. We will process hammocks in $\mcH$ in a BFS order according to the parent-child relationships induced between hammocks in $\mcH$. To process a single hammock $H_i$ we simply find a candidate cross edge $\{u, v\} = e \in D \setminus E(H_i) \cap E_c$ and add to our nested ear decomposition $D$ the ear consisting of the concatenation of $P_u$, $P_v$ and $e$ where $P_u$ is the path gotten by going from $u$ towards $v$ until hitting a vertex in $D$ and $P_v$ is defined symmetrically. We always choose as our cross edge $\{u,v\}$ an edge such that there does not exist another cross edge $\{u', v'\}$ in $E(H_i)$ with either $u \prec u'$ or $v \prec v'$.
    
    It is easy to verify that if no candidate cross edge with the above properties exists then there exists a clawed cycle.
    
    The above construction satisfies properties (1) and (2) in the above theorem statement by the properties of our hammock decompositions and so it remains only to argue that the result is a nested ear decomposition.
    
    First, we observe that the result of the above process will indeed partition all edges. Every cross edge is included by construction. An edge in $\TBFS$ will be added to $D$ the first time any cross edge with an endpoint below it is processed. Since our graph is 2-vertex connected every edge in $\TBFS$ has some cross edge with an endpoint below it and so every edge in $\TBFS$ will end up in $D$
    
    Our ear decomposition will be open by definition of our hammock decompositions and, in particular, by the fact that the roots of a hammock are unrelated and by the fact that $\mcH$ is lca-respecting.
    
    Next, we verify that the result of this process is a tree ear decomposition. Suppose for the sake of contradiction that the result of our construction is not a tree ear decomposition. In particular, suppose $P_l$ is an ear with one endpoint in $P_i$ and another endpoint in $P_j$. By definition of our construction it must be the case that $P_i$ and $P_j$ are derived from cross edges, say $e_i$ and $e_j$, which belong to the same hammock. However, it then follows that the hammock fundamental cycle formed by $e_i$ and $e_j$ can be used to construct a clawed cycle, a contradiction. A similar argument shows that if our ear decomposition fails to be nested then we can find a clawed cycle.
\end{proof}

\section{Hammock Decomposition Construction Figures}
In this section we give all of the illustrations of the construction of our hammock decompositions on a single graph in \Cref{fig:allConstructionsAtOnce}.
\begin{figure}
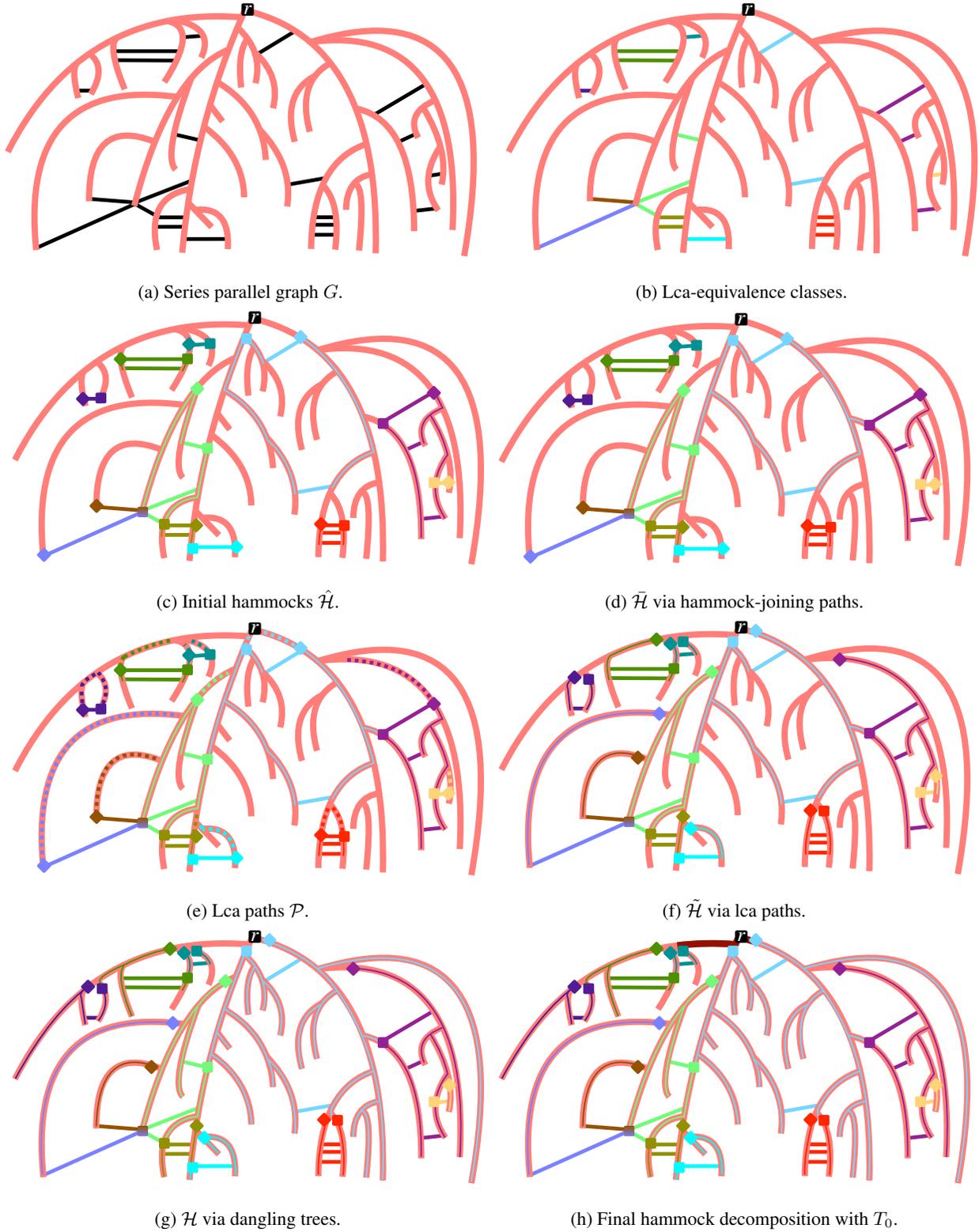

    \centering
    \begin{subfigure}[b]{0.49\textwidth}
        \centering
        \includegraphics[width=\textwidth,trim=0mm 0mm 0mm 0mm, clip]{./figures/HDCon1.pdf}
        \caption{Series parallel graph $G$.}
    \end{subfigure}
    \hfill
    \begin{subfigure}[b]{0.49\textwidth}
        \centering
        \includegraphics[width=\textwidth,trim=0mm 0mm 0mm 0mm, clip]{./figures/HDCon2.pdf}
        \caption{Lca-equivalence classes.}
    \end{subfigure}
    \begin{subfigure}[b]{0.49\textwidth}
    \centering
    \includegraphics[width=\textwidth,trim=0mm 0mm 0mm 0mm, clip]{./figures/HDCon3.pdf}
    \caption{Initial hammocks $\hat{\mcH}$.}
    \end{subfigure}
    \begin{subfigure}[b]{0.49\textwidth}
        \centering
        \includegraphics[width=\textwidth,trim=0mm 0mm 0mm 0mm, clip]{./figures/HDCon4.pdf}
        \caption{$\bar{\mcH}$ via hammock-joining paths.}
    \end{subfigure}
    \begin{subfigure}[b]{0.49\textwidth}
        \centering
        \includegraphics[width=\textwidth,trim=0mm 0mm 0mm 0mm, clip]{./figures/HDCon5.pdf}
        \caption{Lca paths $\mcP$.}
    \end{subfigure}
    \begin{subfigure}[b]{0.49\textwidth}
        \centering
        \includegraphics[width=\textwidth,trim=0mm 0mm 0mm 0mm, clip]{./figures/HDCon6.pdf}
        \caption{$\tilde{\mcH}$ via lca paths.}
    \end{subfigure}
    \begin{subfigure}[b]{0.49\textwidth}
        \centering
        \includegraphics[width=\textwidth,trim=0mm 0mm 0mm 0mm, clip]{./figures/HDCon7.pdf}
        \caption{$\mcH$ via dangling trees.}
    \end{subfigure}
        \begin{subfigure}[b]{0.49\textwidth}
        \centering
        \includegraphics[width=\textwidth,trim=0mm 0mm 0mm 0mm, clip]{./figures/HDCon8.pdf}
        \caption{Final hammock decomposition with $T_0$.}
    \end{subfigure}
    \caption{An illustration of the construction of our hammock decomposition on a series-parallel graph.}\label{fig:allConstructionsAtOnce}
\end{figure}

\end{document}